\documentclass[
pdflatex, 
]{sn-jnl}

\usepackage{amsmath,amsfonts}
\usepackage{array}
\usepackage[caption=false,labelformat=empty,font=normalsize,textfont=sf]{subfig}
\usepackage{textcomp}
\usepackage{stfloats}
\usepackage{url}
\usepackage{verbatim}
\usepackage{graphicx}
\usepackage{cite}
\usepackage[title]{appendix}%

\usepackage{mathrsfs}
\usepackage[shortlabels]{enumitem} 
\usepackage{amssymb}
\usepackage{bm}
\usepackage{xcolor}
\usepackage{listings}
\usepackage{xfrac}
\usepackage{comment}
\usepackage[utf8]{inputenc}
\DeclareUnicodeCharacter{03BA}{\ensuremath{\kappa}}
\usepackage[T1]{fontenc}
\usepackage{bigints} 
\usepackage{orcidlink}
\usepackage{cancel}
\usepackage{multirow}
\usepackage{lmodern}
\usepackage[section]{placeins} 
\usepackage{csquotes}
\usepackage{manyfoot}
\usepackage{booktabs}
\usepackage{makecell}

\usepackage{etoolbox} 

\usepackage{hyperref}
\hypersetup{
    colorlinks=true,
    linkcolor=blue,
    citecolor=blue,
    filecolor=blue,
    urlcolor=blue,
}
\usepackage[capitalize]{cleveref}

\theoremstyle{thmstyleone}
\newtheorem{theorem}{Theorem}
\newtheorem{lemma}[theorem]{Lemma}
\newtheorem{phyAxioms}[theorem]{Physical Axioms}

\theoremstyle{thmstyletwo}
\newtheorem{definition}{Definition}

\newtheorem{remark}{Remark}

\theoremstyle{thmstylethree}

\newcommand{\coupledsum}[3]{\sum_{#1}^{#2}{\!}_{\raisebox{-0.3ex}{\scalebox{1.5}{$\scriptscriptstyle#3$}}}}
\newcommand{\coupledprod}[3]{\prod_{#1}^{#2}{\!}_{\raisebox{-0.3ex}{\scalebox{1.5}{$\scriptscriptstyle#3$}}}}

\usepackage{abstract} 
\renewcommand{\abstractname}{\textbf{Abstract}} 
\makeatletter
\renewcommand{\thesection}{}          
\renewcommand{\thesubsection}{}       
\renewcommand{\thesubsubsection}{}    
\makeatother

\begin{document}

\title[Coupled Entropy]{A unique calibration of entropy for all complex universality classes}

\author[1]{\fnm{Kenric P.} \sur{Nelson} \orcidlink{0000-0001-6962-7459}}\email{kenric.nelson@photrek.io}

\affil[1]{\orgname{Photrek, Inc.}, \orgaddress{\street{56 Burnham St \#1}, \city{Watertown}, \postcode{02472}, \state{MA}, \country{USA}}}

\maketitle

\begin{abstract} A unique calibration of entropy requiring the nonlinear degrees of freedom to be zero is proven to be comprehensive across all complex universality classes. The Shannon entropy is shown to measure energetic, configurational, and nonlinear degrees of freedom (DOF) for complex systems with non-exponential distributions. A variety of generalized entropies have been proposed, some of which can be tuned to be extensive for a particular nonlinearity. An additional requirement is proposed in which a calibrated growth of the microstates includes the information scale, which is independent of the nonlinearity. The unique calibration of entropy is proven to be the coupled entropy, which has a metric inverse to the calibrated growth of states, and a moment over the independent-equals distribution. The coupled entropy, when constrained by the informational scale, is maximized by the coupled stretched exponential distributions, and for its maximizing distribution measures the number of dimensions plus a generalized logarithm of the partition function, together raised to the inverse power of the stretching parameter. A new set of physical axioms regarding the asymptotic growth of states, the calibration from an informational scale, and nonlinear dynamics provide a foundation for the proof that the calibrated entropy is unique and comprehensive. A consistent thermodynamics for nonequilibrium systems is proposed with the temperature equal to the scale constraint divided by the Boltzmann constant. Physical and informational applications are highlighted.
\end{abstract}

\section*{Introduction}

Life has been shown to reside at the edge of chaos \cite{lewinComplexityLifeEdge1999,marinoMethodologyPerformingGlobal2008, eskovClassificationUncertaintiesModeling2019}, just one of the motivations for precise models of complex adaptive systems \cite{ComplexAdaptiveSystems,yoshidaNonlinearScienceChallenge2010}. Likewise, artificial intelligence systems are pushing the boundaries of safe control \cite{shneidermanHumanCenteredArtificialIntelligence2020,almeidaPredictiveNonlinearModeling2002,benderskayaNonlinearTrendsModern2013}, necessitating robust designs trained against outliers from fluctuating uncertainty \cite{dietterichStepsRobustArtificial2017}. In these domains, the difficulty with quantifying heavy-tailed, non-equilibrium \cite{resnickHeavyTailPhenomenaProbabilistic2007, casas-vazquezTemperatureNonequilibriumStates2003} phenomena and integrating these phenomena into foundational methodologies creates an analytical bottleneck restricting analysis and design.

The Boltzmann-Gibbs-Shannon (BGS) entropy is well-established as the appropriate uncertainty metric for equilibrium systems with linear sources of uncertainty \cite{goldsteinGibbsBoltzmannEntropy2019}. Given moment constraints, the BGS entropy is maximized by members of the exponential family of distributions, which E.T. Jaynes showed provides a broad principle for model development \cite{jaynesInformationTheoryStatistical1957}. Uncertainty quantification of non-equilibrium systems with nonlinear sources of uncertainty presents significantly more difficult challenges. First, the moments depend on both the scale and shape of the distributions, and diverge for large values of the shape parameter. Secondly, the BGS entropy is quickly dominated by the shape of the distribution, breaking the connection between the scale and entropy. The generalized Pareto \cite{arnoldParetoDistribution2015, pickandsStatisticalInferenceUsing1975} and Gosset (Student's t) \cite{pearsonStudentsCollectedPapers1942} distributions are shown to be of critical importance for defining a generalized entropy function, since they each utilize a definition of scale that is independent of the shape of the distribution. 

The scientific investigation of entropy functions for complex systems has, to date, focused on revising the fourth, additive, Shannon-Khinchin axiom \cite{khinchinMathematicalFoundationsInformation1957}. Examples include the Havrda-Charvát function\cite{havrdaQuantificationMethodClassification1967}, popularized by Tsallis \cite{tsallisBoltzmannGibbsShannon2019}, and axiomized by Suyari  \cite{suyariGeneralizationShannonKhinchinAxioms2004}, the Sharma-Taneja's two-parameter function \cite{sharmaAxiomaticCharacterizationInformationtheoretic1974}, Tempesta's group-theory composition \cite{tempestaShannonKhinchinFormulation2016}, and the comprehensive function proven by
Hanel-Thurner using just the first three axioms \cite{hanelComprehensiveClassificationComplex2011}. However, these mathematical justifications disconnected entropy from a specific physical measurement. 

Neither Shannon nor Khinchin referred to their observations regarding the properties of the exponential family and the measurement of information as axioms. Rather, Shannon begins his 1948 treatise on communication \cite{shannonMathematicalTheoryCommunication1948} with the pragmatic statement regarding how a measure of information is defined, ``If the number of messages in a set is finite then this number or a monotonic function of this number can be regarded as a measure of the information produced when one message is chosen from the set, all choices being equal. ... the most natural choice is the logarithmic
function. Although this definition must be generalized considerably when we consider the influence of the
statistics of the message ... One feels, for example, that ... two identical
channels [should have] twice the capacity of one for transmitting information''. 

I propose that it is Shannon's requirement that the number of channels or degrees of freedom of a system that must be central to the definition of entropy in nonlinear, complex systems. Shannon anticipates the need to generalize the definition depending on the statistics of the system. To achieve the goal of maintaining a consistent count of degrees of freedom, I propose a definition of entropy based on the requirement that, given a calibrated growth function for the equipartition microstates $W_c(N)$, where $N$ is degrees of freedom, the entropy is $H(W_c(N))=N$. Three new axioms are proposed which a) specify the asymptotic scaling properties of microstate growth in complex systems proven by Hanel-Thurner, b) specify a calibrated unit of measure based on a unique informational scale, and c) require that the nonlinear differential equation specifying the probability, $p_{W(N)}$ is specified by three constants. From this definition and axioms, I prove that the only function satisfying the requirement to be a calibrated measurement of the degrees of freedom across the comprehensive classification of the complex system is the coupled entropy \cite{nelsonAverageUncertaintySystems2017, nelsonReducedPerplexitySimplified2020}, now scaled to account for linear extensivity with a stretch nonlinearity.
 
There are many nonequilibrium applications where a calibrated measure of entropy will improve analysis. For example, 5G wireless service is approximately 100 Mbps, where the 100 megabits represent the degrees of freedom (DOF) for encoding messages. In an interference-free zone, exponential encoding of messages could enable $2^{10^8}$ unique messages. Log base 2 entropy returns the bits $\log_2 2^{10^8}=10^8$. In contrast, a high-interference environment creates heavy-tailed noise, which would cause the number of messages to grow at a power-law, rather than exponential rate. The proposed entropy incorporates these nonlinearities into the metric and is calibrated to still return the number of bits. Calibrated measurement of information is required across of variety of complex adaptive systems. Another example is the distinction between coding and non-coding DNA \cite{cohenGeneralDesignsReveal2022, adamiInformationTheoryMolecular2004}, in which the four nucleotides form strings of potential messages. Whereas, a direct linear path between coding DNA and proteins is well-established, the complexity of the non-coding DNA and its function in regulatory networks is an active area of research. Across the sciences, non-equilibrium systems such as financial markets \cite{MandelbrotFinancialMarkets}, socio-political networks \cite{huckfeldtNetworksContextSocial1987}, ecosystems \cite{ulanowiczInformationTheoryEcology2001}, quantum mechanics \cite{biroNonextensiveApproachQuark2009}, energy fluctuations \cite{roccaStrategyStoppingFailure}, and many more systems require precise, calibrated nonlinear analysis.

The specification for the unique, universal entropy function starts with the comprehensive classification of complex systems proven by Hanel-Thurner \cite{hanelComprehensiveClassificationComplex2011, hanelWhenGeneralizedEntropies2011}. They showed that the asymptotic growth of systems, and more specifically entropy functions, could be specified by two parameters regarding stretched exponential growth and power-law (or its inverse super-exponential) growth. I show that this classification can be made more precise by quantifying the degree of variable nonlinearity (stretch) and the degree of functional nonlinearity (power-law). From this foundation, I prove that there is a unique unit of measure, which I refer to as the informational scale. Together, the complexity classification and the unit of measure define a calibrated growth of states, which uniquely specifies the coupled entropy, which given the informational scale as a constraint is maximized by the coupled stretched exponentials. This framework is known as \textit{Nonlinear Statistical Coupling} \cite{nelsonNonlinearStatisticalCoupling2010}.

The new contributions reported here include:
\begin{itemize}
    \item  In the Results subsection \nameref{subsec_InfoScale}, a precise definition for the informational scale \eqref{subsec_InfoScale} and proof that it is unique in its quantification of the linear source of uncertainty, and thereby independent of the nonlinearity. The multiplicative stochastic differential equation is used to illustrate the significance of properly separating the linear and nonlinear sources of uncertainty.
    \item In the Results subsection \nameref{subsec_axioms}, the requirement that an entropy function must be calibrated using the informational scale. The calibration process assures that the entropy function with matching nonlinearities measures the correct degrees of freedom. For the equipartition case, the measurement is required to be exclusively the energetic DOF. For the canonical ensemble, the measurement is required to be the energetic DOF and a generalized configurational DOF.  In both cases, the nonlinear DOF are zero.
    \item In the Methods subsection \nameref{subsec_maxProof} proof that the calibrated entropy is comprehensive in covering the full range of the scaling universality classes defined by Hanel-Thurner \cite{hanelComprehensiveClassificationComplex2011}. This assures that, when matched to a particular nonlinear system, entropy is an extensive metric.
    \item In the Results subsection \nameref{subsec_axioms}, the canonization of the Shannon-Khinchin axioms, whereby the compositional axiom is used to justify generalizations, is shown to be inadequate. Instead, physical axioms focusing on the requirements for a calibrated measurement of degrees of freedom are proposed.
    \item In the Results subsection \nameref{subsec_apps}, the structure for a consistent generalization of the thermodynamics is proposed. The proposal includes three novel discoveries: a) the temperature rather than the internal energy is proportional to the constraint, b) the internal energy is a nonlinear function of the temperature, and c) the free energy of a nonlinear system splits into macroscopic and microscopic definitions related by a hyperbolic geometry.
    \item In the Results subsection \nameref{subsec_apps}, applications to information theory, particularly the generalization of variational inference and active inference, are presented.
\end{itemize}

To begin, the methods of nonlinear statistical coupling are introduced via deterministic and stochastic differential equations.

\section*{Results}
\label{sec_Results}

\subsection*{Preliminaries: Nonlinear Dynamics} \label{subsec_prelim}

Complex systems are governed by two distinct sources of nonlinearity. The first is the stretching parameter $\alpha$, which governs short-range dependence and modifies the power of the variable in a linear system, as defined in this ordinary differential equation $y' + x^{\alpha -1}y =0, y(0)=1$, whose solution is an exponential function, $y = \exp (\frac{-x^\alpha}{\alpha})$.  The second is the nonlinear statistical coupling, $\kappa$, which governs long-range dependence and modifies the power of the function in a nonlinear system, as defined in this nonlinear differential equation $y(x)' + y(x)^{1+\kappa} = 0, y(0)=1$, whose solution is a power-law function, $y = (1+\kappa x)^{-\frac{1}{\kappa}}$, which can be expressed as a generalization of the exponential function, $y = \exp_\kappa \left(\frac{x^\alpha}{\alpha}\right)^{-1},\ \exp_\kappa(x)\equiv(1+\kappa x)^\frac{1}{\kappa}$. 

These two nonlinearities, long-term, functional and short-term, variable, $(\kappa,\alpha)$, form a comprehensive classification of complex systems introduced by Hanel and Thurner \cite{hanelComprehensiveClassificationComplex2011}, in which each combination is universal for all systems with that asymptotic scaling property of the microstates. In the next subsection, I will prove that there is a unique scale, $\sigma$, which determines a unit of measurement for the system. This scale will be used to establish the calibration necessary to specify the unique calibration of entropy for the complex universality classes. For the moment, the triplet $(\kappa,\alpha,\sigma)$ forms a foundational nonlinear differential equation, whose solution is the coupled stretched exponential function,
\begin{align}
    y' &+ \frac{x^{\alpha -1}}{\tilde{\sigma}^\alpha } y^{1+\tilde{\kappa}} = 0, 
    \ x= 0, \ y=1
    \quad \tilde{\kappa} = \frac{\alpha\kappa}{\alpha+\kappa}, \quad
    \tilde{\sigma}^\alpha = \frac{\alpha\sigma^\alpha}{\alpha+\kappa}
    \nonumber\\
    y &= \exp_{\kappa}\left(\frac{x^\alpha}{ \alpha \sigma^\alpha}\right)
    ^{-\frac{\alpha+\kappa}{\alpha}}  
    \equiv \left(1+\kappa \frac{x^\alpha}{ \alpha\sigma^\alpha}\right)
    ^{-\frac{\alpha+\kappa}{\alpha\kappa}}.
\end{align}

The coupled stretched exponential and its inverse, the coupled logarithm, have the following nonlinear properties \cite{borgesPossibleDeformedAlgebra2004}
\begin{align}
    \exp_\kappa \left(\frac{x^\alpha}{\alpha}\right)^{-1} 
    &\equiv \left(1+\kappa\frac{x^\alpha}{\alpha}\right)
    ^{-\frac{1}{\kappa}}
    \nonumber\\ 
    \exp_\kappa \left(\frac{x^\alpha}{\alpha}\right)
    \otimes_\kappa \exp_\kappa \left(\frac{x^\alpha}{\alpha}\right) 
    &\equiv \exp_\kappa \left(\frac{x^\alpha}{\alpha} + \frac{x^\alpha}{\alpha}\right) 
    \nonumber\\
    \exp_\kappa^{-1} y &\equiv \ln_\kappa x \equiv 
    \frac{1}{\kappa} \left(x^\kappa - 1\right)
    \nonumber\\
    \ln_\kappa x \oplus_\kappa \ln_\kappa y &= \ln_\kappa xy 
    = \ln_\kappa x + \ln_\kappa y + \kappa (\ln_\kappa x )( \ln_\kappa y )
\end{align}
The derivation of these functions, and other elements of the coupled algebra, is detailed in the Methods subsection \nameref{subsec_CEF}. 

The role of the coupled stretched exponential is particularly important in the definition of entropy for nonlinear, complex systems since it incorporates all the universality, complexity classes. The four principal asymptotes of the distribution are exponential $(e^{-x})$, stretched exponential $(e^{-x^\alpha})$, stretched power-law, $(x^{-(1+\frac{\alpha}{\kappa})}), \kappa>0$, and stretched compact support $(x^{-(1+\frac{\alpha}{\kappa})}), \kappa<0$.

There are a variety of origins of complexity that lead to distributions based on the coupled stretched exponential. The superstatistics origin was reviewed when the coupled entropy was introduced, \cite{nelsonAverageUncertaintySystems2017}. Closely related is the restriction of samples, which motivated the development of the Student's t distribution \cite{pearsonStudentsCollectedPapers1942}. Another important origin is the multiplicative noise process, which is reviewed next to show the distinction separation between the linear $(\sigma)$ and nonlinear $(\kappa,\alpha)$ sources of uncertainty.

The fluctuations of multiplicative noise provide further evidence that the structure of the coupled stretched exponentials provides a unique model of the mathematical physics of complex systems. A multiplicative noise process with a parabolic potential relating the deterministic and multiplicative functions has a non-equilibrium stationary state (NESS) of a coupled Gaussian, $X_{NESS} \sim N_\kappa(X_0,\sigma),$ \cite{al-najafiIndependentApproximatesProvide2026, anteneodoMultiplicativeNoiseMechanism2003}. The Stratonovich stochastic differential equation has a clear expression due to its evaluation at the midpoint of the random variable, denoted with a circle $(\circ)$. The Stratonovich SDE has three components: drift, additive noise, quantified by the scale $\sigma$, and multiplicative noise, quantified by the coupling $\kappa$, 
\begin{equation}\label{equ_multproc}
    dX_t=-X_tdt+ \sqrt{2}\sigma \circ dW_t^{(a)}+ \sqrt{\kappa} X_t \circ dW_t^{(m)}.
\end{equation}
The scale parameter is a measure of the linear source of uncertainty, independent of the nonlinear source of uncertainty. Figure \ref{fig:MultProc} shows several samples of the process for $\sigma=5$ and a set of couplings $\kappa= (0.1, 1, 5)$. While the coupled Gaussian, specified below for $\alpha=2$, uses an independent set of parameters, the $q$-Gaussian of NSM \cite{tsallisBoltzmannGibbsShannon2019, anteneodoMultiplicativeNoiseMechanism2003} has a scale of $\sqrt{\beta^{-1}}=\frac{\sigma}{\sqrt{1+\kappa/2}}$. The inverse dependence on the nonlinearity undermines the ability of $\sqrt{\beta^{-1}}$ to model a generalized scale and physical properties, such as a generalized temperature.
\begin{figure}[ht]
    \centering
    \includegraphics[width=1\linewidth,page=1]{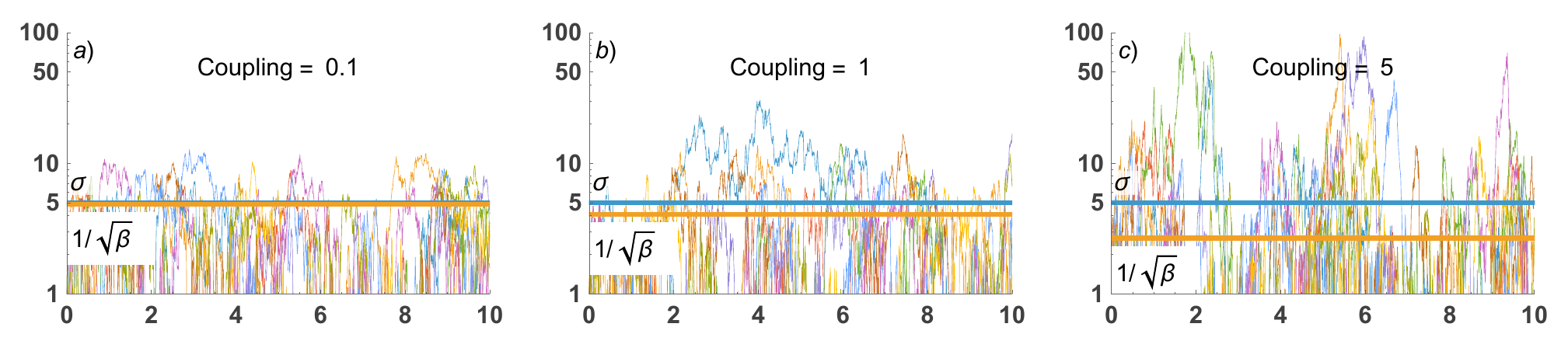}
\caption{\textbf{The increase in fluctuations from a multiplicative noise process.} 10 samples from the multiplicative noise process defined by equation \ref{equ_multproc}, showing the positive values on a logarithmic scale. The NESS distribution is a coupled Gaussian with $\sigma=5$ and a set of couplings $\kappa=(0.1,1,5)$ in subfigures a, b, and c, respectively. As the coupling increases, the fluctuations of the process increase, while the scale $\sigma=5$ is independent of the fluctuations. This independence of the scale is central to calibrating the measurement of entropy and defining a temperature. In contrast, the $q$-Gaussian scale $\sqrt{\beta^{-1}}$  is dependent on the multiplicative noise, which undermines its ability to be a measure of generalized temperature. }
    \label{fig:MultProc}
\end{figure}

A very broad definition of the multiplicative stochastic process has the full family of the coupled exponential family (CEF) as the stationary distribution. The power of two for the coupled Gaussian case generalizes to the stretching nonlinearity, $\alpha$. The balanced drift, $-X_ tdt$, generalizes to form the base function of the CEF. Therefore, the general multiplicative process in the Stratonovich form is
\begin{align}
dX_t = A(X_t)\,dt + \sqrt{\alpha\sigma^{\alpha}}\circ dW_t^{(a)} + \sqrt{\kappa}\,|X_t|^{\alpha/2}\circ dW_t^{(m)},
\qquad \alpha>0,\;\sigma>0,\;\kappa\ge0,
\end{align}
which has the stationary distribution,
\begin{align}
    p(x) &= \frac{h(x)}{Z_\kappa (\alpha,d)}\,\left(1 + \kappa\frac{|x|^{\alpha}}{\alpha\sigma^{\alpha}}\right)^{-1/\kappa - 1/\alpha}
    = \frac{h(x)}{Z_\kappa (\alpha,d)} 
    \exp_\kappa \left(\frac{|x|^{\alpha}}{\alpha\sigma^{\alpha}}\right)
    ^{-\frac{\alpha+\kappa}{\alpha\kappa}}
\end{align}
where $h(x)$ is the base function, and $Z_\kappa(\alpha,d)$ is the partition function, 
\begin{align}
    h(x) &\propto \exp\!\left( \int^x \frac{2A(y) + \bigl(\alpha+\kappa - \frac{\kappa\alpha}{2}\bigr)y^{\alpha-1}}{\alpha\sigma^{\alpha}+\kappa y^{\alpha}}\,dy \right),
    \nonumber \\
    Z_\kappa(\alpha,d) &= \int^x
h(x)\,\left(1 + \kappa\frac{|x|^{\alpha}}{\alpha\sigma^{\alpha}}\right)
^{-\frac{\alpha+\kappa}{\alpha\kappa}}
dx.
\end{align}
The full multivariate definition for the CEF and a proof of its solution for the multiplicative noise process is provided in the Methods subsection \nameref{subsec_CEF}. The importance of this definition for the multiplicative SDE and the coupled exponential family is the ability to define the independent, primary physical properties of the system:
\begin{itemize}
    \item Scale, $\sigma$, the linear source of uncertainty;
    \item Coupling, $\kappa$, the long-range nonlinear source of uncertainty associated with multiplicative noise and nonlinear differential equations; and 
    \item Stretching, $\alpha$, the short-range nonlinear source associated with the stretching of the variable.
\end{itemize}

Next, I will prove that the definition of the scale, $\sigma$, is not a free, arbitrary choice.

\subsection*{Definition, uniqueness, and measurement of the informational scale}\label{subsec_InfoScale}

Specifying the scale of a shape-scale distribution would appear to be an elementary task, but, in fact, the lack of a clear criterion has resulted in different statistical physics communities using different criteria. For instance, NSM \cite{shalizicosmar.TsallisStatisticsStatistical2021} uses a scale derived from the function Tsallis hypothesized for entropy, while the space plasma \cite{pierrardKappaDistributionsTheory2010} community specifies the average velocity as a scale. A proof is provided here that the only scale that separates the linear and nonlinear sources of uncertainty is the value at which the derivative of the log-log of the distribution is negative one. Given the connection to the surprisal of the distribution, this will be referred to as the \textit{informational scale}. 

Recall that a shape-scale distribution can be scaled if given a standard (non-scaled) random variable, $X \sim f(x),$ the scaled random variable is $\sigma X \sim \frac{1}{\sigma}f(\frac{x}{\sigma})$. While this is a criterion for the distribution, it does not specify the definition of the scale, since $\sigma'=a\sigma$ also satisfies the relationship.

\begin{definition}[Informational Scale]
    The informational scale, $\sigma$, of a distribution with shape parameters $(\alpha,\kappa)$ and location, $\mu$, is the value of $x$ such that:
    \begin{align}\label{equ_info_scale}
        \frac{\mathrm{d}\ln f(x-\mu;\sigma,\kappa,\alpha)}{\mathrm{d}\ln x}\vert_{x=\mu+\sigma}=(x-\mu) \frac{f'(x-\mu;\sigma,\kappa,\alpha)}{f(x-\mu;\sigma,\kappa,\alpha)}\vert_{x=\mu+\sigma}=-1
    \end{align}
    Equivalently, using the surprisal rather than the derivative of the log-log, the criterion is:
    \begin{align}
        -\frac{\mathrm{d}\ln f(x-\mu;\sigma,\kappa,\alpha)}{\mathrm{d} x}\vert_{x=\mu+\sigma}=\frac{1}{\sigma}
    \end{align}
\end{definition}

As shown in Table \ref{tab_scale}, the parameterization for the generalized Pareto \cite{arnoldParetoDistribution2015} and the Student's t \cite{pearsonStudentsCollectedPapers1942} distributions satisfy the informational scale, while the NSM $q$-exponential, $q$-Gaussian, and the space plasma kappa distribution do not. The NSM distributions incorrectly define the inverse-scale, thereby resulting in a parameter that is dependent on both the scale and the shape of the distribution. The consequences of defining or not defining the scale to be independent of the shape is shown in Figure \ref{fig_scale}. Likewise, the space plasma definition seeks to represent energy due to the velocity of a plasma, $\frac{2k_B T}{m}$,  ($k_B:$ Boltzmann constant;  $T:$ temperature;  $m:$ mass); however, in doing so, dynamics above $\kappa>\sfrac{3}{2}$ are lost.

\begin{table}[ht]
\centering
\caption{\textbf{The informational scale and BGS entropy of distributions.}
For each distribution, its PDF, informational scale, and Shannon entropy are shown. The purpose is to highlight the importance of the Generalized Pareto and Student's t distributions in isolating the primary physical properties. In contrast, the $q$-exponential, $q$-Gaussian, and Plasma kappa distributions do not use primary properties and thus do not correctly identify the scale. The Shannon entropy of the GPD is one way to identify the primary properties, given its separation of the configurational and nonlinear DOF. }
\label{tab_scale}
\begin{tabular}{ccc>{\centering\arraybackslash}p{0.25\linewidth}}
\toprule
\textbf{Distribution} & \textbf{PDF} & \textbf{Info Scale} & \textbf{BGS Entropy}\\
\midrule
Gen. Pareto & $\frac{1}{\sigma}\left(1+\kappa\frac{x}{\sigma}\right)
^{-\left(\frac{1}{\kappa}+1\right)}$ 
& $\sigma$ & $1+\ln\sigma +\kappa$ \\
$q$-exponential & $(2-q)\beta\left(1+(q-1)\beta x\right)^{\frac{1}{1-q}}$
& $\frac{1}{(2-q)\beta}$ & $1+\ln\frac{1}{(2-q)\beta} +\frac{q-1}{2-q}$ \\
Student's t & $\frac{1}{Z_S}\left(1+\frac{x^2}{\nu\sigma^2}\right)^{-\frac{1}{2}(\nu+1)}$& $\sigma$ & $\ln Z_S +\frac{\nu + 1}{2}\left(\psi(\frac{\nu+1}{2})-\psi(\frac{\nu}{2})\right)$ \\
$q$-Gaussian & $\frac{1}{Z_{qG}}\left(1+(q-1)\beta x^2\right)^{\frac{1}{1-q}}$& $\frac{1}{\sqrt{(3-q)\beta}}$ & $1+\ln Z_{qG} +\frac{1}{q-1}\left(\psi(\frac{1}{q-1})-\psi(\frac{q-1}{2(3-q)})\right)$\\
Plasma $\kappa$-Dist. & $\frac{1}{Z_{\kappa_P}} \left[ 1 + \frac{m x^2}{2k_BT (\kappa - 3/2)} \right]^{-\kappa - 1}$ 
& $\frac{2k_BT}{m}\frac{\kappa_P-\sfrac{3}{2}}{\sqrt{2\kappa_P+1}}$ & $1+\ln Z_{\kappa_P} +(\kappa_P+1)\left(\psi(\kappa_P+1)-\psi(\kappa_P+\sfrac{1}{2})\right)$\\
\bottomrule
\end{tabular}
\end{table}

The informational content gained by clearly separating the scale and shape of non-exponential distributions is shown by the Boltzmann-Gibbs-Shannon (BGS) entropy, $H(f(x))=-\int_{x\in X}f(x)\ln f(x) \mathrm{d}x$, of the distributions, shown in the last column of Table \ref{tab_scale}. The structure of the entropy for the shape-scale distributions is a constant plus the logarithm of the partition function (normalization) plus a function of the shape. The entropy of the GPD equation has a direct simplicity. The components of the solution in terms of degrees of freedom can be characterized as $1$: energetic DOF (i.e., one random variable); $\ln \sigma$: configurational DOF; $\kappa$: nonlinear DOF. The focus of a generalized entropy for nonlinear, complex systems is to set the nonlinear DOF to zero and transfer that role to a nonlinear generalization of the configurational DOF. Towards this end, the uniqueness proof for entropy will result in a solution for a broad range of $d$-dimensional shape-scale distributions, $f(\mathbf{x};\sigma,\kappa,\alpha,d)$, that has a simple structure, whose dependence on the asymptotic shape, $\kappa$, is contained within a generalized logarithm, $\ln_g x$, of the partition function, Z:
\begin{equation}
\textbf{Required Entropy Solution:  } G\left(\frac{d}{\alpha}+\ln_{g(\kappa,\alpha,d)}Z(\sigma,\kappa,\alpha,d);\alpha\right).
\end{equation}
The outer function $G(u)$ adjusts for the role of the short-range nonlinearity, $\alpha$. The term $u$ quantifies the uncertainty at the informational scale of the density.

For clarity, the proof of a unique scale is completed for the one-dimensional case.  
\begin{lemma}[Independent Properties of the CSED Parameters]\label{lem_CSED}
   Given a random variable X distributed as the CSED pdf \eqref{equ_cexpdist}, then 
   \begin{enumerate}[label=\bfseries\Roman*., left=\parindent]
       \item the coupling divided by the stretching parameter, $\sfrac{\kappa}{\alpha}$, is the asymptotic tail shape  of the distribution, and when  $|x-\mu|\gg\sigma\left(\frac{\alpha}{\kappa}\right)^\frac{1}{\alpha}$   the distribution approaches a Type I Pareto or a pure power law distribution;
       \item the shape near the location is equal to the stretching parameter, $\alpha$, and when  $|x-\mu|\ll\sigma\left(\frac{\alpha}{\kappa}\right)^\frac{1}{\alpha}$    the distribution approaches a stretched exponential distribution;
       \item the scale parameter, $\sigma$, is exclusively the informational scale as defined by equation \ref{equ_info_scale}.
   \end{enumerate}
\end{lemma}
\begin{proof}
    $ $\\
   \begin{enumerate}[label=\bfseries\Roman*., left=\parindent]
   \item 
   \begin{align}
        \frac{1}{Z}\left(1+\kappa\frac{1}{\alpha}\left|\frac{x-\mu}        {\sigma}\right|^\alpha\right)^{-\frac{\alpha+\kappa}{\alpha\kappa}}
        &\approx \frac{1}{Z}\left(\kappa\frac{1}{\alpha}\left|\frac{x-\mu}{\sigma}\right|^\alpha\right)^{-\frac{\alpha+\kappa}{\alpha\kappa}} \nonumber \\
        &=\frac{\alpha\left(\frac{\kappa}{\alpha}\right)^{-\frac{1}{\kappa}}}
        {\sigma B\left(\frac{1}{\alpha},\frac{1}{\kappa}\right)}
       \left|\frac{x-\mu}{\sigma}\right|^{-\frac{\alpha+\kappa}{\kappa}}, \quad
       \text{when }  |x-\mu|\gg\sigma\left(\frac{\alpha}{\kappa}\right)^\frac{1}{\alpha}  
    \end{align}
Therefore, for $|x-\mu|\gg\sigma\left(\frac{\alpha}{\kappa}\right)^\frac{1}{\alpha}$  the CSED approximates a pure power law with exponent, $-\frac{\alpha}{\kappa}-1$. Therefore, the asymptotic shape of the CSED is $\kappa/\alpha$.
    \item For $|x-\mu|\ll\sigma\left(\frac{\alpha}{\kappa}\right)^\frac{1}{\alpha}$, the CSED converges to the stretched exponential distribution.  Therefore,
    \begin{align}
        f_\kappa^{\text{Stretch}}(x) \approx \frac{\alpha}{\sigma \Gamma\left(\frac{1}{\alpha}\right)}
        \exp\left(-\frac{1}{\alpha}\left|\frac{x-\mu}{\sigma}\right|^\alpha\right) 
        \text{ when  } |x-\mu|\ll\sigma\left(\frac{\alpha}{\kappa}\right)^\frac{1}{\alpha}.
    \end{align}
    \item Without loss of generality, assuming $\mu=0$,
    \begin{align}
        -\frac{\mathrm{d}\ln\left[\frac{1}{Z}\left(1+\frac{\kappa}{\alpha}\left|\frac{x}{\sigma}\right|
        ^\alpha\right)^{-\frac{1+\kappa/\alpha}{\kappa}}\right]}{\mathrm{d}\ln(x)} 
        &= -x \frac{\partial_x\left[\frac{1}{Z}\left(1+\frac{\kappa}{\alpha}\left|\frac{x}{\sigma}\right|
        ^\alpha\right)^{-\frac{1+\kappa/\alpha}{\kappa}}\right]}{\left[\frac{1}{Z}\left(1+\frac{\kappa}{\alpha}\left|\frac{x}{\sigma}\right|
        ^\alpha\right)^{-\frac{1+\kappa/\alpha}{\kappa}}\right]}\nonumber \\
        &=-\frac{(1+\frac{\kappa}{\alpha})\left|\frac{x}{\sigma}\right|^\alpha}{1+\frac{\kappa}{\alpha}\left|\frac{x}{\sigma}\right|^\alpha},
    \end{align}
    which for $x=\pm\sigma$ is equal to $-1$. Therefore, $\sigma$ is the informational scale.
   \end{enumerate}
\end{proof}

\begin{figure*}[ht] 
    \centering
    \subfloat{
        \includegraphics[width=0.475\linewidth,page=1]{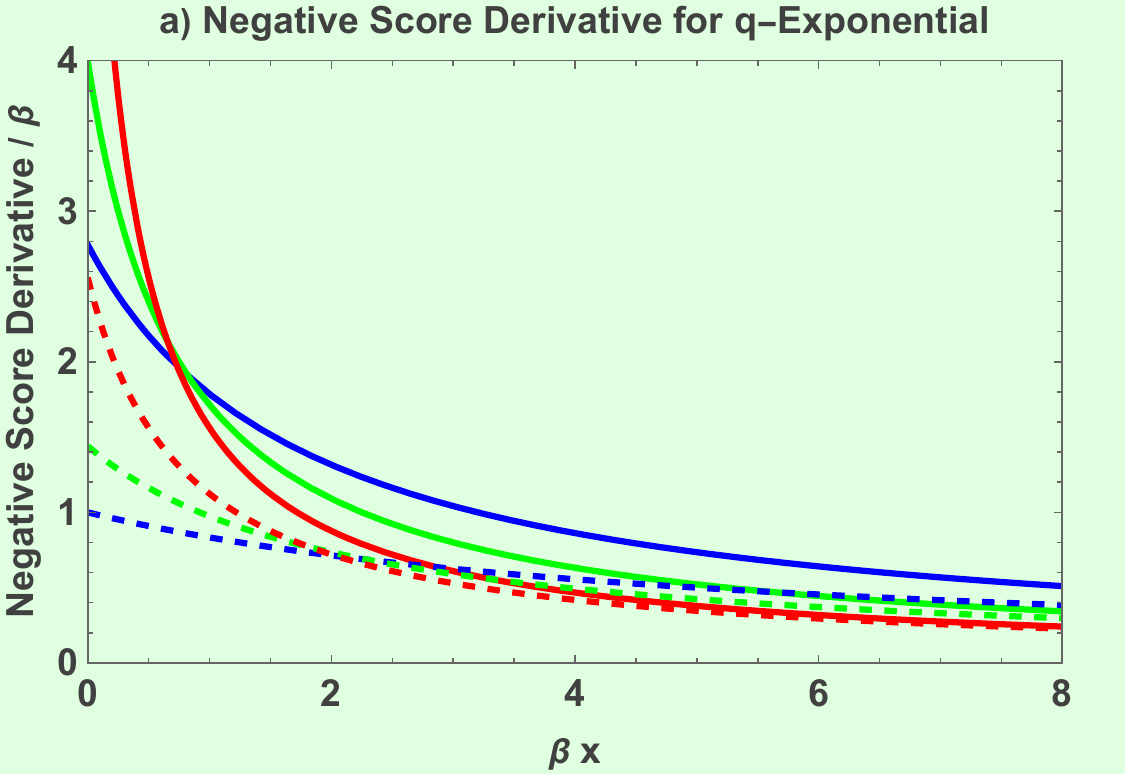}
        \label{fig_scale_a}
        }
    \hfill
    \subfloat[]{
        \includegraphics[width=0.475\linewidth,page=1]{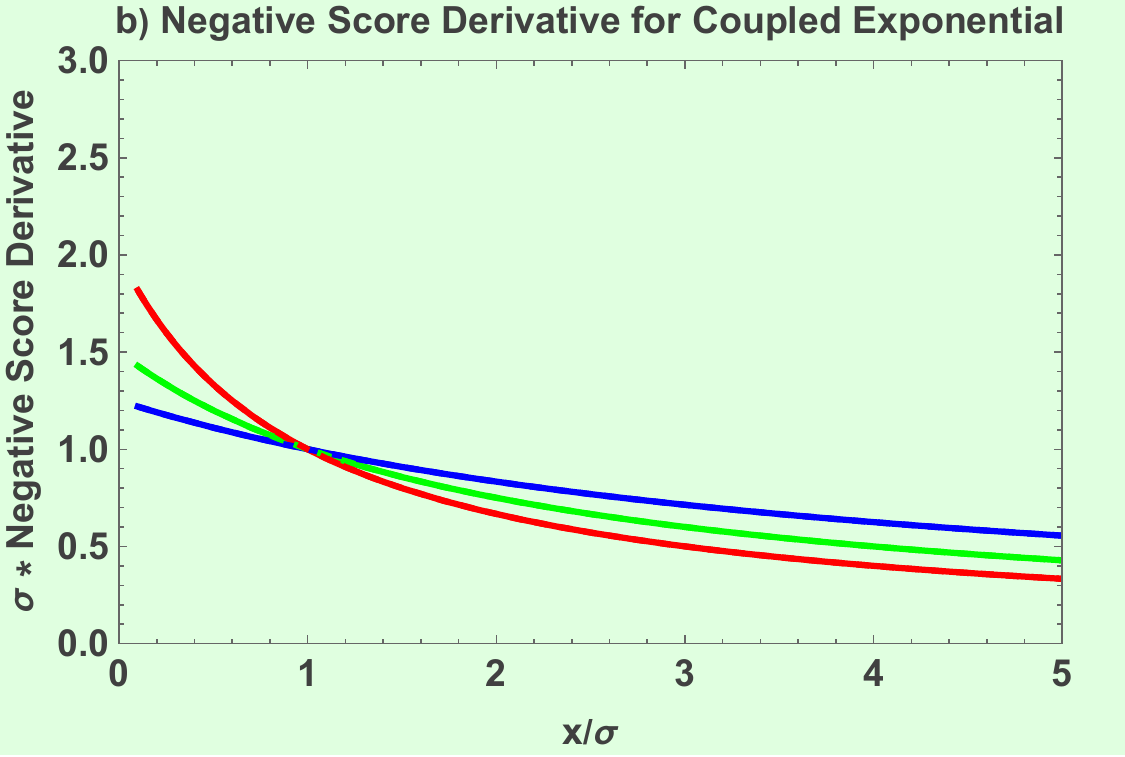}
        \label{fig_scale_b}
        }
    \caption{\textbf{The negative derivative of the score function (logarithm of the distribution) shows the uniqueness of the informational scale, $\boldsymbol{\sigma}.$ a)} The inverse-scale $(\beta)$  of the $q$-exponential does not have a common intersection. \textbf{b)} The informational scale $(\sigma),$  when normalized, has a common intersection independent of the shape $\kappa.$ The $-\sigma \ \mathrm{d}\ln(f(x))/\mathrm{d}x$ function is not dependent on $\sigma.$ } 
    \label{fig_scale}
\end{figure*}

Given the centrality of the informational scale, its measurement must be obtainable across the full range of the universality classes $(\kappa,\alpha)$ for nonlinear, complex systems. A proof is provided that samples from a CSED with $\kappa \neq 0$ should be regarded as occupying a fractal \textit{information dimension}. Furthermore, in the applications section on informational thermodynamics, we will see that this same information dimension is occupied by a nonlinear generalization of the system's internal energy. The information dimension is used to raise a distribution to a power and renormalize. This \textit{independent-equals} distribution represents the distribution of a fractal number of samples rather than one, thereby accounting for the long-range dependence between non-exponential samples.

\begin{theorem}[Measurement of Scale] \label{th_scale}
\par\noindent
\begin{enumerate}[left=\parindent]
    \item Given a complexity universality class and an independent informational scale $(\kappa,\alpha, \sigma)$, the scale of the coupled stretched exponential distribution for a particular class is measured by the average over the information dimension, $\iota(\kappa,\alpha)=1+\frac{\alpha\kappa}{\alpha+\kappa}$, of the density, which is called the independent-equals distribution 
    \begin{align}
        \sigma^\alpha &= \int_X x^\alpha f_\kappa(x;\alpha,\sigma) 
        /^{\iota(\kappa,\alpha)}\mathrm{d}x,
        \ \kappa > -\alpha
        \nonumber\\
        f_\kappa(x;\alpha,\sigma)/^{\iota(\kappa,\alpha)}
        &= \frac{f_\kappa(x;\alpha,\sigma)^{\iota(\kappa,\alpha)} }
        {\int_X f_\kappa(x;\alpha,\sigma)^{\iota(\kappa,\alpha)} \mathrm{d}x}
        \nonumber\\
        f_\kappa(x;\alpha,\sigma) &= \frac{1}{Z}\exp_\kappa 
        \left(\frac{x^\alpha}{\alpha \sigma^\alpha}\right)^{-\frac{\alpha+\kappa}{\alpha}}.
    \end{align}

\item An average over one information dimension, $\iota'(\kappa,\alpha) = 1$, is dependent on the complexity class. For the power-law domain, $\kappa>0$, the moment is finite over a limited domain,
    \begin{align}
        \frac{\sigma^\alpha}{1-\kappa} &= \int_X |x|^\alpha f_\kappa(x;\alpha,\sigma) 
        \mathrm{d}x, \ 
        -\alpha<\kappa <1.
    \end{align}
\end{enumerate}
\end{theorem}
\begin{proof}
\par\noindent
    \begin{enumerate}[left=\parindent]
        \item The proof will be completed by solving for the general case of $\iota'(\kappa,\alpha)$, then showing there is a unique solution. The exponential class, $\kappa=0$, establishes that the moment equals $\sigma^\alpha$. This moment is required to extend to $\kappa \neq 0$ to make sure that the definition for the scale is independent of the functional nonlinearity. The proof is completed for  $\kappa>0$, and extends by continuation to $-\frac{\alpha}{1+\alpha}<\kappa<0$. Solving for $\int_X x^\alpha f_\kappa(x;\alpha,\sigma)/^{\iota'(\kappa,\alpha)} \mathrm{d}x$, with $y=\frac{x^\alpha}{\alpha\sigma^\alpha}$ and $dy=\frac{x^{\alpha-1}}{\sigma^\alpha}dx$, the full solutions involves two integrals. The normalization is
   \begin{align}
        \int_X f_\kappa(x;\alpha,\sigma)^{\iota'(\kappa,\alpha)} \mathrm{d}x
        &= 2Z^{-\iota'(\kappa,\alpha)} \sigma\alpha^{\frac{1}{\alpha}-1}
        \int_0^\infty 
        \left(1+\kappa y\right)
        ^{-\frac{\iota'(\kappa,\alpha)(\alpha+\kappa)}{\alpha\kappa}}
         y^\frac{1-\alpha}{\alpha}\mathrm{d}y
    \end{align}
Using the Beta integral, $\int_0^\infty y^{\frac{1}{\alpha}-1} (1+\kappa y)^{-p} dy = \kappa^{-\frac{1}{\alpha}} B(\frac{1}{\alpha},p-\frac{1}{\alpha}),\ p=\frac{\iota'(\kappa,\alpha)(\alpha+\kappa)}{\alpha\kappa}$,
    \begin{align}
        \int_X f_\kappa(x;\alpha,\sigma)^{\iota'(\kappa,\alpha)} \mathrm{d}x
        &= 2Z^{-\iota'(\kappa,\alpha)} 
        \sigma\alpha^{\frac{1}{\alpha}-1} \kappa^{-\frac{1}{\alpha}}
        B\left(\frac{1}{\alpha},
        \frac{\iota'(\kappa,\alpha)(\alpha+\kappa)-\kappa}{\alpha\kappa}\right).
    \end{align}
And the moment component is 
    \begin{align}
        \int_X |x|^\alpha f_\kappa(x;\alpha,\sigma)^{\iota'(\kappa,\alpha)} 
        \mathrm{d}x
        &= 2Z^{-\iota'(\kappa,\alpha)} \sigma\alpha^{\frac{1}{\alpha}-1}
        \int_0^\infty 
        \left(1+\kappa y\right)
        ^{-\frac{\iota'(\kappa,\alpha)(\alpha+\kappa)}{\alpha\kappa}}
         y^\frac{1}{\alpha}\mathrm{d}y
         \nonumber\\
         &= 2Z^{-\iota'(\kappa,\alpha)} 
        \sigma^{\alpha+1}\alpha^{\frac{1}{\alpha}} \kappa^{-\frac{1}{\alpha}}
        B\left(\frac{\alpha+1}{\alpha},\frac{\iota'(\kappa,
        \alpha)(\alpha+\kappa)-\kappa(\alpha+1)}{\alpha\kappa}\right).
    \end{align}
So, the full moment is 
    \begin{align} 
        \label{equ_gen_mom}
        \int_X |x|^\alpha f_\kappa(x;\alpha,\sigma)/^{\iota'(\kappa,\alpha)} \mathrm{d}x
        &= \frac{
        \sigma^\alpha  \alpha 
        B\left(\frac{\alpha+1}{\alpha},\frac{\iota'(\kappa,
        \alpha)(\alpha+\kappa)-\kappa(\alpha+1)}{\alpha\kappa}\right)
        }{ 
         \kappa B\left(\frac{1}{\alpha},
        \frac{\iota'(\kappa,\alpha)(\alpha+\kappa)-\kappa}{\alpha\kappa}\right)
        }
    \end{align}
     Setting $A = \frac{\iota(\kappa,\alpha)(\alpha+\kappa)-\kappa}{\alpha\kappa}$, dividing out the moment, and canceling the beta functions, gives the equation
    \begin{align}
        1&= \frac{
        \alpha 
        B\left(\frac{\alpha+1}{\alpha},A-1\right)
        }{ 
         \kappa B\left(\frac{1}{\alpha},
        A\right)}
        \nonumber\\
        &= \frac{1}{\kappa (A-1)}
        = \frac{\alpha}{\iota(\kappa,\alpha)(\alpha+\kappa)-\kappa-\alpha\kappa}
        \nonumber\\
        \iota(\kappa,\alpha) &= \frac{\alpha + (1 +\alpha)\kappa}
        {\alpha+\kappa} 
        =1 + \frac{\alpha\kappa}{\alpha+\kappa}.
    \end{align}
    This proves the first statement.
    \item The second statement is proven by substituting $\iota'(\kappa,\alpha)=1$ into \ref{equ_gen_mom}. 
    \begin{align}
        \frac{
        \sigma^\alpha  \alpha 
        B\left(\frac{\alpha+1}{\alpha},\frac{(\alpha+\kappa)-\kappa(\alpha+1)}{\alpha\kappa}\right)
        }{ 
         \kappa B\left(\frac{1}{\alpha},
        \frac{(\alpha+\kappa)-\kappa}{\alpha\kappa}\right)
        }
        = \frac{\sigma^\alpha  }
        {\kappa(\frac{(\alpha+\kappa-\kappa-\alpha\kappa}{\alpha\kappa})}
        = \frac{\sigma^\alpha}
        {1-\kappa}
    \end{align}
    Completing the second proof.
    \end{enumerate}
\end{proof}
The specification of the information dimension for nonlinear systems provides the proof regarding the requirement for the entropy function \ref{equ_entropy} and its principal constraint to be an average over the independent-equals moment. This also resolves that the root $1/\alpha$ of the coupled logarithm must be outside the average; otherwise the moment would be over $x$ instead of $x^\alpha$.

With these foundations regarding the precision with which the scale of a shape-scale distribution must be defined and measured, we can now turn to the calibrated definition for entropy. Central to the approach will be a recognition that entropy is not simply a mathematical abstraction, but a physical measurement of a system's degrees of freedom.

\subsection*{Physical definition, axioms, and theorem for entropy}
\label{subsec_axioms}

Shannon \cite{shannonMathematicalTheoryCommunication1948} and Khinchin \cite{khinchinMathematicalFoundationsInformation1957} established their, now famous, ``properties'' for defining the entropy function. Later, these were canonized as the ``axioms'' for continuity, maximality, extensibility, and additivity, which were proven to be satisfied by the entropy function, $-\sum_i p_i\ln p_i$, established by Gibbs and Boltzmann \cite{jaynesGibbsVsBoltzmann1965}. In turn, investigations into extending entropy to nonlinear, complex systems have focused on modifying the fourth axiom, with the assumption that a well-designed axiom would provide the basis for a uniqueness proof regarding alternative functions. A comprehensive generalization was defined by Hanel-Thurner \cite{hanelComprehensiveClassificationComplex2011}, such that the first three axioms are satisfied by a trace-form function. Tempesta \cite{tempestaShannonKhinchinFormulation2016} defined a general axiom based on a universal group composition, which includes non-trace functions. The unintended consequence of canonizing the Shannon-Khinchin axioms is that these mathematical properties disconnected the definition of entropy from the measurement of a physical property.

I propose to address this lack of physicality in proposed entropy functions by focusing on the physical measurements instrumental to entropy, and thereby prove that the establishment of a common unit of measure across all complex universality classes leads to a unique calibration for entropy. The proposed definition and axioms emphasize clarity regarding the important assumptions, rather than a minimal foundation. Future research could certainly trim these axioms and strengthen the theorems; however, given the proposed calibration, the definition of entropy will not change. While it may be possible to extend the definition, it will not be easy to do so. An investigator seeking an alternative calibrated measure of entropy will now have to a) prove that the Hanel-Thurner classification of complex systems was not comprehensive and/or b) prove that the independent scale is not the unique variable quantifying the linear uncertainty, independently of the nonlinearity.

\begin{definition}[Entropy]
Entropy is the measure of a system's dimensionless degrees of freedom, $N$, given the probabilities, $\mathbf{p}$, of the microstates. Given a microcanonical ensemble (i.e., equal probabilities) with a growth of microstates, $W(N)$, dependent on short-range, $\alpha$, and long-range, $\kappa$, nonlinearities, there is a calibration such that the dimensionless entropy, $H$, is equal to the degrees of freedom.  Physical measures of entropy, $S$, include a a conversion factor, such as the Boltzmann constant, $k_B = 1.38 \times 10^{-23} \mathrm{\ Joules/Kelvin}$, 
    \begin{align}
        S(p_W)=kH(p_W)  = k N = k \frac{x}{\sigma},
    \end{align}
where $x$ is the physical variable quantifying the degrees of freedom, such as Joules of energy, and $\sigma$ is the informational scale calibrating the granularity, such as the Joules per degree of freedom. For a general probability distribution, entropy is the continuous extension of this counting principle, such that the informational scale, $\sigma$, is strictly independent of the nonlinearities. 
\end{definition}

\begin{phyAxioms}[Entropy for Complex Systems]
Given a system with a growth of microstates as a function of the dimensionless degrees of freedom, $W(N)$, the function has the following axiomatic properties that provide a physical foundation for the measurement of entropy:
\begin{enumerate}[label=\alph*),left=\parindent]
    \item  \textbf{Universality Class:} in the asymptotic limit, the growth is specified by a universality class $(\kappa,\alpha)$ such that
    \begin{itemize}
        \item $\alpha>0$, the short-range (variable) nonlinearity of the system, that specifies the variable stretch, $\frac{N^\alpha}{\alpha}$, with the option to substitute $N\rightarrow \ln M$ for lognormal and related scaling, and
        \item $\kappa$, the long-range (functional) nonlinearity, that specifies the type of scaling, including
            \begin{itemize}
                \item $\kappa=0$, exponential scaling, $p_{W(N)}\sim e^{-\frac{N^\alpha}{\alpha}}$, or 
                log-normal, $p_{W(N)}\sim e^{-\frac{(\ln M)^\alpha}{\alpha}}$, for $N\rightarrow\infty$;
                \item $\kappa>0$, power-law scaling, $p_{W(N)}\sim N^{-(1 + \frac{\alpha}{\kappa})}$ , for $N\rightarrow\infty$;
                \item ${-\frac{\alpha}{1+\alpha}}<\kappa<0$, super-exponential scaling, $p_{W(N)}\sim N^{-(1+\frac{\alpha}{\kappa})}$, 
                for $N \rightarrow N_{max}=\left|\frac{\alpha}{\kappa}\right|^\frac{1}{\alpha}$.
            \end{itemize}
    \end{itemize} 
    \item \textbf{Calibrated Moment:} the unique, dominant, informational scale, $\sigma$, specifies the granularity of a physical variable, $x$, such that $N =\frac{x}{\sigma}$, and $x=\sigma$ establishes a boundary between stretched exponential growth and power-law growth such that
    \begin{align} \label{equ_loglogGrowth}
        \left. \frac{\mathrm{d}\ln p_{W\left(x/\sigma\right)}}
        {\mathrm{d}\ln x} \right|_{x=\sigma}
        = -1. 
    \end{align}
This informational scale is measurable by a unique moment $\sigma^\alpha=\int_X x^\alpha g(f(x);\kappa,\alpha)dx$, where $f(x)$ is a distribution with this log-log derivative property and $g(u;\kappa,\alpha)$ is a modified distribution. Entropy, which must be measurable, is an average over $g(u;\kappa,\alpha)$.   
    \item \textbf{Nonlinear Dynamics} The entropy function with a precise unit of measure is calibrated from the probability of the equipartition microstates, $p_W(x)$, given dynamics specified by a nonlinear differential equation:
    \begin{align} \label{eq_nde}
        p'_{W(x/\sigma)} + a x^b p_{W(x/\sigma)}^c = 0, \quad \text{with } p_{W(0)} = 1.
    \end{align}
The parameters $a,b, \text{and }c$ are constants, since the axioms specify constants for the universality class $(\kappa,\alpha)$ with a scale $\sigma$. 

\end{enumerate}
\end{phyAxioms}

\begin{theorem}[Unique, Universal Entropy]
    Given the requirement that entropy measures the physically-scaled degrees of freedom, and the axioms that the asymptotic growth of microstates contains a unique granularity scale, independent of the universality classes \eqref{equ_loglogGrowth} the unique measure of entropy that is comprehensive across all scaled universality classes $(\kappa,\alpha,\sigma)$ is the coupled entropy, 
    \begin{align} \label{equ_entropy}
        H(\mathbf{p};\kappa,\alpha,k) 
    &=k\alpha^\frac{1}{\alpha} \left( \sum_{i=1}^W
    p/_i^{\iota(\kappa,\alpha)}
    \ln_{\kappa} p_i^{-\frac{\alpha}{\alpha+\kappa}} 
    \right)^\frac{1}{\alpha}
    \nonumber \\
    & = k\left(\frac{\alpha}{\kappa}\right)^\frac{1}{\alpha}\left(\left(
    \sum_{j=1}^W p_j^\frac{\alpha+(1+\alpha)\kappa}{\alpha+\kappa}
    \right)^{-1}-1\right)^\frac{1}{\alpha},
    \end{align}
where,
\begin{itemize}[,left=\parindent]
    \item k is a conversion factor, such as the Boltzmann constant, $k_B = 1.38 \times 10^{-23} J/K$,
    \item $\iota(\kappa,\alpha)=1+\frac{\alpha\kappa}{\alpha+\kappa}$ is the information dimension of a sample, 
    \item $p_i/^{\iota(\kappa,\alpha)}= \frac{ p_i^{\iota(\kappa,\alpha)}}{\sum_{j=1}^W p_j^{\iota(\kappa,\alpha)}}$ is the independent-equals probability, and
    \item $\ln_\kappa x = \frac{1}{\kappa}\left(x^\kappa-1\right)$ is the coupled logarithm.
\end{itemize}
\end{theorem}
\begin{proof}
    \par\noindent
    \begin{enumerate}[left=\parindent]
        \item The constants of the nonlinear differential equation \ref{eq_nde} specifying the probability of the equipartition microstates are functions of the universality class and scale: 
        \begin{itemize}
        \item $a(\kappa,\alpha,\sigma)$
            \item $b(\alpha)$ since the nonlinearity of the variable is only determined by the short-range nonlinearity, and 
        \item $c(\kappa,\alpha)$ since the scale does not influence the asymptotic decay. 
        \end{itemize} 
The constants $a(0,\alpha,\sigma)$ and $b(\alpha)$ are solved using the exponential system $(\kappa=0)$,  
        \begin{align}
            \int_0^\infty \frac{p'_{W(x/\sigma)}}{p_{W(x/\sigma)} }&= - a(0,\alpha,\sigma)
            \int_0^\infty  x^{b(\alpha) } \mathrm{d}x
            \nonumber\\
            p_{W(x/\sigma)} &= \exp ^{-a(0,\alpha,\sigma)\frac{x^{b(\alpha)+1}}{(b(\alpha)+1)}} + C.
        \end{align}
$C=0$, since $p_W(0)=e^0 + C=1$. From the exponential scaling
\begin{align}
b(\alpha) &=\alpha-1
\nonumber\\
 a(0,\alpha,\sigma) &= \frac{1}{\sigma^\alpha}
\end{align}
The constants $a(\kappa,1/\sigma^\alpha)$ and $c(\kappa,\alpha)$ are solved using the full nonlinear system, assuming $c \neq 1$,
\begin{align}
             \int_0^\infty \frac{p'_{W(x/\sigma)}}{p_{W(x/\sigma)}^{c(\kappa,\alpha)} }&= - a(\kappa,1/\sigma^\alpha)
            \int_0^\infty  x^{\alpha-1} \mathrm{d}x
            \nonumber\\
            \frac{1}{1-c(\kappa,\alpha)} p_{W(x/\sigma)}^{1-c(\kappa,\alpha)}
            &= -a(\kappa,1/\sigma^\alpha) \frac{x^\alpha}{\alpha} + C.
        \end{align}
From, $p_W(0)=1$, $C = \frac{1}{1-c(\kappa,\alpha)} $, and simplifying
\begin{align}
    p_{W(x/\sigma)}
            &= \left(
            1+ (1-c(\kappa,\alpha))a(\kappa,1/\sigma^\alpha) \frac{x^\alpha}{\alpha} 
            \right)^\frac{1}{1-c(\kappa,\alpha)}.
\end{align}
The asymptotic tail requires
\begin{align}
    \frac{\alpha}{1-c(\kappa,\alpha)} &= -\left(\frac{\alpha+\kappa}{\kappa}\right)
    \nonumber\\
    c(\kappa,\alpha) &=1+\frac{\alpha\kappa}{\alpha+\kappa}
    = \iota(\kappa,\alpha).
\end{align}
The informational scale condition $\left. \frac{\mathrm{d}\ln p_{W\left(x/\sigma\right)}}
        {\mathrm{d}\ln x} \right|_{x=\sigma}
        = -1$, requires
        \begin{align}
            -(1-c(\kappa,\alpha))a(\kappa,1/\sigma^\alpha) 
            &= \frac{\kappa}{\sigma^\alpha}
            \nonumber \\
            a(\kappa,1/\sigma^\alpha) &= \frac{\alpha+\kappa}{\alpha \sigma^\alpha}.
        \end{align}
Therefore, the solution for the differential equation is
\begin{align}
        p'_{W(x/\sigma)} + \frac{\alpha+\kappa}{\alpha \sigma^\alpha} 
        x^{\alpha-1} p_{W(x/\sigma)}^{\iota(\kappa,\alpha)} = 0, \quad \text{with } p_{W(0)} = 1.
    \end{align}
is the coupled stretched exponential function
\begin{align}
    p_{W(x/\sigma)}
            &= \left(
            1+ \kappa \frac{x^\alpha}{\alpha \sigma^\alpha} 
            \right)^{-\frac{\alpha+\kappa}{\alpha\kappa}}.
\end{align}

        \item   Given the coupled stretched exponential model for the asymptotic growth of states, the equipartition entropy is proven by the inverse of the growth of states,
        \begin{align}
            p_W &= W(N)^{-1} = \exp_\kappa^{-(1+\kappa/\alpha)} \frac{N^\alpha}{\alpha}
            \nonumber\\
            W^{-1}(N)^{-1}&= \left(\alpha \ln_\kappa p_W^{-\frac{\alpha}{\alpha+\kappa}}\right) ^\frac{1}{\alpha}
            \nonumber\\
            H(p_W;\kappa,\alpha,k) &= k W^{-1}(N)^{-1} = 
            k \left(\alpha\ln_\kappa p_W^{-\frac{\alpha}{\alpha+\kappa}}\right)^\frac{1}{\alpha}
        \end{align}
            Substituting $p_W$ and simplifying gives the degree of freedom,
            \begin{align} \label{equ_equiEnt}
                H(p_W;\kappa,\alpha,k) &= 
            k \alpha^\frac{1}{\alpha} \left(\ln_\kappa 
            \left(\exp_\kappa^{-(1+\kappa/\alpha)} \frac{N^\alpha}{\alpha} \right)
            ^{-\frac{\alpha}{\alpha+\kappa}}\right)^\frac{1}{\alpha}
            =kN = k\frac{x}{\sigma}.
            \end{align}
            \item Given the calibration axiom, entropy is a unique moment of the measurement $\ln_\kappa p_W^{-\frac{\alpha}{\alpha+\kappa}}$.   The measurement of scale theorem \ref{th_scale} proves that the independent-equals distribution is the only moment that returns the correct scale. Therefore, the entropy for the $(\kappa,\alpha)$ class must be an average over the matching independent-equals moments with the information dimension $\iota(\kappa,\alpha) = 1 + \frac{\alpha\kappa}{\alpha+\kappa}$. 
            \item Return of the degree of freedom requires that the root $1/\alpha$  of the moment $N^\alpha$ be computed. The root cannot be computed inside the summation, because then the moment would be over $x$ rather than $x^\alpha$. The multiplier $\frac{1}{\alpha}N^\alpha$ must also be inverted; hence the multiplication by $\alpha^{1/\alpha}$ for the entropy function.
            \item The proof that the entropy function is comprehensive across all of the complex universality classes is provided in theorem \ref{th_universality}.
    \end{enumerate}
\end{proof}

The composability of entropy is an outcome of the physical measurement, rather than an assumed axiom.
\begin{lemma}[Composability of entropy]
    \par\noindent
    \begin{enumerate}[left=\parindent]
        \item 
        Given an $\alpha$-norm composition of the degrees of freedom, $W\left(a(\frac{N^\alpha}{\alpha} +\frac{M^\alpha}{\alpha}) ^\frac{1}{\alpha}\right)$, within the complexity class $(\kappa,\alpha)$, the probabilities compose by the coupled product and are thus dependent, and the entropy composes by the $\alpha$-norm
        \begin{align}
            p_{W\left(a(\frac{N^\alpha}{\alpha} +\frac{M^\alpha}{\alpha}) ^\frac{1}{\alpha}\right)} 
            &= p_{W(N)} \otimes_{-\frac{\alpha\kappa}{\alpha+\kappa}}
            p_{W(M)}
            = \left(p_{W(N)}^{-\frac{\alpha\kappa}{\alpha+\kappa}} + 
            p_{W(M)}^{-\frac{\alpha\kappa}{\alpha+\kappa}} -1
            \right)^{-\frac{\alpha+\kappa}{\alpha\kappa}} 
            \\
            H(p_{W\left((N^\alpha+M^\alpha)^\frac{1}{\alpha}\right)},\kappa,\alpha,k) &=
            a\left(\frac{1}{\alpha}H(p_{W(N)},\kappa,\alpha,k)^\alpha + 
            \frac{1}{\alpha}H(p_{W(M)},\kappa,\alpha,k)^\alpha)\right)^\frac{1}{\alpha}.
        \end{align}
        \item Given the multiplication of the microstates, W(N)W(M), within the complexity class $(\kappa,\alpha,k)$, the probabilities multiply and are thus independent, and the entropy composes via the coupled sum
        \begin{align}
            p_{W(N)} p_{W(M)} &= p_{W(\frac{N^\alpha}{\alpha} 
            \oplus_\kappa \frac{M^\alpha}{\alpha})}
            \nonumber\\
            H(p_{W(N)} p_{W(M)},
            \kappa,\alpha,k)
            &= a\left(\frac{1}{\alpha}H(p_{W(N )},\kappa,\alpha,k)^\alpha \oplus_\kappa 
            \frac{1}{\alpha}H(p_{W(M)},\kappa,\alpha,k)^\alpha\right)^\frac{1}{\alpha}.
        \end{align}
    \end{enumerate}
\end{lemma}
\newpage
\begin{proof}
    \par\noindent
    \begin{enumerate}[left=\parindent]
        \item
        \begin{align}
            p_{W\left(a(\frac{1}{\alpha}N^\alpha+\frac{1}{\alpha}M^\alpha)^\frac{1}{\alpha}\right)}
            &= 
            \exp_\kappa^{-(1+\kappa/\alpha)} 
            \frac{N^\alpha+M^\alpha}{\alpha }
            \nonumber\\
            &= \exp_\kappa \frac{N^\alpha}{\alpha }
            \otimes_{-\frac{\alpha\kappa}{\alpha+\kappa} }
            \exp_\kappa
            \frac{M^\alpha}{\alpha }
            \nonumber\\
            &=\left(p_{W(N)}^{-\frac{\alpha\kappa}{1+\kappa/\alpha}} 
            +
            p_{W(M)}^{-\frac{\alpha\kappa}{1+\kappa/\alpha}} - 1
            \right)^{-\frac{\alpha+\kappa}{\alpha\kappa}}
            \nonumber\\
            H(p_{W\left(a(\frac{1}{\alpha}N^\alpha+\frac{1}{\alpha}M^\alpha)
            ^\frac{1}{\alpha}\right)},\kappa,\alpha,k) 
            &= \left( a(\frac{1}{\alpha}N^\alpha+\frac{1}{\alpha}M^\alpha)
            ^\frac{1}{\alpha} \right)^\frac{1}{\alpha}
            \nonumber\\
            &= a\left(\frac{1}{\alpha}H(p_{W(N)},\kappa,\alpha,k)^\alpha + 
            \frac{1}{\alpha}H(p_{W(M)},\kappa,\alpha,k)^\alpha)\right)^\frac{1}{\alpha}
        \end{align}
        \item 
        \begin{align}
             p_{W(N)} p_{W(M)}&= 
             \left(\exp_\kappa
            \frac{N^\alpha}{\alpha k^\alpha} \ 
            \exp_\kappa 
            \frac{M^\alpha}{\alpha k^\alpha}
            \right)^{-(1+\kappa/\alpha)} 
            \nonumber\\
            &= \exp_\kappa^{-(1+\kappa/\alpha)}
            \frac{N^\alpha \oplus_\kappa M^\alpha}{\alpha k^\alpha}
            = p_{W(N^\alpha \oplus_\kappa M^\alpha)} 
            \nonumber\\
            H\left(p_{W(N)} p_{W(M)},
            \kappa,\alpha,k\right)
            &= k \alpha^\frac{1}{\alpha} \left(
            \ln_\kappa \exp_\kappa
            \frac{N^\alpha \oplus_\kappa M^\alpha}{\alpha k^\alpha}
            \right)^\frac{1}{\alpha}
            \nonumber\\
            &= \left(N^\alpha \oplus_\kappa M^\alpha\right)^\frac{1}{\alpha}
            \nonumber\\
            &= \left(H(p_{W(N )},\kappa,\alpha,k)^\alpha \oplus_\kappa 
            H(p_{W(M)},\kappa,\alpha,k)^\alpha\right)^\frac{1}{\alpha}
        \end{align}
    \end{enumerate}
\end{proof}

This coupled entropy is maximized by the coupled stretched exponential distribution.

\begin{lemma}[Entropy maximized by the coupled stretched exponential distribution]
    Given the entropy function
    \begin{equation} \frac{1}{\alpha}H(f(x);\kappa,\alpha)^\alpha = \int_{x\in\mathcal{X}}  
        f(x)/^{\left(\frac{\alpha+(\alpha+1)\kappa}{\alpha+\kappa}\right)}\ln_{\kappa}f(x)^{-\frac{\alpha}{\alpha+\kappa}} \mathrm{d}x
    \end{equation} 
and  two constraints, the normalization and  the independent-equals moment, $\mu_1^{\iota(\kappa,\alpha)}=\sigma^\alpha$;
then, the maximum coupled entropy distribution is the coupled stretched exponential of definition \ref{def_CSED}, 
\begin{align}
f_\kappa(x;\alpha) &= \frac{1}{Z_\kappa(\sigma,\alpha)}\left(1+\kappa\frac{x^\alpha}{\alpha\sigma^\alpha}\right)_+^{-\frac{\alpha+\kappa}{\alpha\kappa}}
=\exp_{\kappa}
\left(\frac{x^\alpha}{\alpha \sigma^\alpha}
\oplus_{\kappa} \ln_{\kappa} Z^\frac{\alpha}{\alpha+\kappa}
\right)^{-\frac{\alpha+\kappa}{\alpha}},
\end{align}
with a domain of concavity $\kappa > -\frac{\alpha}{\alpha+1}$. The Lagrangian multipliers are
\begin{align}
    \begin{matrix}
        \text{Normalization:} & \lambda_0 =  
        - \frac{\alpha} {(\alpha+\kappa) 
    \int_{x\in\mathcal{X}} 
    f(x)^{\iota(\kappa,\alpha)} d(x)}; \\
    \text{Independent Equals Moment:} & \lambda_1 =
    \frac{Z_\kappa^\frac{\alpha\kappa}{\alpha+\kappa} }
     {\alpha \sigma^\alpha}.
    \end{matrix}
\end{align}
\end{lemma}

\begin{proof}
    See Methods subsection \nameref{subsec_maxProof}.
\end{proof}

Given the proof regarding the calibrated definition of entropy for the comprehensive complexity classes, the unique and universal properties of entropy will be reviewed next, followed by thermodynamic and information applications.

\subsection*{Entropy's unique properties}\label{subsec_unique}
The proof regarding the unique calibration of the coupled entropy is based on the properties of the growth of states given the microcanonical ensemble. The uniqueness of the coupled entropy extends to the canonical ensemble in which the entropy is maximized given a constraint on the informational scale.

In the next two sections, I'll use the $d$-dimensional definition for the entropy, which accounts for the derivative term in the PDF's asymptotic exponent $\underbrace{\left(\frac{\alpha}{\kappa}+d\right)}_{pdf}\underbrace{\left(\frac{\kappa}{\alpha+d\kappa}\right)}_{entropy}$.

\begin{definition}[Multivariate definition for the coupled entropy]
    Given a $d$-dimensional random variable $\mathbf{X}$ over the sample space, $\mathcal{X}$, a coupling parameter that constitutes the long-range dependence, $\kappa > -\frac{\alpha}{\alpha+d},$  and a stretching parameter that constitutes the short-range dependence, $\alpha>0$, the coupled entropy is defined as: 
    \begin{align} \label{equ_CE}
\frac{1}{\alpha}H(\mathbf{X};\kappa,\alpha,d,k)^\alpha =
k^\alpha
\begin{cases}
\displaystyle 
\int_{\mathbf{x}\in\mathcal{X}}
f(\mathbf{x})/^{1+\frac{\alpha\kappa}{\alpha+d\kappa}}
\ln_{\kappa}f(\mathbf{x})^{-\frac{\alpha}{\alpha+d\kappa}}
\, d\mathbf{x}
& \\[1.2em]
\displaystyle
= \ln_\kappa\left( 
\left( \int_{\mathbf{x}\in\mathcal{X}}
f(\mathbf{x})^{1+\frac{\alpha\kappa}{\alpha+d\kappa}}
\, d\mathbf{x}
\right)^{-\frac{1}{\kappa}}
\right)
& \kappa \neq 0, \; \kappa > -\frac{\alpha}{\alpha+d} \\[1.2em]
\displaystyle
-\int_{\mathbf{x}\in\mathcal{X}}
f(\mathbf{x})\ln f(\mathbf{x})\, d\mathbf{x}
& \kappa = 0
\end{cases}
\end{align}
The first expression shows the independent-equals average over the coupled logarithm, while the second expression is equivalent to the generalized mean (geometric mean for $\kappa=0$) within the coupled logarithm, though there is a cancellation of terms in the exponent that is not shown. The equivalence of the expressions is derived in the Methods subsection \nameref{subsec_Entropies}.
\end{definition}

While the coupled stretched exponential distributions have a domain, $\kappa>-\frac{\alpha}{d}$, defined by the requirement of integration, the coupled entropy has a more restricted domain, $\kappa>-\frac{\alpha}{\alpha+d}$, defined by the requirement that entropy be concave. The concavity domain proof is within the maximization proof in Methods subsection \nameref{subsec_maxProof}.

The most important reason this particular calibration is required is that it guarantees that the DOF due to the nonlinearity of the coupled stretched exponentials is set to zero and the configurational DOF is measured with a generalized logarithm.

\begin{theorem}[Unique solution for entropy of the coupled stretched exponentials]
    Given a $d$-dimensional random variable distributed as a coupled stretched exponential, $\mathbf{X} \sim f_\kappa(\mathbf{x}; \boldsymbol{\mu},\boldsymbol{\Sigma},\alpha,d)$, as defined in equation \ref{def_CSED}, the Boltzmann-Gibbs-Shannon entropy expressed in terms of the energetic, configurational and nonlinear degrees of freedom is
    \begin{align}
H_{BGS}(\textbf{X}) = \frac{1}{\alpha}
\begin{cases}
\displaystyle
\overbrace{d}^{\text{energetic}}
+
\overbrace{\alpha \ln Z(\kappa,\alpha,d)}^
{\text{configurational}}
+
\overbrace{
\frac{\alpha+d\kappa}{\kappa}
\left[
\psi\left(\frac{d}{\alpha}+\frac{1}{\kappa}\right)
-
\psi\left(\frac{1}{\kappa}\right)
\right]-d
}^{\text{nonlinear DOF}},
& \kappa>0,\\[1.5em]
\displaystyle
d+ \alpha \ln Z(0,\alpha,d),
& \kappa=0,\\[1.5em]
\displaystyle
d+
\alpha \ln Z(\kappa,\alpha,d)
+
\frac{\alpha+d\kappa}{\kappa}
\left[
\psi\left(1-\frac{d}{\alpha}-\frac{1}{\kappa}\right)
-
\psi\left(1-\frac{1}{\kappa}\right)
-d \right],
& -\frac{\alpha}{d}<\kappa<0.
\end{cases}
    \end{align}
    
The coupled entropy, with its unique calibration, matched to the nonlinearity and dimension $(\kappa,\alpha,d)$ of the CSED reduces the nonlinear DOF to zero and modifies the configurational DOF to a generalized logarithm 
\begin{align}
    H (\mathbf{X};\kappa,\alpha,d)^\alpha  &= \underbrace{d}_{\text{energetic}} 
    + \  \underbrace{\ln_\frac{\kappa}{\alpha+d\kappa}
    Z_\kappa(\sigma,\alpha,d)^\alpha}_{\text{configurational DOF}}
    \ + \underbrace{0}_{\text{nonlinear}}
\end{align} 
The expression on the right is the value of the density of the distribution at the radius $r=\left(((\mathbf{x}-\boldsymbol{\mu})^{\circ \sfrac{\alpha}{2} \top} 
(\boldsymbol{\Sigma}^{\circ \sfrac{\alpha}{2}})^{-1}(\mathbf{x}-\boldsymbol{\mu})^{\circ \sfrac{\alpha}{2}}\right)^\frac{1}{\alpha}$ which for example occurs when $x_i=\mu_i+\sigma_{ii}$ for all $i.$ In the limit of infinite coupling, the entropy is proportional to the partition function,
\begin{align}
    \lim\limits_{\kappa\rightarrow\infty} H(\mathbf{X};\kappa,\alpha,d)= 
    d^\frac{1}{\alpha} Z_\kappa(\sigma,\alpha,d)^\frac{1}{d}.
\end{align}
\end{theorem}
\begin{proof}
    See Methods subsubsection \nameref{subsec_maxProof}.
\end{proof}

Figure \ref{fig_CE} shows the coupled entropy of the coupled stretched exponential across a range of coupling values for $(\alpha = (1/2,1,2))$. In each case the limit as $\kappa \rightarrow \infty$ is for $H(X,\kappa,\alpha)= Z(\sigma,\alpha)$. The property of the coupled entropy representing the uncertainty at the informational scale of the distribution is shown in Figure \ref{fig_entDensity}. By removing the DOF due to the nonlinearity, the coupled entropy reflects the uncertainty due to the linear uncertainty, thereby tracking the equilibrium portion of the otherwise nonequilibrium process. In contrast, the BGS entropy is dominated by the nonlinear DOF as the coupling increases. The Rényi and Tsallis entropies were incomplete attempts to define a function that could discount the effects of nonlinearity. Their definitions are provided in the Methods subsection \nameref{subsec_Entropies}.

\begin{figure}[ht] 
    \centering
    \includegraphics[width=0.75\linewidth,page=1]{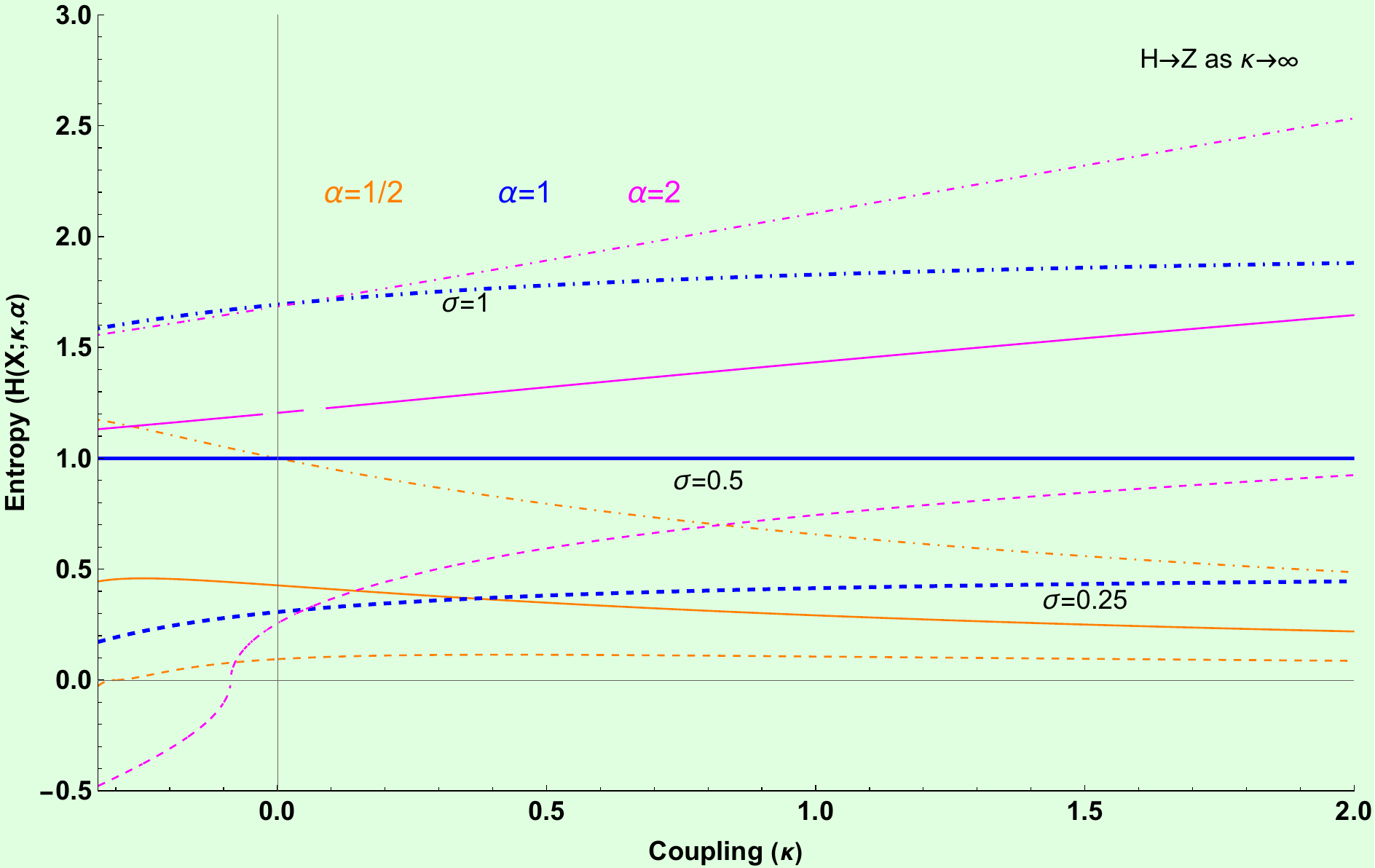}
    \caption{\textbf{The calibrated entropy of coupled stretched exponential.}   The graph shows three values of the variable nonlinearity, $\alpha$, 0.5 - orange, 1 - blue, 2 - magenta, and three values of the scale, $\sigma$, 0.25 - dashed, 0.5 - line, 1.0 - dash-dotted. For $\alpha=1$ (blue), a coupled exponential with $Z=\sigma$, the calibrated entropy is only slightly dependent on the coupling. For $\alpha=2$, a coupled Gaussian, the normalization has a strong dependence on the coupling. With $\alpha=1/2$, there is an inverse dependence on the coupling. In the limit, $\kappa \rightarrow \infty.$, the calibrated entropy converges to $Z$.} 
    \label{fig_CE}
\end{figure}

\hspace{1em}

\begin{figure}[ht]
    \centering
    \includegraphics[width=0.75\linewidth,page=1]{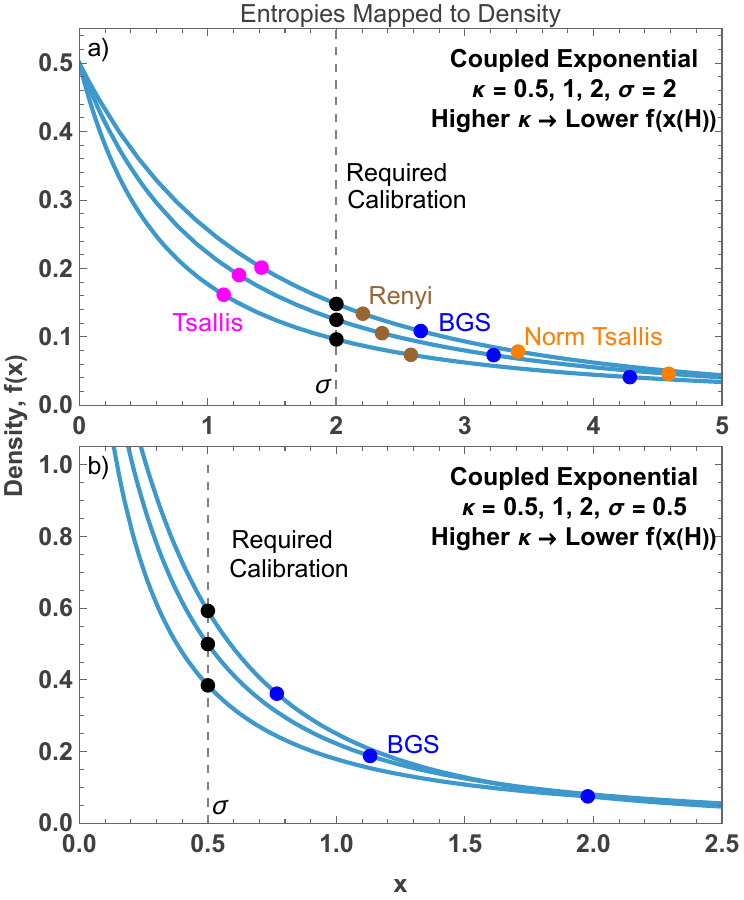}
    \caption{\textbf{Various entropy functions of the coupled exponential distribution are mapped to points on the density.} By treating the entropy of the distribution as an argument of the distribution, a visualization is created regarding the how the measure of entropy relates to the informational scale. Coupling values of 0.5, 1, and 2 are displayed. The required solution, satisfied by the calibrated entropy, aligns along the scale: a) $\sigma=2,$ and b) $\sigma=0.5.$ At the scale, entropy measures the energetic and configurational DOF. BGS entropy also measures the nonlinear DOF. The purpose of a calibrated entropy matched to the nonlinear extensivity is to remove the nonlinear DOF from the measurement. The Rényi, Tsallis, and Normalized Tsallis functions, shown in (a), do not achieve this objective.}
    \label{fig_entDensity}
\end{figure}

\subsection*{The coupled entropy's universal properties}
The classification of the scaling properties for proposed entropy functions was discovered by Hanel and Thurner \cite{hanelComprehensiveClassificationComplex2011, thurnerEntropyNonergodicComplex2012}. They showed that functions that satisfy the first three Shannon-Khinchin axioms satisfy two scaling laws $(c,d)$. Given the focus here on calibrating the measurement of degrees of freedom, the SK axioms are derived properties. The continuity of entropy is a consequence of the functional form of the entropy. The maximality for the uniform distribution is a consequence of the function's concavity. The expandability, which leaves the entropy unchanged for a state with zero probability, is optional. While it is a convenient convention, zero probability states could also be restricted to be outside the domain of the distribution.

The classification is based on the equipartition measurement of the power-law scaling using $\lambda W,$ which has a limit of $\lambda^{1-c},$ and the exponential scaling using $W^{1+a},$ which has a limit of $(1+a)^d.$ To focus on the power, the $c$ scaling parameter is modified to $\tilde{c}=1-c$. In this classification, generalized entropies are characterized by their trace, $g$, and non-trace components, $G$, which simplifies via the asymptotic equipartition property to the expression on the right,
\begin{align} \label{equ_scalefunction}
    F_{G,g}(\mathbf{p})=G\left(\sum_{i=1}^W g(p_i)\right)
    \Rightarrow G\left(Wg(W^{-1})\right).
\end{align}
The $\tilde{c}$ classification of an entropy function, associated with the power-law tail, is determined by the limit of the ratio
\begin{align}
    \lim_{W\rightarrow \infty}\left[
    \frac{G\left(\lambda W 
    g\left(\frac{1}{\lambda W}\right)\right)}{G\left(W 
    g\left(\frac{1}{W}\right)\right)}\right] 
    = \lambda^{\tilde{c}}.
\end{align}
The $d$ classification, associated with the stretched exponential, is based on substituting $\lambda \rightarrow W^a$, isolating a secondary scaling by multiplying by $W^{-a\tilde{c}}$, and taking limit to $W^*=\left\{\begin{matrix}
    \infty & \kappa \ge 0 \\
    0 & -\frac{\alpha}{\alpha + d} < \kappa < 0 
\end{matrix}\right.$,
\begin{align}
    \lim_{W\rightarrow W^*}\left[
    \frac{G\left( W^{1+a} 
    g\left(\frac{1}{ W^{1+a}}\right)\right)}{G\left(W 
    g\left(\frac{1}{W}\right)\right)} W^{-a\tilde{c}}\right] 
    = (1+a)^d.
\end{align}

\begin{theorem}[Universality Classes of the Coupled Entropy]
    \label{th_universality}
The  entropy  $H(X;\kappa,\alpha,d)$, fulfills the full range of the Hanel-Thurner $(\tilde{c},d)$  universality classes with 
\begin{align}
(\tilde{c},d) = \left\{
    \begin{matrix}
        \left(\frac{\kappa}{\alpha+d\kappa},
        \frac{1}{\alpha}\right), & \kappa>0 \\
        \left(0,\frac{1}{\alpha}\right), & \kappa=0 \\
         \left(0,0\right), & -\frac{\alpha}{\alpha+d}< \kappa <0.
    \end{matrix}
    \right.
\end{align}
\end{theorem}
\begin{proof}
The Methods subsection \nameref{subsec_maxProof} reviews the Hanel-Thurner complexity classes and provides the proof \ref{thproof_univ} of the applicability of the entropy function across the comprehensive scope of complexity universality classes. 

\end{proof}
In the sub-exponential domain $(\kappa>0)$ the power-law scaling is $\tilde{c}=\frac{\kappa}{\alpha+d\kappa}$ and the exponential scaling is $d=\sfrac{1}{\alpha}$ The  range of power-law scaling is $0<\tilde{c}<\frac{1}{d}$. The exponential domain has long-range scaling of zero and short-range scaling of $d=\sfrac{1}{\alpha}$. The scaling for the super-exponential domain, $\left(-\frac{\alpha}{\alpha+d} < \kappa < 0\right)$, is the fixed constants $(0,0)$. The reason the scaling is zero rather than a negative power is that the entropy function also includes constants, which will dominate over the negative powers. The scaling properties assure that given a particular asymptotic growth of states, $W(N) \sim \exp_\kappa(\frac{1}{\alpha}N^\alpha)^\frac{\alpha+d\kappa}{\alpha}$ specified by the classification $(\kappa,\alpha,d)$, the entropy will be extensive
\begin{align}
    \lim_{N\rightarrow\infty} 
    \frac{H(W(N),\kappa,\alpha,d)}{N} \sim 
    \frac{\left(\alpha\ln_\kappa
    \left(\exp_\kappa(\frac{1}{\alpha}N^\alpha)^\frac{\alpha+d\kappa}{\alpha}\right)^\frac{\alpha}{\alpha+d\kappa}\right)^\frac{1}{\alpha}}{N}
    \sim 1.
\end{align}
The asymptotic growth of states was expressed in terms of the coupled exponential function for convenience, but extensivity property does not depend on the calibration; other asymptotic functions with the same scaling property are also suitable. 

\subsection*{Thermodynamic and Information Applications}\label{subsec_apps}
Having established the unique, universal function that fulfills a calibration to measure entropy, a plethora of applications will benefit from the improved analytical precision. Central to these applications will be the establishment of a model for NESS systems.  After introducing definitions for temperature, internal energy, and free energy, its application to information systems is described.

Boltzmann famously discovered the microscopic principles for computing the macroscopic property of entropy before the discovery of the atom. Gibbs extended entropy to probabilistic analysis of systems. The Nonlinear Statistical Coupling now provides a framework for the thermodynamics of (NESS) systems \cite{biroNonAdditiveEntropyComposition2022, korbelQuoVadisStochastic2026,lebellacEquilibriumNonequilibriumStatistical2004}. Remarkably, this framework splits the macroscopic free energy (physics) and the microscopic free energy (information) into separate functions connected via a hyperbolic geometry. Thus, Amari's insights regarding the duality of information geometry are key to connecting the thermodynamics and information theory of nonlinear systems. The starting point for this analysis is the requirement that the temperature is defined as the informational scale divided by the Boltzmann constant, $T=\sigma/k_B$, thereby assuring its independence from the nonlinear fluctuations, $\kappa$. For simplicity, this introduction to the NESS thermodynamic relationships focuses on a continuous one-dimensional energy model distributed as a coupled exponential distribution.

The zeroth, first, and second laws of thermodynamics can be viewed as establishing fundamental definitions for temperature, energy, and entropy, respectively. The third law of thermodynamics establishes a minimum entropy at zero degrees Kelvin, which will not be part of this discussion. The laws of thermodynamics assume equilibrium in part because it's important that temperature is an intensive property of a system, while energy and entropy are extensive. Extending these laws to non-equilibrium systems is challenging. For instance, the fluctuations associated with non-equilibrium make it unclear how a single temperature can be defined.

Nevertheless, physicists have sought approaches for defining the intensive and extensive properties of NESS systems. Such a generalization has been a principal objective of the research in nonextensive physics \cite{naudtsGeneralizedThermostatisticsBased2004, naudtsGeneralisedThermostatistics2011}. However, despite significant progress \cite{curadoGeneralAspectsThermodynamical1999, tsallisIntroductionNonextensiveStatistical2004, ferriEquivalenceFourVersions2005,nobreEffectivetemperatureConceptPhysical2012} there remain difficulties in defining a consistent generalized temperature \cite{biroThereTemperatureConceptual2011} that supports a full thermodynamic framework.  I will show that a consistent temperature, independent of the system's nonlinearity, can be defined using the informational scale constraint. Given the independence of the generalized temperature from the non-equilibrium fluctuations, a generalized zeroth law of thermodynamics is formulated. This result motivates further research into understanding the transitivity properties of this definition of temperature, and how the definition can be extended to systems with different coupling values.

\newpage
\begin{definition}[Generalized Zeroth Law of Thermodynamics]
    Let systems A, B, and C each possess the same nonlinear statistical coupling, $\kappa$. If, under weak thermal contact, A and C reach a joint non‑equilibrium steady state, and B and C also reach a joint non‑equilibrium steady state, then A and B, when placed under weak thermal contact, will reach a joint non‑equilibrium steady state with the same $\kappa$.
\end{definition}

The establishment of an intensive temperature, independent of the nonlinear fluctuations, requires a nonlinear relationship with Lagrangian constraint and its conjugate multiplier, which are derived in the Methods subsection \nameref{subsec_maxProof}. In the linear equilibrium case, the constraint, $\sigma$.  is equal to the average internal energy, and its Lagrangian multiplier, $\sfrac{1}{\sigma}$, is the inverse temperature. For the NESS system, the Lagragian multiplier is $\lambda_1^{\text{phys}} = \left(\frac{\sigma}{\epsilon_0}\right)^{-\frac{1}{1+\kappa}}$ for $\alpha=1$, where $\epsilon_0$ is a reference energy level that maintains the dimensionless units of the multiplier. Thus, the physical generalized temperature has a nonlinear relationship with multiplier, $T=\frac{\sigma}{k_B}=\frac{1}{k_B}\lambda_1^{-(1+\kappa)}$.  Likewise, the Legrendre thermodynamic properties determines a nonlinear relationship between the internal energy and the temperature. Using the Legrendre expression, $dU = TdS$ and integrating  gives

\begin{align}
    U_\kappa& = \int T \frac{dS_\kappa}{dT}dT 
    = \int T \frac{k_B d[1+\ln_\frac{\kappa}{1+\kappa}(\frac{k_B T}{\epsilon_0})]}{dT} dT
    \nonumber \\
    &=\int k_B\left(\frac{k_B T}{\epsilon_0}\right)^\frac{\kappa}{1+\kappa} dT
    \nonumber \\
     &= \frac{1+\kappa}{1+2\kappa}\left(\frac{k_B T}{\epsilon_0}\right)^\frac{1+2\kappa}{1+\kappa}\epsilon_0+c=\frac{1}{\iota(\kappa)}\left(\frac{k_B T}{\epsilon_0}\right)^{\iota(\kappa)}\epsilon_0+c,
\end{align}
where $\epsilon_0$ is a reference energy level. In the equilibrium system $(\kappa=0)$, the energy grows linearly with temperature. In contrast, in the NESS systems with $\kappa>0$, the energy must grow nonlinearly with the temperature. This is hypothesized to be due to an increase in the fractal dimension of the system's phase space.

The Lagrangian multipliers for nonlinear systems, which are derived in the Methods subsection \nameref{subsec_maxProof}, can now be understood to include a heat capacity dependent on the fractal phase space. The generalized heat capacity, $C_\kappa$ is equal to the normalization Lagrangian multiplier derived in equation \ref{equ_L0}
\begin{align}
    C_\kappa &= \frac{dU_\kappa}{dT} 
    = \frac{d}{dT} \left[ \frac{1+\kappa}{1+2\kappa} 
    \epsilon_0\left(\frac{k_B T}{\epsilon_0}\right)
    ^{\frac{1+2\kappa}{1+\kappa}} + c\right]
    \nonumber \\
    &=k_B\left(\frac{k_B T}{\epsilon_0}\right)^\frac{\kappa}{1+\kappa}
    \nonumber \\
    \lambda_0^{\text{phys}} &=  
        \frac{1} {(1+\kappa) 
    \int_{x\in\mathcal{X}} 
    f(x)^{1+\frac{\kappa}{1+\kappa}} dx};\  f(x) = \frac{\epsilon_0}{\sigma}\exp_\kappa^{-(1+\kappa)}\left(\frac{\epsilon_0}{\sigma}x\right)
    \nonumber \\
    &= \frac{(1+\kappa)(\sigma/\epsilon_0)^\frac{1+2\kappa}{1+\kappa}(\sigma/\epsilon_0)}{1+\kappa}
    =\left(\frac{\sigma}{\epsilon_0}\right)^\frac{\kappa}{1+\kappa}
    \nonumber \\
    C_\kappa &=k_B \lambda_0.
\end{align}
Thus, the heat capacity of the fractal space minus one, $C_\kappa=(k_BT)^{\iota(\kappa)-1}$. In turn, the Lagrangian multiplier for the temperature constraint is the heat capacity divided by the temperature. Whereas, in equilibrium thermodynamics the constraint and thus the change in entropy is with respect to the energy with the NESS formulation, the constraint and change in the entropy is with respect to the temperature:
\begin{align}
    \lambda_1 &= \frac{dS_\kappa}{d(k_BT)} 
    =\frac{d(1+\ln_\frac{\kappa}{1+\kappa}\sigma)}{d(k_BT)}
    \nonumber \\
    &= (k_BT)^{-\frac{1}{1+\kappa}}
    = \sigma^{-\frac{1}{1+\kappa}},
    \nonumber \\
    \frac{dS_\kappa}{d(k_BT)} &= \frac{dU_\kappa/d(k_BT)}{d(k_BT)} 
    = \frac{C_\kappa}{k_BT}.
\end{align}
And finally, the ratio of the Lagrangian multipliers is $k_B\frac{\lambda_0}{\lambda_1}=\frac{k_B T}{\epsilon_0}=\frac{\sigma}{\epsilon_0}$.

The Legendre transformations for nonlinear systems can be formed in two ways, resulting in complementary expressions for the coupled free energy. To simplify the expression, the reference energy is set to one, $\epsilon_0=1$. With the nonlinearity expressed in the phase space of the internal energy, the linear expressions for the coupled free energy provide a macroscopic perspective of the physical thermodynamics:
\begin{align}
    F_\kappa^P(T) &= U_\kappa(T)-TS_\kappa(T)
    \nonumber \\
    &= \frac{1}{\iota(\kappa)}(k_BT)^{\iota(\kappa)}
    - k_BT (1+\ln_\frac{\kappa}{1+\kappa}(k_BT))
    \nonumber \\
    &= \frac{1+\kappa}{1+2\kappa}(k_BT)^\frac{1+2\kappa}{1+\kappa}
    -k_BT - \frac{1+\kappa}{\kappa}(k_BT)^\frac{1+2\kappa}{1+\kappa} + \frac{1+\kappa}{\kappa}(k_BT)
    \nonumber \\
    &=\frac{(1+\kappa)\kappa-(1+\kappa)(1+2\kappa)}{\kappa(1+2\kappa)}
    (k_BT)^\frac{1+2\kappa}{1+\kappa}
    +\frac{1}{\kappa}(k_BT)
    \nonumber \\
    &= - k_BT\left(\frac{(1+\kappa)^2}{\kappa(1+2\kappa)}
    (k_BT)^\frac{\kappa}{1+\kappa}-\frac{1}{\kappa}\right)
    \nonumber \\
    &= -k_BT\left(\left(\frac{(1+\kappa)^2}{\kappa(1+2\kappa)}
    -\frac{1}{\kappa}\right)
    +\frac{1+\kappa}{1+2\kappa}\ln_\frac{\kappa}{1+\kappa}Z\right)
    \nonumber \\
    &=-\frac{k_BT}{1+2\kappa}\left(\kappa
    +(1+\kappa)\ln_\frac{\kappa}{1+\kappa}Z\right).
\end{align}
The coupled logarithm expression utilizes the fact that $Z=k_BT=\sigma$ for the coupled exponential distribution. The coupled logarithm term is the information component of useful work, while the Z-independent component, $-\frac{\kappa k_B T}{1+2\kappa}$ is an energy debt needed to maintain the long-range correlation of the system. This energy debt is due to the temperature fluctuations. 

The physical coupled free energy is complemented by a microscopic informational coupled free energy, which utilizes the linear measure of the internal energy and a nonlinear function to combine the energy and entropy terms:
\begin{align}
    F_\kappa^I(T) &= \left(\iota(\kappa) U_\kappa(T)\right)^\frac{1}{\iota(\kappa)}
    \ominus_\frac{\kappa}{T} T S_\kappa(T)
    \nonumber\\
    &=\frac{k_BT-k_BT(1+\ln_\frac{\kappa}{1+\kappa}(k_BT))}
    {1+\kappa(1+\ln_\frac{\kappa}{1+\kappa}(k_BT))}
    \nonumber\\
    &=-k_BT \frac{\frac{1+\kappa}{\kappa}Z^\frac{\kappa}{1+\kappa}
    -\frac{1+\kappa}{\kappa}}
    {(1+\kappa)Z^\frac{\kappa}{1+\kappa}}
    \nonumber\\
    &=k_BT \left(\frac{1}{\kappa}Z^{-\frac{\kappa}{1+\kappa}}
    -\frac{1}{\kappa}\right)
    =k_BT \ln_\kappa Z^{-\frac{1}{1+\kappa}}.
\end{align}
From the microscopic perspective of the probabilistic states, the partition function, even if heavy-tailed, can still be divided evenly without the energy debt term that was necessary for macroscopic physical expression of the coupled free energy.

The physical and informational coupled free energy terms are related via a hyperbolic curve. Setting $x=Z^\frac{\kappa}{1+\kappa}$ and using normalized free energies, $\hat{F}\equiv\frac{\kappa}{T}F$, facilitates deriving the identification
\begin{align}
    \hat{F}_\kappa^P(T) &= -\frac{\kappa}{1+2\kappa}\left(\kappa
    +\frac{(1+\kappa)^2}{\kappa}\left(x-1\right)\right)
    \nonumber\\
    &= -\frac{1}{1+2\kappa}
    \left(\kappa^2-1-2\kappa-\kappa^2
    +(1+\kappa)^2x\right)
    \nonumber\\
    &= 1 - \frac{(1+\kappa)^2}{1+2\kappa}x.
\end{align}
Similarly, for the informational coupled free energy, 
\begin{align}
    \hat{F}_\kappa^I(T) &= x^{-1}-1.
\end{align}
Therefore, the two forms of the coupled free energy are related by
\begin{align}
    \hat{F}_\kappa^P(T) = 
    1 - \frac{(1+\kappa)^2}{1+2\kappa}
    \frac{1}{1+ \hat{F}_\kappa^I(T)}
    \nonumber\\
    \left(1-\frac{\kappa}{k_BT}F_\kappa^P(T)\right)
    \left(1+\frac{\kappa}{k_BT}F_\kappa^I(T)\right) =\frac{(1+\kappa)^2}{1+2\kappa},
\end{align}
which is a hyperbolic equation with a curvature defined by $\frac{(1+\kappa)^2}{1+2\kappa}$. The informational coupled free energy is ideally suited to form natural gradients for the microscopic training process of variational inference, while the physical free energy provides a macroscopic measurement of the apparent energy available to do work. These measurements, which in the linear limit $(\kappa \rightarrow0)$ are equal, are mediated by a thermodynamic exchange rate equal to the curvature. The discovery of a hyperbolic geometry for the free energy functions is a prelude to future research on the information geometry of the coupled exponential family.  We anticipate that a hyperbolic statistical manifold utilizing standard calculus will have a conjugate, deformed dually flat manifold, utilizing the coupled calculus.

This proposed thermodynamic framework is consistent with results reported by J. Cleymans, D. Worku, and colleagues \cite{cleymansTsallisDistributionProton2012,cleymansSystematicPropertiesTsallis2013} in their thermodynamic model of high-energy particle collisions and reviewed in Biró's textbook \cite{biroThereTemperatureConceptual2011}. Although they called the distribution Tsallis-B, it uses a modified definition of $q$, such that it is aligned with the coupled exponential distribution, with $q^{CM}=1+\kappa$ and is the form maximized by the coupled entropy. This approach correctly defines temperature, $\sigma=kT$, proportional with the informational scale. Table \ref{tab_energy} shows the temperature and coupling results reported in \cite{cleymansSystematicPropertiesTsallis2013} along with computations of the coupled entropy and the physical and informational coupled free energy. An important objective for future research is whether the numerical values for the entropy and free energy can be compared with independent characteristics, thereby strengthening the physical understanding of complex systems.

\begin{table}[ht]
\centering
\caption{\textbf{The thermodynamic properties of particle collisions.}}
\label{tab_energy}
\begin{tabular}{cccccc}
\toprule
\makecell{\textbf{Beam Energy} \\ \textbf{(TeV)}} & 
\makecell{\textbf{Temperature} \\ \textbf{(MeV)}} & 
\textbf{Coupling} & 
\makecell{\textbf{Calibrated}  \textbf{Entropy}}& 
\makecell{\textbf{Physical} \\ \textbf{Free Energy}} & 
\makecell{\textbf{Info} \\ \textbf{Free Energy}} \\
\midrule
0.54 (UA1) & $77.59 \pm 1.40$ & $0.1175 \pm 0.0014$ & $6.52 \pm 0.11$ &  $-395 \pm 11$ &  $-242 \pm 6$ \\
0.9 (ALICE) & $75.45 \pm 3.18$ & $0.1305 \pm 0.0031$ & $6.61 \pm 0.23$ & $-387 \pm 23$ & $-227 \pm 12$ \\
0.9 (ATLAS) & $83.89 \pm 1.35$ & $0.1217 \pm 0.0007$ & $6.69 \pm 0.06$ & $-439 \pm 8$ & $-263 \pm 5$ \\
2.36 (ATLAS) & $75.79 \pm 4.01$ & $0.1419 \pm 0.0025$ & $6.73 \pm 0.19$ & $-395 \pm 25$ & $-222 \pm 13$\\
7 (ATLAS) & $82.42 \pm 1.30$ & $0.1479 \pm 0.0008$ & $6.94 \pm 0.06$ & $-443 \pm 8$ & $-242 \pm 4$ \\
\bottomrule
\end{tabular}
\end{table}
 
A consistent generalization of thermodynamics is crucial to establishing a coherent model for information and intelligence in non-equilibrium environments. The optimization of the coupled free energy can advance the design of variational inference for extreme environments, impact the active inference model of neuroscience, and improve the design of congested communication systems.

Within neuroscience, the predictive coding model \cite{fristonPredictiveCodingFreeenergy2009, smithRecentAdvancesApplication2021} is a leading candidate for explaining the efficiency and effectiveness of biological intelligence. The model is based on local updates at each neuron, in which error signals from lower layers are compared with predictions from higher layers, with the difference generating new error and prediction signals. Nevertheless, within artificial intelligence, global signaling via backpropagation \cite{millidgePredictiveCodingFuture2022} dominates the training of deep learning algorithms \cite{smithRecentAdvancesApplication2021, rosenbaumRelationshipPredictiveCoding2022}. Variational inference \cite{bleiVariationalInferenceReview2017}, in which probabilistic models are learned and utilized for generative and discriminative capabilities, may be a bridge between these two paradigms. Variational inference, including predictive coding, turns the intractable problem of learning an unknown posterior distribution into an optimization problem of approximating that distribution given a parameterized family of distributions. The negative of the evidence lower bound of the optimization is equivalent to the free energy and consists of a negative log-likelihood, which is the error signal between a generated and original dataset, and the divergence between the posterior and prior latent distribution, which is a regulator maintaining model simplicity.

Expanding the scope of variational inference to non-exponential distributions could be crucial to improving neurological models and accelerating the capability of artificial systems \cite{goertzelActPCGeomScalableOnline2025}. Such methods could be deployed in the heavy-tailed domain to improve inference robustness \cite{caoCoupledVAEImproved2022} or in the compact-support domain \cite{martinsNonextensiveInformationTheoretic2009} to improve efficiency via spareness. Efforts to utilize the Rényi \cite{liRenyiDivergenceVariational2016}, Tsallis \cite{kobayashisQVAEDisentangledRepresentation2020}, and an initial implementation with the coupled entropy \cite{caoCoupledVAEImproved2022}, were limited to the domain of finite variance $(\kappa<1)$ to ensure convergence of the training. Recently, the coupled Variational Autoencoder (CVAE) \cite{nelsonVariationalInferenceOptimized2025} was designed to draw samples from the independent-equals modification to the latent coupled Gaussian distribution, thereby guaranteeing that even the most extreme heavy-tailed distributions with $\kappa \gg 1$ could be trained. This is because the coupling and scale of the latent distribution are transformed by $\frac{\kappa}{1+\kappa}$ and $\frac{\sigma}{\sqrt{1+\kappa}}$, respectively. This design demonstrated improvements of $10\%$ and $25\%$ in measures of Learned Perceptual Image Patch Similarity and Multiscale Structural Similarity, respectively, across coupling values ranging from $10^{-5} \text{ to } 10^5.$ The code base for the nonlinear statistical coupling methods and the CVAE algorithm is referenced in Methods subsection \nameref{subsec_avail}.

A further development in neuroscience is the extension of variational to active inference \cite{parrActiveInferenceFree2022}, by which planning, decision-making, and actions are incorporated into the same theoretical framework. Traditionally, decisions and actions have been modeled as optimizing utility functions. An important example of this is the reinforcement algorithm in which an optimal expected utility is chosen from a variety of plans. Friston \cite{fristonFreeEnergyPrinciple2006} showed that normalization of the utility is equivalent to a prior distribution weighting an agent's preferences. An agent, such as a cell, is separated from its environment, such as a fluid host, by a Markov Blanket\cite{bruinebergEmperorsNewMarkov2022}, the cell membrane. Importantly, the Markov Blanket must provide a conditionally independent boundary between the agent and its surroundings; however, this condition assumes equilibrium. The coupled free energy function enables the definition of a pseudo-Markov Blanket in which non-equilibrium fluctuations that transverse the boundary could be included in the model. The transition from an exponential equilibrium model to a non-exponential non-equilibrium model would still require that the agent/environment cross terms of the correlation matrix be zero, but now a nonlinear dependence between the boundary would be permitted. While developed for neuroscience, the active inference model is a general intelligence model that can be applied to variety of agent types navigating turbulent environments. 

In parallel with the complexity of AI systems, modern communications systems, such as 5G wireless and cognitive radio networks, are facing increased complexity as dense traffic creates heavy-tailed interference. In contrast to a Gaussian noise channel, interference fluctuates in intensity depending on the instantaneous use of the channel. For example, Clavier et al. measured \cite{clavierExperimentalEvidenceHeavy2021} tail shapes above $\kappa>2,$ in 5G networks, indicative of quite extreme fluctuations. In cognitive radio designs \cite{munozRenyiEntropyBasedSpectrum2020}, secondary users must adaptively switch to open channels to guarantee non-interference with primary users. These applications cannot use standard metrics such as the mean-square error or the Shannon entropy, since the outliers invalidate the metrics' foundational assumptions.

To date, the Rényi entropy has been the most common generalized entropy function used in information systems \cite{principeInformationTheoreticLearning2010, hildAnalysisEntropyEstimators2006, erdogmusAdaptiveBlindDeconvolution2004}. The probability analysis of the entropy function presented in Methods subsection \nameref{subsec_Entropies} provides insight regarding why the Rényi entropy, despite being non-optimal, is nevertheless effective in these circumstances. A high-interference channel characterized by a shape between $1<\kappa/2<2, $ such as the distributions shown in Figure \ref{fig_entropies}, has a Shannon entropy measure dominated by the nonlinear DOF. The Rényi function dramatically reduces this influence if its index matches this tail shape via $q(\kappa,\alpha,d)=1+\frac{\alpha\kappa}{\alpha+d\kappa}$. The calibrated entropy framework is expected to further improve communication design in the presence of heavy-tailed noise in two ways. First, the connection between multiplicative noise, heavy-tailed tail shape, and the index of the coupled entropy is made explicit, which will make the design choices more explicit. Second, the full discounting of the nonlinear DOF requires the calibrated entropy.

\section*{Discussion}\label{sec_disc}

A precise, calibrated measurement of entropy as the degrees of freedom of system across the spectrum of nonlinear, complexity classifications is defined and proven. The definition is ground in four independent measurements of a systems characteristics:
\begin{itemize}
    \item $\kappa$ - the functional nonlinearity that creates long-range dependence
    \item $\alpha$ - the variable nonlinear stretch that creates short-range dependence
    \item $d$ - the number of dimensions of the system
    \item $\sigma$ - the informational scale of the system.
\end{itemize}

Proving the uniqueness of the informational scale and its central role in correctly calibrating the measurement of entropy is the principal contribution of this paper. With this insight, the coupled entropy is proven to be a unique, universal measure of entropy; all other proposed functions fail in some way to satisfy the three axioms needed to calibrate entropy's measurement of degrees of freedom:
\begin{itemize}
    \item specification of the universality classes $(\kappa,\alpha)$ (plus dimensions $d$) for nonlinear, systems which include exponential $(\kappa=0,\alpha=1)$, stretched exponential $(\kappa=0, \alpha>0)$, stretched power-law $(\kappa>0, \alpha>0)$, and stretched super-exponential $(-\frac{\alpha}{d+\alpha}<\kappa<0, \alpha>0)$;
    \item specification of the granularity via the informational scale $(\sigma)$; and
    \item the requirement that calibrated growth of states be based on a four-parameter (including dimensions) nonlinear ordinary differential equation.
\end{itemize}

This approach to defining entropy is a significant departure from the traditional Shannon-Khinchin axioms; however, it is truer to the original intentions of Shannon to define entropy as a monotonic measure (logarithm) of the number of potential messages, and to the efforts of Boltzmann and Gibbs to measure the degrees of freedom of physical systems. The phrase ``axioms'' was not used by either Shannon or Khinchin \cite{shannonMathematicalTheoryCommunication1948, khinchinMathematicalFoundationsInformation1957}; rather, this term was canonized by scientists seeking to broaden the definition for nonlinear, complex systems \cite{aczelMEASURESINFORMATIONTHEIR1975, sharmaAxiomaticCharacterizationInformationtheoretic1974}. By focusing exclusively on the mathematical question of composition beyond additivity, the proposals for a generalized entropy became ungrounded from the physical requirement to measure the energetic and configurational degrees of freedom. 

The specification of a particular complexity class can be viewed as a constraint on the nonlinearity of the system. The BGS entropy is still a relevant measure for all systems, but its count of degrees of freedom includes the impact of nonlinearity. The purpose of an entropy measure tuned to the complexity class of the particular system is to discount the nonlinear degrees of freedom and instead count the energetic DOF and the generalized configurational DOF. That said, a particular nonlinear constraint on the entropy can also be applied to any system, providing its own unique measurement of the system's characteristics. The spectrum of these entropy measurements may encode the full information of a system, though that is a hypothesis which will have to be evaluated.

The only function that satisfies the proposed calibration requirement is the coupled entropy. This calibrated entropy for complex systems has three components:
\begin{enumerate}[left=\parindent]
    \item a generalized logarithm, inverse to the structure of the coupled stretched exponential function;
    \item an average over the independent-equals distribution;
    \item a power outside the average equal to the inverse of the stretching parameter.
\end{enumerate}
The generalized logarithm and the independent-equals distribution complement each other by increasing (lowering) the cost of outliers and decreasing (increasing) the sampling of outliers for positive (negative) values of coupling $(\kappa)$, respectively. The power equal to the inverse of the stretching parameter ensures the correct thermodynamic extensivity. The result is an elegant solution for the entropy of the maximizing distribution, maintaining the measurement of the energetic and configurational degrees of freedom
\begin{align}
    H (\mathbf{X}_\text{CSED};\kappa,\alpha,d) &= \left( 
    \underbrace{d}_{\text{energetic DOF}} 
    + \underbrace{\ln_\frac{\kappa}{\alpha+d\kappa}
    Z_\kappa(\sigma,\alpha,d)^\alpha}_{\text{configurational DOF}} \right)
    ^\frac{1}{\alpha}.
\end{align}

The uniqueness of this calibrated entropy is complemented by its comprehensive applicability to the universality classes for complex, nonlinear systems. These systems are characterized by the variable nonlinearity using the stretching parameter $(\alpha)$, and the functional nonlinearity using the coupling $(\kappa)$. This classification of complex statistical systems was established by Hanel-Thurner, and was used to define a method of measuring an entropy function's power-law $c$ and exponential $d$ scaling properties. The proposed entropy function satisfies the comprehensive range of $(c,d)$ scaling with the relationship to the primary properties of
\begin{align}
(\tilde{c},d) = (1-c, d)=\left\{
    \begin{matrix}
        \left(\frac{\kappa}{\alpha+d\kappa},
        \frac{1}{\alpha}\right), & \kappa>0 \\
        \left(0,\frac{1}{\alpha}\right), & \kappa=0 \\
         \left(0,0\right), & -\frac{\alpha}{\alpha+d}< \kappa <0.
    \end{matrix}
    \right.
\end{align}

The proof of the proposed entropy's uniqueness and universality provides a foundation for a deeper understanding of complex phenomena and thereby the creation of new applications for the improvement of living and artificial systems. Already, the development of robust variational inference has been demonstrated, whereby models with near-infinite Shannon entropy can be trained without losing stability \cite{nelsonVariationalInferenceOptimized2025}. Research programs are underway to properly characterize the thermodynamics of non-equilibrium steady-state systems and to build intelligent agents using active inference with the coupled free energy that can survive and thrive in these complex systems.

The design of systems that can respond adaptively to the complexities of a natural environment is fundamental to modern science and engineering. The heavy-tailed interference of dense, adaptive communications networks is indicative of these challenges. AI developers aspire to mimic and integrate with biological intelligence \cite{ikleArtificialGeneralIntelligence2025}; city planners seek ways of assuring that human habitats can flourish from and contribute back to natural habitats \cite{whiteModelingCitiesRegions2015}; and modern finance requires management of complex risks impacted by political and environmental instabilities \cite{glennGlobalGovernanceTransition2025}. Each of these domains, and many others, requires detailed modeling of the nonlinearities that create fluctuating noise patterns. The unique, universal entropy function for complex adaptive systems, fulfilled by the nonlinear statistical coupling framework, will be a crucial tool in the development of these precise analytical tools.

\section*{Methods}\label{sec_methods}

\subsection*{The Coupled Exponential Family and its algebra}\label{subsec_CEF}
The coupled exponential family (CEF) of distributions is foundational to the analysis of complex systems, given its origin from nonlinear dynamics. This family generalizes the exponential family $(\kappa=0)$ for heavy-tailed $(\kappa>0)$ and compact-support $(-\frac{\alpha}{d}<\kappa<0)$ distributions, while assuring that each parameter in the family has a precise interpretation for mathematical physics. The lower bound of the range is the uniform distribution below which the function cannot be normalized. The broadest definition of the family comes from a multivariate representation using information geometry. The definition is provided for continuous variables, though it extends to discrete distributions.
\begin{definition}[The Coupled Exponential Family]
    Given a $d$-dimensional random variable $\mathbf{X}_d \sim f_\kappa(\mathbf{x}_d;\boldsymbol{\theta},\alpha)$, within the coupled exponential family, its definition is:
    \begin{align}
        f_\kappa(\mathbf{x}_d;\boldsymbol{\theta},\alpha) &= \frac{h(\mathbf{x})}{Z_\kappa(\boldsymbol{\theta},\alpha,d)}
        \exp_\kappa
        \left[\boldsymbol{\eta}(\boldsymbol{\theta}) \cdot \mathbf{T}(\mathbf{x})\right]
        ^{-\frac{\alpha+d\kappa}{\alpha}}
        \nonumber \\
        &= \exp_\kappa
        \left[\boldsymbol{\eta}(\boldsymbol{\theta}) \cdot \mathbf{T}(\mathbf{x}) \oplus_\kappa \ln_\kappa\left(h(\mathbf{x})\right)^{-\frac{1}{1+d\kappa/\alpha}}
        \oplus_\kappa \ln_\kappa\left(Z_\kappa(\boldsymbol{\theta},\alpha,d)\right)^\frac{1}{1+d\kappa/\alpha}
        \right]^{-\frac{\alpha+d\kappa}{\alpha}}, \label{equ_cefPDF}
    \end{align}
where $\alpha$ is the highest power of the vector of statistics $\mathbf{T}(\mathbf{x})$, $\boldsymbol{\eta}(\boldsymbol{\theta})$ is the vector of parameters, $Z_\kappa$ is the normalization or partition function, $h(\mathbf{x})$ is the base measure, and $[a]_+=\max(0,a)$. For distributions, which must be positive, the coupled exponential function includes a cut-off, $\exp_\kappa x \equiv (1+\kappa x)_+^\frac{1}{\kappa}$, where $(a)_+\equiv \max(0,a)$.
\end{definition}
Notice that this definition of the coupled exponential family preserves the definition of the partition function, shown in Figure \ref{fig_CG_Norm}, as the integral of the other terms:
\begin{align}
    Z_\kappa(\boldsymbol{\theta},\alpha,d)= \int_{\mathbf{X}} 
    h(\mathbf{x}) \exp_\kappa
        \left[\boldsymbol{\eta}(\boldsymbol{\theta}) \cdot \mathbf{T}(\mathbf{x})
        \right]^{-\frac{\alpha+d\kappa}{\alpha}}
        \mathrm{d}\mathbf{x}
\end{align}
This property is assured by bringing $Z$ and possibly $h$ into the coupled exponential function via the coupled sum of the coupled logarithm. This contrasts with the definition for the $q$-exponential family defined by Ohara and Amari \cite{oharaGeometryDistributionsAssociated2007, amariGeometryQExponentialFamily2011}:
\begin{align}
    f_q(\mathbf{x}_d;\boldsymbol{\theta}) &=h(\mathbf{x})
    \exp_q
        \left[\boldsymbol{\eta}(\boldsymbol{\theta}) \cdot \mathbf{T}(\mathbf{x})
        - \ln_qZ_q^{-1}(\boldsymbol{\theta})
        \right]
        \nonumber\\
        &= h(\mathbf{x})
    \exp_q
        \left[\boldsymbol{\eta}(\boldsymbol{\theta}) \cdot \mathbf{T}(\mathbf{x})
        \right]
        \otimes_q
    \exp_q
        \left[ \ln_{2-q} Z_q(\boldsymbol{\theta})
        \right]
        \nonumber\\
        \exp_q(x) &\equiv (1 + (1-q) x)_+^\frac{1}{1-q}; \quad 
        \ln_q(x) \equiv \frac{1}{1-q}\left(x^{1-q} - 1\right), \label{equ_qexpln}
\end{align}
in which the term $Z_q$ is not the normalization and thus loses crucial properties regarding the partition function.
\begin{figure*}[ht]
    \centering
    \includegraphics[width=0.5\linewidth,page=1]{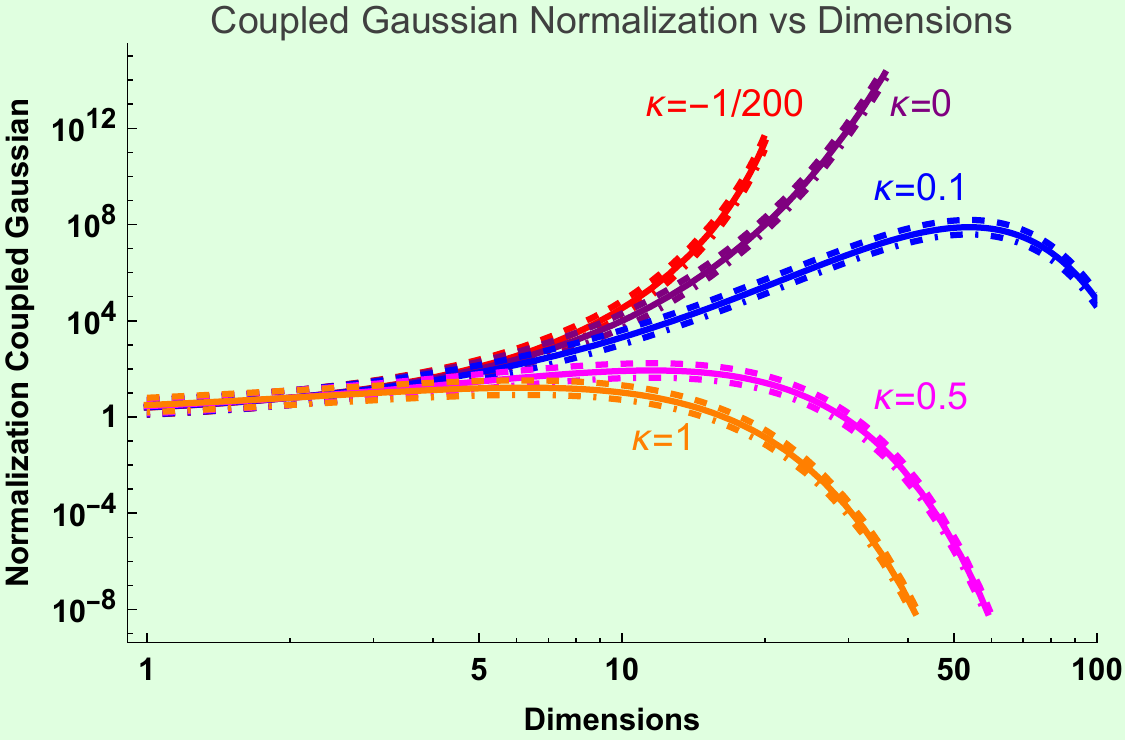}
    \hfill
    \caption{\textbf{The Coupled Gaussian normalization as a function of dimensions.} The nonlinear coupling has a sigificant impact on the normalization beyond 10 dimensions.}
    \label{fig_CG_Norm}
\end{figure*}

    Two broad groups of the CEF are when the base $h(\mathbf{x})=1$ (the coupled stretched exponential distribution) and when the survival function is only $\exp_{\kappa}^{-1}
        \left[\boldsymbol{\eta}(\boldsymbol{\theta}) \cdot \mathbf{T}(\mathbf{x})\right]$ (the coupled Weibull distribution). In this paper the focus is on one statistic, $\mathbf{T(x)}=\frac{x^\alpha}{\alpha}$ and its associated parameter $\boldsymbol{\eta(\theta)}= \frac{1}{\sigma^\alpha}$. To simplify the multivariate expression, a radial variable with the stretching parameter power $\sfrac{\alpha}{2}$ applied element-wise to the vectors and matrix (signified with $\circ$) is defined as $r^\alpha\equiv ((\mathbf{x}-\boldsymbol{\mu})^{\circ \sfrac{\alpha}{2} \top} 
(\boldsymbol{\Sigma}^{\circ \sfrac{\alpha}{2}})^{-1}(\mathbf{x}-\boldsymbol{\mu})^{\circ \sfrac{\alpha}{2}})\geq 0$. For the compact-support domain, the statistical domain is $-\frac{\alpha}{d}<\kappa<0$. The domain of concave entropies will further restrict this domain to $-\frac{\alpha}{\alpha+d}<\kappa<0$. Each distribution includes a definition for the Survival Function (SF) and/or the Probability Density Function (PDF). The PDF is the  $(-1)^dd^{th}$ derivative of the SF.
    \begin{definition}{Coupled Stretched Exponential Distribution}
        \label{def_CSED} \\
        PDF:
        \begin{align}
f_\kappa^{\text{Stretch}}(r) &\equiv \frac{1}{Z_\kappa}\left(1+\kappa
\frac{r^\alpha}{\alpha}\right)_+
^{-\frac{\alpha+d\kappa}{\alpha\kappa}} \nonumber \\ \label{equ_csePDF}
&\equiv \frac{1}{Z_\kappa}\exp_{\kappa}^{-(1+d\kappa/\alpha)}
\left(\frac{r^\alpha}{\alpha}\right) \nonumber \\
Z_\kappa&=\frac{2\pi^\frac{d}{2}}{\Gamma\left(\frac{d}{2}\right)}
|\boldsymbol{\Sigma}|^\frac{1}{2}
\left\{
\begin{matrix}
    \frac{1}{\alpha}\left(\frac{\kappa}{\alpha}\right)^{-d/\alpha}
    B\left(d/\alpha,\frac{1}{\kappa}\right)& \kappa > 0 \\
    \alpha^{\left(d/\alpha-1\right)}\Gamma\left(d/\alpha\right) & \kappa=0 \\
    \frac{1}{\alpha}(-\frac{\kappa}{\alpha})^{-d/\alpha} B\left(d/\alpha,1-\frac{\alpha+d\kappa}{\alpha\kappa}\right)& -\frac{\alpha}{d}<\kappa<0 .
\end{matrix}
\right. \\
\nonumber \end{align}
        SF = 1 - CDF:
        \begin{align}
            S_\kappa^{\mathrm{Stretch}}(r)=1-F_\kappa(r) 
            &\equiv \left\{
    \begin{matrix}
    1-I_z\left(d/\alpha,\frac{1}{\kappa}\right)
     & \kappa > 0 \\
      \frac{\Gamma\left(d/\alpha,r^\alpha/\alpha\right)}{\Gamma\left(d/\alpha\right)} & \kappa = 0 \\
     1-I_z\left(d/\alpha,1-\frac{\alpha+d\kappa}{\alpha\kappa}\right)
     & -\frac{\alpha}{d}<\kappa< 0 
    \end{matrix}\right. \label{equ_cseSF} \\ 
    z&=\left\{
    \begin{matrix}
        \frac{\kappa r^\alpha}{\alpha+\kappa r^\alpha} & \kappa > 0 \nonumber \\
        -\kappa r^\alpha/\alpha & -\frac{\alpha}{d}<\kappa<0
    \end{matrix}\right. \nonumber \\
    I_z&=\frac{B_z(a,b)}{B(a,b)};\text{  Regularized Incomplete Beta Function} \nonumber
        \end{align}
        Special Cases:
        $\alpha=1$: Coupled Exponential Distribution
        $\alpha=2$: Coupled Gaussian Distribution

    \end{definition}

        \begin{definition}[Coupled Weibull Distribution]
        The survival function has just the coupled exponential structure. \\
        SF = 1 - CDF:
        \begin{align}
            S_\kappa^{\mathrm{Weibull}}(r)=1-F_\kappa(r) \equiv
            \left\{\begin{matrix}
                \left(1+\kappa \frac{r^\alpha}{\alpha}\right)_+^{-\frac{\alpha}{\kappa}}
                & \kappa > -\frac{1}{d}; \ \kappa \neq 0\\ 
                \exp\left(-\frac{r^\alpha}{\alpha}\right) & \kappa = 0
            \end{matrix}\right. \label{equ_cwSF}
        \end{align} 
The PDF includes angular dependence in the Jacobian, $J(r,\boldsymbol{\theta)}$,
        \begin{align}
            f_R(r;\kappa,2) &= - \frac{\mathrm{d}S_R(r)}{\mathrm{d}r};\ 
            f_\mathbf{X}(\mathbf{x};\kappa,2) = \frac{\Gamma(d/2)
            f_R(r(\mathbf{x});\kappa,2)}
            {J(r;\boldsymbol{\theta})}; 
            \nonumber\\
            J(r;\boldsymbol{\theta})
            &= \left(\frac{2}{\alpha}\right)^d
            |\boldsymbol{\Sigma}^{\circ \sfrac{\alpha}{2}}|^\frac{1}{2}
            r^{d-1} \prod_{i=1}^d |x_i - \mu_i|^\frac{2-\alpha}{2}
            \nonumber\\
            f_\mathbf{X}^{\mathrm{Weibull}}
            (\mathbf{x}; \kappa,2) 
            &\equiv \left\{\begin{matrix}
                \frac{r^{\alpha-d} 
                \prod_{i=1}^d |x_i - \mu_i|
                ^\frac{\alpha-2}{2}}{ Z}
                \left(1+\kappa \frac{r(\mathbf{x})
                ^\alpha}{\alpha}\right)_+
                ^{-\frac{1}{\kappa}-1}
                & \kappa > -\frac{\alpha}{d}; \ \kappa \neq 0\\
                 \frac{r(\mathbf{x})^{ \alpha-d}
                 \prod_{i=1}^d |x_i - \mu_i|
                ^\frac{\alpha-2}{2}}{ Z}
                \exp\left(-\frac{r(\mathbf{x})
                ^\alpha}{\alpha}\right) & \kappa = 0
            \end{matrix}\right. \label{equ_cwPDF} 
            \nonumber\\ 
            Z &= \left(\frac{2}{\alpha}\right)^d
            |\boldsymbol{\Sigma}^{\circ \sfrac{\alpha}
            {2}}|^\frac{1}{2}
        \end{align}

        Special Cases:
        $\alpha=1$: Coupled Exponential Distribution
        $\alpha=2$: Coupled Rayleigh Distribution

    \end{definition}

        The two most important distributions, which will be utilized frequently, are the coupled exponential and the coupled Gaussian distributions. The one-dimensional coupled exponential distribution is unique in retaining just the coupled exponential structure for both the SF and the PDF.

\begin{definition}[Coupled Exponential Distribution]\label{def_CED}
For location $\mu$, scale $\sigma$, shape $\alpha=1$,  and $d=1,$ the survival function for the one-sided coupled exponential distribution is:
\begin{align} 
\text{SF: } S_\kappa^\text{exp}(x;\mu,\sigma,1)
\equiv \left(1 + \kappa\left(\frac{x-\mu}{\sigma}\right)\right)_+^{-\frac{1}{\kappa}}
\equiv\exp_{\kappa}^{-1}
\left(\frac{x-\mu}{\sigma}\right); \\ 
x>\mu; \ \kappa > -1 \nonumber
\end{align}

The probability density function, $-\frac{\mathrm{d}S}{\mathrm{d}x},$ for the one-sided coupled exponential distribution is:
\begin{align} 
\text{PDF: }f_\kappa^{\text{exp}}(x;\mu,\sigma,1) &\equiv \frac{1}{\sigma}\left(1+\kappa\frac{x-\mu}{\sigma}\right)_+^{-\frac{1+\kappa}{\kappa}}\label{equ_cexpdist} 
\equiv \frac{1}{\sigma}\exp_\kappa^{-(1+\kappa)}\left(\frac{x-\mu}{\sigma}\right), \\  \label{equ_cepdf}
\text{for } x\geq \mu, \kappa>-1.\nonumber
\end{align}
\end{definition}

\begin{definition}[Coupled Gaussian Distribution]\label{def_CGD}
For location $\mu$, scale $\sigma$, stretching parameter $\alpha=2$, and $d=1,$ the survival function for the coupled Gaussian distribution is:
 \begin{align}
            S_\kappa(x;\mu,\sigma,2)=1-F_\kappa(x) 
            &\equiv \left\{
    \begin{matrix}
    \left(1-I_z\left(\frac{1}{2},\frac{1}{\kappa}\right)\right)
     & \kappa > 0 \\
      \frac{1}{2}\left(1-\text{erf}\left(\frac{x-\mu}{\sqrt{2}\sigma}\right)\right) & \kappa = 0 \\
     \left(1-I_z\left(\frac{1}{2},\frac{-1+\kappa/2}{\kappa}\right)\right)
     & -1<\kappa< 0 
    \end{matrix}\right.\\
    z&=\left\{
    \begin{matrix}
        \frac{\kappa \frac{1}{2}\left(\frac{x-\mu}{\sigma}\right)^{2}}{1+\kappa
        \frac{1}{2} \left(\frac{x-\mu}{\sigma}\right)^{2}} & \kappa > 0 \nonumber \\
        -\kappa \frac{1}{2}\left(\frac{x-\mu}{\sigma}\right)^{2} & -2<\kappa<0
    \end{matrix}\right. \nonumber \\
    I_z&=\frac{B_z(a,b)}{B(a,b)};\text{  Regularized Incomplete Beta Function} \nonumber
        \end{align}

The probability density function, $-\frac{\mathrm{d}S}{\mathrm{d}x},$  for the coupled Gaussian distribution is:
\begin{align} 
f_\kappa^{\text{exp}}(x;\mu,\sigma,2) &\equiv \frac{1}{Z}\left(1+\kappa\frac{1}{2}
\left(\frac{x-\mu}{\sigma}\right)^2\right)_+^{-\frac{2+\kappa}{2\kappa}} 
\nonumber \\
&\equiv \frac{1}{Z}\exp_{\kappa}^{-(1+\kappa/2)}\left(\frac{1}{2}\left(\frac{x-\mu}{\sigma}\right)^2\right), 
\ \text{for } \kappa>-2. \label{equ_cgpdf}
\end{align}
\end{definition}

\begin{lemma}[CEF is stationary distribution of multiplicative noise ]
    Given a Stratonovich stochastic differential equation with independent sources of additive and multiplicative noise,
    \begin{align}
        dX_t = A(X_t)\,dt + \sqrt{\alpha\sigma^{\alpha}}\circ dW_t^{(1)} + \sqrt{\kappa}\,|X_t|^{\alpha/2}\circ dW_t^{(2)},
\qquad \alpha>0,\;\sigma>0,\;\kappa\ge0,
    \end{align}
then the stationary distributions are members of the coupled exponential family,
\begin{align} \label{equ_FPCSED}
    p(x) = \frac{h(x)}{Z_\kappa (\alpha,d)}\,\left(1 + \kappa\frac{|x|
    ^{\alpha}}{\alpha\sigma^{\alpha}}\right)_+
    ^{-\frac{\alpha+\kappa}{\alpha \kappa}},
\end{align}
where the base $h(x)$ is only one if the drift $A(X_t)$ is balanced with the diffusion

\[
h(x) \propto \exp\!\left( \int^x \frac{2A(y) + \bigl(\alpha+\kappa - \frac{\kappa\alpha}{2}\bigr)y^{\alpha-1}}{\alpha\sigma^{\alpha}+\kappa y^{\alpha}}\,dy \right),
\]
and
\[
Z_\kappa(\alpha,d) = \int^x
h(x)\,\left(1 + \kappa\frac{|x|^{\alpha}}{\alpha\sigma^{\alpha}}\right)_+
^{-\frac{\alpha+\kappa}{\alpha \kappa}}
dx.
\]

\end{lemma}
\begin{proof}
    Given the Stratonovich SDE, the conversion to an Itô SDE involves an additional drift term produced from the multiplicative noise term, $g(x)=\sqrt{\kappa}\,|X_t|^{\alpha/2}$, equal to $\frac{1}{2}g'(x)g(x)$. So the Itô SDE is
    \begin{align}
         dX_t = \left(A(X_t)+\frac{\kappa\alpha}{4}|X_t|^{\alpha-1}\text{sgn}(X_t)\right)\,dt 
         + \sqrt{\alpha\sigma^{\alpha}} dW_t^{(1)} 
         + \sqrt{\kappa}\,|X_t|^{\alpha/2} dW_t^{(2)},
\qquad \alpha>0,\;\sigma>0,\;\kappa\ge0,
    \end{align}
with the drift and diffusion terms for the Fokker-Planck equation being
\begin{align}
    \mu_{\text{Itô}} &= A(x)+\frac{\kappa\alpha}{4}|x|^{\alpha-1}\text{sgn}(x)
    \nonumber\\
    \sigma_{\text{eff}} &=  \alpha\sigma^{\alpha} + \kappa |x|^{\alpha}.
\end{align}
The Fokker-Planck equation and its solution are then
\begin{align}
    \frac12\frac{d}{dx}\Bigl[\sigma_{\text{eff}}(x)p(x)\Bigr] - \mu_{\text{Itô}}(x)p(x)=0
\end{align}
which has equation \ref{equ_FPCSED} as a solution.
\end{proof}

The expression and manipulation of nonlinear relationships is facilitated by a pseudo-algebra \cite{borgesPossibleDeformedAlgebra2004,nelsonAverageUncertaintySystems2017} that generalizes the operations of algebra. Raising the coupled exponential function, $\exp_\kappa(x) \equiv (1+\kappa x)^\frac{1}{\kappa}$, to a power is equivalent to modifying the coupling,  $\exp_\kappa(x)^a \equiv \exp_{\kappa/a}(a x)\equiv(1+\kappa x)^\frac{a}{\kappa}$. The inverse of the coupled exponential function is the coupled logarithm, $\ln_{\kappa} x\equiv\frac{1}{\kappa}(x^\kappa-1)$, which has scaling properties $\frac{1}{a}\ln_\kappa x\equiv \ln_{a\kappa} x^\frac{1}{a}\equiv\frac{1}{a\kappa}(x^\kappa-1)$. These generalized functions have the following pseudo-algebraic properties:
\begin{align}
    \exp_\kappa(x+y)&=\exp_\kappa(x) \otimes_\kappa \exp_\kappa(y);\quad 
    A \otimes_\kappa B \equiv (A^\kappa + B^\kappa-1)^\frac{1}{\kappa} \label{equ_cprod}\\
    \ln_\kappa (xy) &= \ln_\kappa (x) \oplus_\kappa \ln_\kappa (y); \quad 
    A \oplus_\kappa B = A + B + \kappa AB. \label{equ_csum}
\end{align}
The coupled subtraction function follows from the coupled sum via the property:
\begin{align}
    A\oplus_\kappa(\ominus_\kappa A)&=0, \text{ therefore, }
    \ominus_\kappa A = \frac{-A}{1+\kappa A}, 
    \text{ for } A \neq \frac{-1}{\kappa} \\
    A\ominus_\kappa B&=A\oplus_\kappa(\ominus_\kappa B)
    = \frac{A-B}{1+\kappa B}, 
    \text{ for } B \neq \frac{-1}{\kappa}.
\end{align}
The coupled subtraction generalizes a common property of the coupled exponential and coupled logarithm :
\begin{align}
    \exp_\kappa \left(\ominus_\kappa x\right) 
    &= \left(1+\frac{-\kappa x}{1+\kappa x}\right)^\frac{1}{\kappa}
    = \exp_\kappa (x)^{-1}\\
    \ominus_\kappa \ln_\kappa x &=\frac{-\frac{1}{\kappa}\left(x^\kappa-1\right)}{x^\kappa}
    =\ln_\kappa x^{-1}.
\end{align}
The coupled division function follows from the coupled product, via the property:
\begin{align}
     A\otimes_\kappa(\oslash_\kappa A)=1, \text{ therefore, }
    \oslash_\kappa A = (2-A^\kappa)_+^\frac{1}{\kappa}, 
    \text{ for } A > 0 \\
    A\oslash_\kappa B=A\otimes_\kappa(\oslash_\kappa B)
    = (A^\kappa - B^\kappa+1)_+^\frac{1}{\kappa} , 
    \text{ for } A,B >0.
\end{align}
The methods extend to summations and products over $N$ variables as follows:
\begin{align}
\coupledsum{i=1}{N}{\kappa} x_i &= x_1 \oplus_\kappa x_2 
    ... \oplus_\kappa x_N \nonumber \\ 
    &= \ln_\kappa \left( \prod_{i=1}^N \exp_\kappa x_i \right)\\
\coupledprod{i=1}{N}{\kappa} x_i &= x_1 \otimes_\kappa x_2 
    ... \otimes_\kappa x_N \nonumber \\
    &=\exp_\kappa\left(\sum_{i=1}^N\ln_\kappa x_i\right).
\end{align}
Finally, the coupled power is:
\begin{align}
    x^{\otimes_\kappa^N} &= (Nx^\kappa-(N-1))^\frac{1}{\kappa} \nonumber \\
    &= \coupledprod{i=1}{N}{\kappa} x = x \otimes_\kappa x ... 
    \otimes_\kappa x \nonumber \\
    &= (2x^\kappa -1) \otimes_\kappa x ... 
    \otimes_\kappa x = (3x^\kappa -2)^\frac{1}{\kappa} ... \otimes_\kappa x 
\end{align}
While the coupled algebra notation may seem cumbersome, it proves to be a powerful method to conceptualize important relationships within the analysis of complex systems. For instance, the discovery of a uniqueness proof for the entropy of nonlinear systems was delayed by ten years in part because I did not simplify the expression for the coupled entropy of the coupled exponential distribution when it was first computed in 2015. Here is the distinction in forms, which the paper will develop and expand upon:
\begin{equation}
    \frac{1}{\kappa}\sigma^\frac{\kappa}{1+\kappa}
    + \sigma^\frac{\kappa}{1+\kappa} 
    -\frac{1}{\kappa}=1+\ln_\frac{\kappa}{1+\kappa} \sigma.
\end{equation}

\subsection*{Maximization and other entropy proofs}\label{subsec_maxProof}

A core methodology for information thermodynamics is the derivation of the distribution that maximizes an entropy function using the Lagrangian function. In this section, the theorem is proven that the coupled stretched exponential distributions maximizes the coupled entropy given the independent-equals moment as a constraint. While this theorem focuses on the shape-scale distributions, a variety of other models are possible by including additional constraints. Without loss of generality, the dimension is set to one, and the location is set to zero.

\begin{lemma}[Entropy maximized by the coupled stretched exponential distribution]
    Given the entropy function
    \begin{equation} \frac{1}{\alpha}H(f(x);\kappa,\alpha)^\alpha = \int_{x\in\mathcal{X}}  
        f(x)/^{\left(\frac{\alpha+(\alpha+1)\kappa}{\alpha+\kappa}\right)}\ln_{\kappa}f(x)^{-\frac{\alpha}{\alpha+\kappa}} \mathrm{d}x
    \end{equation} 
and  two constraints, the normalization and  the independent-equals moment, 
\begin{align}
        \int_{x\in\mathcal{X}}  
        x^\alpha f(x)/^{\left(\frac{\alpha+(\alpha+1)\kappa}{\alpha+\kappa}\right)} \mathrm{d}x 
        =\sigma^\alpha;
        \end{align}
then, the maximum coupled entropy distribution is the coupled stretched exponential of definition \ref{def_CSED}, 
\begin{align}
f_\kappa(x;\alpha) &= \frac{1}{Z_\kappa(\sigma,\alpha)}\left(1+\kappa\frac{x^\alpha}{\alpha\sigma^\alpha}\right)_+^{-\frac{\alpha+\kappa}{\alpha\kappa}}
=\exp_{\kappa}
\left(\frac{x^\alpha}{\alpha \sigma^\alpha}
\oplus_{\kappa} \ln_{\kappa} Z^\frac{\alpha}{\alpha+\kappa}
\right)^{-\frac{\alpha+\kappa}{\alpha}},
\end{align}
with a domain of concavity $\kappa > -\frac{\alpha}{\alpha+1}$. The Lagrangian multipliers are
\begin{align}
    \begin{matrix}
        \text{Normalization:} & \lambda_0 =  
        - \frac{\alpha} {(\alpha+\kappa) 
    \int_{x\in\mathcal{X}} 
    f(x)^{\iota(\kappa,\alpha)} d(x)}; \\
    \text{Independent Equals Moment:} & \lambda_1 =
    \frac{Z_\kappa^\frac{\alpha\kappa}{\alpha+\kappa} }
     {\alpha \sigma^\alpha}.
    \end{matrix}
\end{align}
\end{lemma}

\begin{proof}
    Given the coupled entropy function and the two constraints, the Lagrangian function is:
\begin{align}
\mathcal{L} &= \int_{x\in\mathcal{X}} 
f(x)/^{\iota(\kappa,\alpha)} 
\ln_{\kappa} f(x)^{-\frac{\alpha}{\alpha+\kappa}} dx 
+ \lambda_0 \left(1 - \int_{x\in\mathcal{X}} f(x)  dx \right) \nonumber \\
&\quad+ \lambda_1 \left( \sigma^\alpha -
\int_{x\in\mathcal{X}} x^\alpha f(x)/^{\iota(\kappa,\alpha)} dx
\right),
\end{align}
where $\lambda_0$ and $\lambda_1$ are the Lagrangian multipliers for the normalization and independent-equals alpha-moment, respectively. For maximization, the derivative must be zero, \(\delta \mathcal{L} / \delta f = 0\). The derivative is in terms of a particular value $y=x'$ such that the integrals only have a non-zero derivative at $y$:
\begin{enumerate}[left=\parindent]
\item Entropy Derivative
            \begin{align}
            &\text{Separate the integrals within the coupled entropy:} \nonumber \\
            \frac{1}{\alpha}H(\textbf{X};\kappa,\alpha)^\alpha=\frac{A}{B}; 
            A &= \int_{x\in\mathcal{X}} 
            f(x)^{\iota(\kappa,\alpha)}
\ln_{\kappa} f(x)^{-\frac{\alpha}{\alpha+\kappa}} dx; \ 
            B = \int_{x\in\mathcal{X}} 
            f(x)^{\iota(\kappa,\alpha)} dx 
            \nonumber \\
                \frac{\delta H(f(y);\kappa,\alpha)^\alpha/\alpha}{\delta f(y)}
                &= \frac{\frac{\alpha+(1+\alpha)\kappa}{\alpha+\kappa}}{B}
                f(y)^\frac{\alpha\kappa}{\alpha+\kappa} 
                \ln_{\kappa} f(y)^{-\frac{\alpha}{\alpha+\kappa}}
                -\frac{\alpha}{B(\alpha+\kappa)} \nonumber \\
    &- \frac{\alpha+(1+\alpha)\kappa}{\alpha+\kappa}
    f(y)^\frac{\alpha\kappa}{\alpha+\kappa} \frac{A}{B^2}.
            \end{align} 
\item The normalization derivative is $-\lambda_0$.
\item Independent equals derivative
            \begin{align}
                &\frac{\delta}{\delta f(y)}  \lambda_1 \left( \sigma^\alpha 
                - \int_{x\in\mathcal{X}} x^\alpha
                f (x)/^{\iota(\kappa,\alpha)}dx \right) \nonumber \\
                &= -\lambda_1 y^\alpha \frac{\alpha+(1+\alpha)\kappa}{\alpha+\kappa}
                \frac{f(y)^\frac{\alpha\kappa}{\alpha+\kappa}}{B} 
                \nonumber \\
                &+\lambda_1 \frac{\alpha+(1+\alpha)\kappa}{\alpha+\kappa}
                f(y)^\frac{\alpha\kappa}{\alpha+\kappa}
                \frac{\int_{x\in\mathcal{X}} x^\alpha
                f(x)^{\iota(\kappa,\alpha)}dx}
                {B^2} \nonumber \\
                &= -\lambda_1\frac{\alpha+(1+\alpha)\kappa}{(\alpha+\kappa)B}
                f(y)^\frac{\alpha\kappa}{\alpha+\kappa}
                \left( y^\alpha - \sigma^\alpha\right)
            \end{align}
\end{enumerate}
Combining terms, the Lagrangian derivative is:
\begin{align}\label{equ_Lder}
    \delta \frac{\mathcal{L}}{\delta f} 
    &= \frac{\alpha+(1+\alpha)\kappa}{(\alpha+\kappa)B}
    f(y)^\frac{\alpha\kappa}{\alpha+\kappa}
    \left(\ln_{\kappa} f(y)^{-\frac{\alpha}{\alpha+\kappa}}
    -\frac{A}{B}-\lambda_1 (y^\alpha - \sigma^\alpha)\right) \nonumber \\
     &-\frac{\alpha}{B(\alpha+\kappa)} - \lambda_0 \nonumber \\
    &=0.
\end{align}
Multiplying by $-(1+\kappa/\alpha)\kappa B$ and solving for $f(y)$ gives:
\begin{align}
    &(1+(1+1/\alpha)\kappa)
    \left(1 + \frac{A}{B}\kappa
    +\lambda_1 \kappa(y^\alpha - \sigma^\alpha) \right)
    f(y)^\frac{\alpha\kappa}{\alpha+\kappa} \nonumber \\
    &= (1+(1+1/\alpha)\kappa)
     -\kappa - \lambda_0 (1+\kappa/\alpha)\kappa B \nonumber \\
     f(y)^{-\frac{\alpha\kappa}{\alpha+\kappa}} 
     &= \frac{(\alpha+(1+\alpha)\kappa)
    \left(1 + \frac{A}{B}\kappa
    +\lambda_1 \kappa(y^\alpha - \sigma^\alpha) \right)}
     {(\alpha+(1+\alpha)\kappa)
     -\alpha\kappa - \lambda_0 (\alpha+\kappa)\kappa B} \nonumber \\
     &= \frac{(\alpha+(1+\alpha)\kappa)
     \left(1 + \frac{A}{B}\kappa
    -\lambda_1 \kappa \sigma^\alpha\right)
    \left(1 + \frac{\lambda_1 \kappa y^\alpha}
    {1 + \frac{A}{B}\kappa 
    -\lambda_1 \kappa \sigma^\alpha} 
    \right)}
     {(\alpha+(1+\alpha)\kappa)
     -\alpha\kappa - \lambda_0 (\alpha+\kappa)\kappa B}
\end{align}
The right-hand side has the form $a(1+b y^\alpha),$ confirming that the coupled exponential distribution is the maximizing distribution. The Langragian multipliers are determined from the constraints, which specify that $a=Z^\frac{\alpha\kappa}{\alpha+\kappa}$ and $b=\frac{\kappa}{\alpha\sigma^\alpha}.$ Given $\frac{1}{\alpha} H(f(y);\kappa,\alpha)^\alpha = \frac{A}{B} = \frac{1}{\alpha}+\ln_\frac{\alpha\kappa}{\alpha+\kappa}Z_\kappa$ from Lemma \ref{lem_unique} the moment constraint multiplier is determined from the equation:
\begin{align}
    \frac{1}{\sigma^\alpha} 
    &= \frac{\lambda_1 \alpha}
    {1 + \left(\frac{1}{\alpha}+\ln_\frac{\alpha\kappa}{\alpha+\kappa}
    Z_\kappa\right)\kappa 
    -\lambda_1 \kappa \sigma^\alpha} \nonumber \\
    &= \frac{\lambda_1 \alpha}
    {(1+\kappa/\alpha) Z_\kappa^\frac{\alpha\kappa}{\alpha+\kappa} 
    -\lambda_1 \kappa \sigma^\alpha} \nonumber \\
     \frac{(1+\kappa/\alpha) Z_\kappa^\frac{\alpha\kappa}{\alpha+\kappa}  }{\sigma^\alpha} 
     &= \lambda_1 \alpha(1 + \kappa/\alpha ) \nonumber \\
     \lambda_1 &= \frac{
     Z_\kappa^\frac{\alpha\kappa}{\alpha+\kappa}}
     {\alpha \sigma^\alpha}.
\end{align}
The normalization constraint is determined by the expression for the normalization:
\begin{align} \label{equ_L0}
    Z_\kappa^{-\frac{\alpha\kappa}{\alpha+\kappa} } &= \frac{(1+(1+1/\alpha)\kappa)
     -\kappa - \lambda_0 (1+\kappa/\alpha)\kappa B}
     {(1+(1+1/\alpha)\kappa)
     \left((1+\kappa/\alpha) Z_\kappa^\frac{\kappa}{1+\kappa/\alpha} 
    -Z_\kappa^\frac{\kappa}{1+\kappa/\alpha} \kappa/\alpha \right)} \nonumber 
    \\  
    1 &= \frac{1+\kappa/\alpha - \lambda_0(1+\kappa/\alpha)\kappa B}
    {1+\kappa +\kappa/\alpha}
    \nonumber \\
    \lambda_0 &=  
     \frac{\alpha} {(\alpha+\kappa) 
    \int_{x\in\mathcal{X}} 
    f(x)^{\iota(\kappa,\alpha)} dx}
\end{align}
Thus, completing the proof of maximization.

The test of concavity consists of proving that the second variational of coupled entropy kernel is negative over the gradients specified by the constraints. This will be done by solving for the second variation of the entropy, proving that a portion of the solution is zero given the extremum solution and the constraints, then solving for the negative domain. Starting with the closed-form of the trace coupled entropy function:
    \begin{align}
        \frac{1}{\alpha}H(f;\kappa,\alpha)^\alpha=\frac{1}{\kappa}
        \left(\frac{1}{Z_\iota(f)}-1\right),\ 
        Z_\iota(f) \equiv \int_Xf^{\iota(\kappa,\alpha)}(x)\mathrm{d}x,
    \end{align}
    and the first variation, $\delta f(x)$ with $\delta Z_\iota(f) = \iota(\kappa,\alpha,1)\int_Xf^{\iota(\kappa,\alpha)-1}(x) \delta f(x) \mathrm{d}x$ is:
    \begin{align}
        \delta \frac{1}{\alpha} H_\kappa(f,\alpha)^\alpha
        &=-\frac{\delta Z(f)}{\kappa Z(f)^2},
    \end{align}
    and the second variation is:
    \begin{align}
        \delta^2 H(f;\kappa,\alpha)^\alpha
        &=\frac{2(\delta Z_\iota(f))^2}{\kappa Z_\iota(f)^3}
        -\frac{\delta^2Z_\iota(f)}{\kappa Z_\iota(f)^2}
    \end{align}
    The first term on the right can be proven to be zero, since $\delta Z_\iota(f)=0$. The extrema solution of the Lagrangian has the structure $f(x)^{1-\iota(\kappa,\alpha)}=C+Dx^\alpha$, which when multiplied by $f(x)^{\iota(\kappa,\alpha)-1}$ is $1=Cf(x)^{\iota(\kappa,\alpha)-1} + Dx^\alpha f(x)^{\iota(\kappa,\alpha)-1}$. The constraints define the admissible variations $\delta f(x)$,
    \begin{align}
        \int_X \delta f(x) \mathrm{d}x = 0, \ 
        \int_X (x^\alpha - \sigma^\alpha)f(x)^{\iota(\kappa,\alpha)-1}\delta f(x) \mathrm{d}x = 0.
    \end{align}
    Multiply the extrema equation by $\delta f(x)$ and integrating gives
    \begin{align}
        \int_X \delta f(x) \mathrm{d}x = 
        C \int_X f(x)^{\iota(\kappa,\alpha)-1} \delta f(x) \mathrm{d}x +
        D \int_X x^\alpha f(x)^{\iota(\kappa,\alpha)-1} \delta f(x) \mathrm{d}x.
    \end{align}
    The left-hand side equals zero due to the normalization constraint. Substituting for the rightmost expression using the moment constraint gives
    \begin{align}
        0 &= 
        C \int_X f(x)^{\iota(\kappa,\alpha)-1} \delta f(x) \mathrm{d}x +
        D \sigma^\alpha \int_X f(x)^{\iota(\kappa,\alpha)-1} \delta f(x) \mathrm{d}x
        \nonumber\\
        &= (C+D\sigma^\alpha) \int_X f(x)^{\iota(\kappa,\alpha)-1} \delta f(x) \mathrm{d}x.
    \end{align}
    Since the coefficient is the density at the scale, $(C+D\sigma^\alpha) = f(\sigma)^{1-\iota}$, it must be non-zero and positive. Therefore, the first variant of $Z_\iota$ is zero, since the integral is zero,
    \begin{align}
        \delta Z_\iota = \iota(\kappa,\alpha) 
        \int_X f(x)^{\iota(\kappa,\alpha)-1} \delta f(x) \mathrm{d}x = 0.
    \end{align}
    Therefore, the requirement for the second variation of the entropy to be negative reduces to
    \begin{align}
        \delta^2 \frac{1}{\alpha}H(f;\kappa,\alpha)^\alpha
        &=
        -\frac{\delta^2Z_\iota(f)}{\kappa Z_\iota(f)^2}
        \leq 0
        \nonumber\\
        &= -\frac{\iota(\kappa,\alpha)(\iota(\kappa,\alpha)-1)
        \int_Xf^{\iota(\kappa,\alpha,1)-2}(x) (\delta f(x))^2 dx}
        {\kappa \left(\int_Xf^{\iota(\kappa,\alpha,1)}(x)dx\right)^2}.
    \end{align}
    The integral $Z_\iota$ in the denominator is positive, since it's a power of the density. Likewise, the integral in the numerator is positive, since it is a power of the density multiplied by the square of the variation, both of which are positive. Thus, the requirement for negativity is
    \begin{align}
        - \frac{\left(1+\frac{\alpha\kappa}{\alpha+\kappa}\right)
        \frac{\alpha\kappa}{\alpha+\kappa}}
        {\kappa}
        \leq 0
        \nonumber\\
        1+\frac{\alpha\kappa}{\alpha+\kappa} \geq 0, \text{ and }
        \frac{\alpha}{\alpha+\kappa} \geq 0
        \nonumber\\
        \kappa \geq -\frac{\alpha}{\alpha+1} > - \frac{\alpha+1}{\alpha}.
    \end{align}
\end{proof}
\begin{remark}
    The Lagrangian multiplier, $\lambda_1,$ is the natural parameter of information geometry, dual to the moment parameter  $\sigma^\alpha.$ The presence in the Lagrangian multipliers of terms involving the integral of the distribution, $B=\int_\mathbf{X}f(\mathbf{x})^{\iota(\kappa,\alpha)} d\mathbf{x},$ and $Z_\kappa^\frac{\alpha\kappa}{\alpha+\kappa}$ is related to the nonlinear deformation intrinsic to the coupled exponential family. As noted in the previous section, use of the partition function to normalize the distribution, and the coupled sum to bring this factor within the coupled exponential function, contrasts with the definitions used by Amari to define the $q$-exponential family. Thus, a complete study of the information geometry of the coupled exponential family is merited.
\end{remark}
\begin{remark}
    In the continuous case with $\alpha=1,$ the normalization of the coupled exponential distribution is $Z_\kappa = \sigma.$ In this case, the Lagrangian multipliers are $\lambda_0 = \sigma^\frac{\kappa}{1+\kappa}$ and $\lambda_1 = \sigma^{\frac{\kappa}{1+\kappa}-1} = \sigma^\frac{-1}{1+\kappa}$ and $\frac{\lambda_0}{\lambda_1}=\sigma$. 
\end{remark}

\begin{theorem}[Unique solution for entropy of the coupled stretched exponentials]
    \label{lem_unique}
    Given a $d$-dimensional random variable distributed as a coupled stretched exponential, $\mathbf{X} \sim f_\kappa(\mathbf{x}; \boldsymbol{\mu},\boldsymbol{\Sigma},\alpha,d)$, as defined in equation \ref{def_CSED}, the Boltzmann-Gibbs-Shannon entropy expressed in terms of the energetic, configurational and nonlinear degrees of freedom is
    \begin{align} \label{equ_BGS_CSED}
H_{BGS}(\textbf{X}) = \frac{1}{\alpha}
\begin{cases}
\displaystyle
\overbrace{d}^{\text{energetic}}
+
\overbrace{\alpha \ln Z(\kappa,\alpha,d)}^
{\text{configurational}}
+
\overbrace{
\frac{\alpha+d\kappa}{\kappa}
\left[
\psi\left(\frac{d}{\alpha}+\frac{1}{\kappa}\right)
-
\psi\left(\frac{1}{\kappa}\right)
\right]-d
}^{\text{nonlinear DOF}},
& \kappa>0,\\[1.5em]
\displaystyle
d+ \alpha \ln Z(0,\alpha,d),
& \kappa=0,\\[1.5em]
\displaystyle
d+
\alpha \ln Z(\kappa,\alpha,d)
+
\frac{\alpha+d\kappa}{\kappa}
\left[
\psi\left(1-\frac{d}{\alpha}-\frac{1}{\kappa}\right)
-
\psi\left(1-\frac{1}{\kappa}\right)
-d \right],
& -\frac{\alpha}{d}<\kappa<0.
\end{cases}
    \end{align}
    
The calibrated entropy matched to the nonlinearity and dimension $(\kappa,\alpha,d)$ of the CSED reduces the nonlinear DOF to zero and modifies the configurational DOF to a generalized logarithm 
\begin{align}
     H (\mathbf{X};\kappa,\alpha,d)^\alpha  &= \underbrace{d}_{\text{energetic}} 
    + \  \underbrace{\ln_\frac{\kappa}{\alpha+d\kappa}
    Z_\kappa(\sigma,\alpha,d)^\alpha}_{\text{configurational DOF}}
    \ + \underbrace{0}_{\text{nonlinear}}
\end{align} 
The expression on the right is the value of the density of the distribution at the radius $r=\left(((\mathbf{x}-\boldsymbol{\mu})^{\circ \sfrac{\alpha}{2} \top} 
(\boldsymbol{\Sigma}^{\circ \sfrac{\alpha}{2}})^{-1}(\mathbf{x}-\boldsymbol{\mu})^{\circ \sfrac{\alpha}{2}}\right)^\frac{1}{\alpha}$ which for example occurs when $x_i=\mu_i+\sigma_{ii}$ for all $i.$ In the limit of infinite coupling, the entropy is proportional to the partition function,
\begin{align}
    \lim\limits_{\kappa\rightarrow\infty} H(\mathbf{X};\kappa,\alpha,d)= 
    d^\frac{1}{\alpha} Z_\kappa(\sigma,\alpha,d)^\frac{1}{d}.
\end{align}
\end{theorem}
\begin{proof}
Without loss of generality, $\boldsymbol{\mu}$ is set to zero. The CSED in radial form is:
    \begin{align}
        f_\kappa(\mathbf{x}) &= 
        \frac{1}{Z_\kappa}\exp_{\kappa}
        \left(\frac{r^\alpha}{\alpha}\right)
        ^\frac{\alpha+d\kappa}{\alpha}
        \nonumber \\
        r^\alpha &= (\mathbf{x}^{\circ \sfrac{\alpha}{2} \top} 
(\boldsymbol{\Sigma}^{\circ \sfrac{\alpha}{2}})^{-1}\mathbf{x}^{\circ \sfrac{\alpha}{2}},
    \end{align}
where $\circ$ indicates that the power is applied to each element of the vector or matrix.

The Boltzmann-Gibbs-Shannon entropy for the coupled stretched exponential distribution is given by

\begin{align}
    H_{BGS}(\mathbf{X}) &= -\int_{\mathbf{X}} 
    f_\kappa(\mathbf{x}) \ln f_\kappa(\mathbf{x}) 
    \mathrm{d}\mathbf{x}
    \nonumber\\
   &= -\int_{\mathbf{X}} 
    \frac{1}{Z_\kappa}
        \left(1+\kappa\frac{r^\alpha}{\alpha}\right)
        ^\frac{\alpha+d\kappa}{\alpha\kappa}
        \ln\left[\frac{1}{Z_\kappa}
        \left(1+\kappa\frac{r^\alpha}{\alpha}\right)
        ^\frac{\alpha+d\kappa}{\alpha\kappa}\right] \mathrm{d}\mathbf{x}
        \nonumber\\ 
    &= \ln Z_\kappa
    + \frac{\alpha+d\kappa}{\alpha\kappa}
    \,
    \mathbb{E}\!\left[
    \ln\left(1+\kappa\frac{r^\alpha}{\alpha}\right)
    \right].
\end{align}
For the expectation, let $u = \frac{r^\alpha}{\alpha}$, and the derivative is radial, angular coordinates is $\mathrm{d}\mathbf{x}
=
C_{\Sigma,\alpha}\,
u^{\frac{d}{\alpha}-1}
\,\mathrm{d}u\,\mathrm{d}\Omega$. Since the angular coordinate drops out, the expectation is
\begin{align}
    \mathbb{E}\!\left[\ln(1+\kappa u)\right]
    &=
    \frac{
        \displaystyle
        \int_0^\infty
        u^{\frac{d}{\alpha}-1}
        (1+\kappa u)^{-\frac{\alpha+d\kappa}{\alpha\kappa}}
        \ln(1+\kappa u)
        \,\mathrm{d}u
    }{
        \displaystyle
        \int_0^\infty
        u^{\frac{d}{\alpha}-1}
        (1+\kappa u)^{-\frac{\alpha+d\kappa}{\alpha\kappa}}
        \,\mathrm{d}u
    }.
\end{align}
For $\kappa>0$,
\begin{align}
    \mathbb{E}\!\left[\ln(1+\kappa u)\right]
    &=
    \frac{B\!\left(\frac{d}{\alpha},\frac{1}{\kappa}\right)
    \left[
    \psi\!\left(\frac{d}{\alpha}+\frac{1}{\kappa}\right)
    - \psi\!\left(\frac{1}{\kappa}\right)
    \right]
    }{B\!\left(\frac{d}{\alpha},\frac{1}{\kappa}\right)}
    = \psi\!\left(\frac{d}{\alpha}+\frac{1}{\kappa}\right)
    - \psi\!\left(\frac{1}{\kappa}\right),
\end{align}
where $\psi$ is the digamma function. For, $\kappa=0$ this term is zero, and for $-\frac{\alpha}{\kappa}<\kappa<0$, it is $\psi\!\left(1-\frac{d}{\alpha}-\frac{1}{\kappa}\right) - \psi\!\left(1-\frac{1}{\kappa}\right)$.
Once simplified, the BGS entropy of CSED is given by equation \ref{equ_BGS_CSED}, completing the first part of the proof.

The coupled entropy for $\kappa \neq 0$ is
     \begin{align}
        \frac{1}{\alpha}H(f_\kappa(\mathbf{x});\kappa,\alpha,d)^\alpha&=\bigintssss_\Omega 
        f(\mathbf{x})/^{\iota(\kappa,\alpha)}
        \ln_{\kappa}f(\mathbf{x})^{-\frac{\alpha}{\alpha+d\kappa}}\mathrm{d}\mathbf{x}.
    \end{align}
    The coupled logarithm of the density reduces to
    \begin{align}
        \ln_{\kappa}f(\mathbf{x})^{-\frac{\alpha}{\alpha+d\kappa}}
        &= \ln_{\kappa}\left[
        \frac{1}{Z_\kappa}\exp_{\kappa}
        \left(\frac{r^\alpha}{\alpha}\right)
        \right]^{\frac{1+d\kappa/\alpha}{1+d\kappa/\alpha}} \nonumber \\
        \ln_{\kappa}Z_\kappa^\frac{1}{1+d_\alpha\kappa}
        &\oplus_{\kappa}\frac{r^\alpha}{\alpha}
    \end{align}
    The independent-equals expectation of $u=r^\alpha/\alpha$ is $d_\alpha$, as derived below. In radial-angular coordinates, the volume element is $\mathrm{d}\mathbf{x}= C_{\mathbf{{\Sigma}},\alpha}u^{d_\alpha -1}\mathrm{d}u\quad\mathrm{d}\Omega$, where $C$ is a constant that depends on $\mathbf{\Sigma}$ and $\alpha$. Since the distribution does not depend on the angular variable and its integral appears in both the numerator and denominator, this factor cancels out of the expectation. Thus the independent-equals $\alpha$-moment is equal to 
    \begin{align}
        \mathbb{E}_{f/^{\iota(\kappa,\alpha)}}[u]
        &= \frac{\int_{u\in U} u 
        \exp_{\kappa}
        \left(u \oplus_{\kappa} 
        \ln_{\kappa} Z_\kappa^\frac{\alpha}{\alpha+d\kappa}\right)
        ^{-\frac{\alpha+d\kappa}{\alpha}}
        C_{\mathbf{{\Sigma}},\alpha}u^{\frac{d}{\alpha}-1}\mathrm{d}u}
        {\int_{u\in U}  
        \exp_{\kappa}
        \left(u \oplus_{\kappa} 
        \ln_{\kappa} Z_\kappa^\frac{\alpha}{\alpha+d\kappa}\right)
        ^{-\frac{\alpha+d\kappa}{\alpha}}
        C_{\mathbf{{\Sigma}},\alpha}u^{\frac{d}{\alpha}-1} \mathrm{d}u}.
    \end{align}
    Removing the shared constants from the numerator and denominator reduces the expectation to 
    \begin{align}
        \mathbb{E}_{f/^{\iota(\kappa,\alpha)}}[u]
        &= \frac{\int_{u\in U} u^\frac{d}{\alpha} 
        \exp_{\kappa}
        \left(u\right)^{-\frac{\alpha+d\kappa}{\alpha}}
        \mathrm{d}u}
        {\int_{u\in U} u^{\frac{d}{\alpha}-1}
        \exp_{\kappa}
        \left(u\right)^{-\frac{\alpha+d\kappa}{\alpha}}
        \mathrm{d}u}.
    \end{align}
Expanding the coupled exponential and substituting $t=u\kappa$ clarifies the integrals to be a Beta function
\begin{align}
        \mathbb{E}_{f/^{\iota(\kappa,\alpha)}}[u]
        &= \frac{\int_0^\infty t^\frac{d}{\alpha} 
        \left(1+t\right)^{-\frac{\alpha+d\kappa}{\alpha\kappa}}
        \mathrm{d}t}
        {\frac{\kappa}{\alpha}\int_0^\infty t^{\frac{d}{\alpha}-1}
        \left(1+t\right)^{-\frac{\alpha+d\kappa}{\alpha\kappa}}
        \mathrm{d}t}
        = \frac{B\left[\frac{\alpha}{d}+1,
        \frac{1}{\kappa}\right]}
        {\frac{\kappa}{\alpha} B\left[\frac{\alpha}{d},
        \frac{1}{\kappa}+1\right]}
        =d,
    \end{align}
where the simplification uses the properties $B[a+1,b]=\frac{a}{a+b}B[a,b]$ and $B[a,b+1]=\frac{b}{a+b}B[a,b]$.

Since the coupled sum is affine, the  entropy of the matching CSED is
\begin{align}
   \frac{1}{\alpha}H(f_\kappa(\mathbf{x});\kappa,\alpha,d)^\alpha
    &=\frac{\mathbb{E}_{f/^{\iota(\kappa,\alpha)}}[\frac{r^\alpha}{\alpha}]}{\alpha}
    \oplus_{\kappa} \ln_{\kappa}Z_\kappa(\sigma,\alpha,d)
    ^\frac{\alpha}{\alpha+d\kappa} \nonumber \\
    &=\frac{d}{\alpha} + \ln_\frac{\alpha\kappa}{\alpha+d\kappa}Z_\kappa(\sigma,\alpha,d).
\end{align}
The entropy is $H(f_\kappa(\mathbf{x});\kappa,\alpha,d) = \left( d + \ln_\frac{\kappa}{\alpha+d\kappa} Z_\kappa(\sigma,\alpha,d)^\alpha\right)^\frac{1}{\alpha}$. In the infinity limit, the entropy converges to a constant times the partition function,
\begin{align}
    \lim\limits_{\kappa\rightarrow\infty} H(f_\kappa(\mathbf{x});\kappa,\alpha,d) 
    &= \left( d + d\left(Z_\kappa(\sigma,\alpha,d)
    ^\frac{\alpha}{d}-1\right)\right)^\frac{1}{\alpha} \nonumber \\
    &= d^\frac{1}{\alpha} Z_\kappa(\sigma,\alpha,d)^\frac{1}{d}.
\end{align}
\end{proof}

\begin{theorem}[Comprehensive universality classes of entropy]
\label{thproof_univ}
The  entropy  $H(X;\kappa,\alpha,d)$, fulfills the full range of the Hanel-Thurner $(\tilde{c}=1-c,d)$  universality classes with 
\begin{align}
(\tilde{c},d) = \left\{
    \begin{matrix}
        \left(\frac{\kappa}{\alpha+d\kappa},
        \frac{1}{\alpha}\right), & \kappa>0 \\
        \left(0,\frac{1}{\alpha}\right), & \kappa=0 \\
         \left(0,0\right), & -\frac{\alpha}{\alpha+d}< \kappa <0.
    \end{matrix}
    \right.
\end{align}
\end{theorem}
\begin{proof}

Using the radial variable $r$, the equiprobability $p$, and considering $d$ dimensions, the entropy \ref{equ_CE}, is the function $g(p)= \left( p\ln_{\kappa} p^{-\frac{\alpha}{\alpha+d\kappa}}\right)^\frac{1}{\alpha}$$=\frac{p}{\kappa}\left(p^{-\frac{\alpha\kappa}{\alpha+d\kappa}}-1\right)$. From equation \ref{equ_scalefunction}, with $p=\sfrac{1}{W}$, the scaling is evaluated based on $F_{G,g}=G\left(Wg(W^{-1})\right)=\left(\ln_\kappa 
W^{\frac{\alpha\kappa}{\alpha+d \kappa}}\right)^\frac{1}{\alpha}$. 

The power-law, $\tilde{c},$ scaling properties of the coupled entropy are: 
\begin{align}
    \lim_{W\rightarrow \infty}\left(
    \frac{\ln_{\kappa} (\lambda W)
    ^{\frac{1}{1+d_\alpha\kappa}}}
    { \ln_{\kappa} W^{\frac{1}{1+d_\alpha\kappa}}}
    \right)^\frac{1}{\alpha} 
    &= \left\{ \begin{matrix}
    \lim_{W\rightarrow \infty}\left(
    \frac{(\lambda W)
    ^{\frac{\kappa}{1+d_\alpha\kappa}}-1}
    { W^{\frac{\kappa}{1+d_\alpha\kappa}}-1}
    \right)^\frac{1}{\alpha} & \kappa \neq 0 \nonumber \\
    \lim_{W\rightarrow \infty}\left(
    \frac{\ln \lambda + \ln W}
    { \ln W}
    \right)^\frac{1}{\alpha} & \kappa=0
    \end{matrix} \right.
    \\
    &= \left\{ \begin{matrix}
        \lambda^\frac{\kappa}{\alpha+d\kappa} & \kappa>0 \\
        1 & -\frac{\alpha}{\alpha+d} < \kappa \leq 0.
    \end{matrix} \right.
\end{align}
Thus, the power-law scaling is
\begin{align}
    \tilde{c}
     &= \left\{ \begin{matrix}
        \frac{\kappa}{\alpha+d\kappa} & \kappa>0 \\
        0 & -\frac{\alpha}{\alpha+d} < \kappa \leq 0.
    \end{matrix} \right.
\end{align}

Analysis of the $d$ scaling is based on a substitution of $\lambda \rightarrow W^a$ and compensation of the $\tilde{c}$ scaling by multiplication of $ W^{-a\tilde{c}}$. The $d$ scaling, utilizing L'Hôpital's Rule, is determined for $\kappa\geq0$ by:
\begin{align}
    \lim_{W\rightarrow \infty}\left[
    \left(
    \frac{\ln_{\kappa} ( W^{1+a})
    ^{\frac{\alpha}{\alpha+d\kappa}}}
    { \ln_{\kappa} W^{\frac{\alpha}{\alpha+d\kappa}}}
    \right)^\frac{1}{\alpha}W^{-a\tilde{c}}
    \right] 
    &= \lim_{W\rightarrow \infty}\left[\left(
    \frac{W^\frac{\alpha(1+a)\kappa}{\alpha+d\kappa}
    -1}
    {W^\frac{\alpha\kappa}{\alpha+d\kappa}-1}
    \right)^\frac{1}{\alpha}W^{-a\tilde{c}}\right] \nonumber \\
    &= \lim_{W\rightarrow \infty}\left[\left(
    (1+a)\frac{W^\frac{(1+a-d/\alpha-a)\alpha\kappa-\alpha}{\alpha+d\kappa}}
    {W^\frac{(\alpha-d)\kappa-\alpha}{\alpha+d\kappa}}
    \right)^\frac{1}{\alpha}\right] =(1+a)^\frac{1}{\alpha}, 
\end{align}
and for $\kappa=0$ by,
\begin{align}
    \lim_{W\rightarrow \infty}\left[\left(
    \frac{\ln W^{1+a}}
    { \ln W}
    \right)^\frac{1}{\alpha}W^0\right] = (1+a)^\frac{1}{\alpha},
\end{align}
and for $-\frac{\alpha}{\alpha+d} < \kappa < 0$ by,
\begin{align}
    \lim_{W\rightarrow 0}\left[
    \left(
    \frac{\ln_{\kappa} ( W^{1+a})
    ^{\frac{\alpha}{\alpha+d\kappa}}}
    { \ln_{\kappa} W^{\frac{\alpha}{\alpha+d\kappa}}}
    \right)^\frac{1}{\alpha} W^0
    \right] 
    &= \lim_{W\rightarrow 0}\left[\left(
    \frac{W^\frac{(1+a)\alpha\kappa}{\alpha+d\kappa}
    -1}
    {W^\frac{\alpha\kappa}{\alpha+d\kappa}-1}
    \right)^\frac{1}{\alpha}\right] = 1. 
\end{align}
Thus, the stretch scaling is 
\begin{align}
    d
     = \left\{ \begin{matrix}
        \frac{1}{\alpha} & \kappa \geq 0 \\
        0 & -\frac{\alpha}{\alpha+d\kappa} < \kappa < 0.
    \end{matrix} \right.
\end{align}
\end{proof}
Assuming just the first three Shannon-Khinchin axioms, presented in \nameref{subsec_axioms}, Hanel and Thurner \cite{hanelComprehensiveClassificationComplex2011, Hanel2011a,hanelComprehensiveClassificationComplex2011, hanelHowMultiplicityDetermines2014, amigoBriefReviewGeneralized2018, korbelClassificationComplexSystems2018} used two different scaling properties to define a universal family of (c,d)-entropies. The (c,1) entropies are maximized by shape-scale distributions. The (1,d)-entropies are maximized by the stretched exponential distributions. Hanel-Thurner derived a family of $(c,d)$-entropies and maximizing distributions; however, although the calibrated entropy fulfills the Hanel-Thurner extensivity classes, it is not a member of $(c,d)$-entropy. 

Tempesta \cite{tempestaShannonKhinchinFormulation2016, tempestaUniversalityClassesInformationTheoretic2020} defined another universal class of entropies that are composable with a power series of nonlinear terms. Thus, just as Euclid's fifth axiom can be generalized to define a variety of Riemannian geometries, the fourth Shannon-Khinchin axiom generalizes to define non-additive but composable entropies. In fact, the information geometry of the resulting metrics induces non-Euclidean geometries \cite{amariGeometryQExponentialFamily2011}. Tempesta's universal entropy, $H^U(\textbf{p})\equiv \sum_i p_i G\left(\ln p_i^{-1}\right)$, where the power series composition function is $G(t) \equiv \sum_{k=0}^\infty a_k \frac{t^{k+1}}{k+1}$, uniquely satisfies the following composition axiom. Given two statistically independent systems with probabilities, $p_i^I$ and $p_j^{II}$, the joint entropy of the two systems is: 
\begin{align}
    H^U(\textbf{p}^I,\textbf{p}^{II}) = G\left(G^{-1}(H^U(\textbf{p}^I)) + G^{-1}(H^U(\textbf{p}^{II}))\right)
\end{align}
where $G^{-1}(s)$ is the compositional inverse of G(t).

\subsection*{Proposed entropy functions}\label{subsec_Entropies}

This subsection will show a variety of functions that have been proposed to measure entropy. To understand their structure, various forms of the functions are provided. Of particular interest is deriving the form in which a generalized mean aggregates the probabilities and then is transformed into an entropy measure via a generalized logarithm. This will assist in explaining the relationship between the entropies. Also shown with each function is its translation onto the coupled stretched exponential density.

The expressions will utilize the relationship between the exponent of the $q-$exponential distributions and the coupled algebra via $q=\iota(\kappa,\alpha,d)=1+\frac{\alpha\kappa}{\alpha+d\kappa}=\frac{\alpha+(\alpha+d)\kappa}{\alpha+d\kappa}.$ Within the generalized mean expressions, $f(\mathbf{x})$ or 
\begin{align}
    f(\mathbf{x})/^{\iota(\kappa,\alpha,d)}
    =\frac{f(\mathbf{x})^{\iota(\kappa,\alpha,d)}}
    {\int_{\mathbf{x\in\mathcal{X}}} 
    f(\mathbf{x})^{\iota(\kappa,\alpha,d)}}\mathrm{d}x
\end{align} 
will be separated out to clarify the difference between the probability serving as a weight and the function of the probabilities being averaged. The generalized mean associated with each function can be shown by applying the inverse of its respective generalized logarithm. 

\textbf{Rényi}
\begin{align}
    H_q^R(\mathbf{X}) &= -\ln\left(\int_{\mathbf{x\in\mathcal{X}}} f(\mathbf{x})f(\mathbf{x})^{q-1}d\mathbf{x}\right )^\frac{1}{q-1} \nonumber \\
&=-\ln\left(\int_{\mathbf{x\in\mathcal{X}}}f(\mathbf{x})f(\mathbf{x})^
\frac{\alpha\kappa}{\alpha+d\kappa}d\mathbf{x}\right)^
{\frac{\alpha+d\kappa}{\alpha\kappa}}
\end{align}

\textbf{Tsallis}
\begin{align} \label{equ_TEnt}
     H_q^T(\mathbf{X}) &=\frac{1}{1-q}
     \left(\int_{
     \mathbf{x\in\mathcal{X}}}
     f(\mathbf{x})f(\mathbf{x})^{q-1}\mathrm{d}\mathbf{x}
     -1\right) 
     \nonumber \\
     H_\kappa^{T}(\mathbf{X};\alpha,d)
     &= {-\frac{\alpha+d\kappa}{\alpha\kappa}}
     \left(\int_{\mathbf{x\in\mathcal{X}}}
     f(\mathbf{x})f(\mathbf{x})^
     \frac{\alpha\kappa}{\alpha+d\kappa}
     \mathrm{d}\mathbf{x}
     -1\right) \nonumber \\
     &=- \int_{\mathbf{x\in\mathcal{X}}} f(\mathbf{x})
     \ln_\frac{\alpha\kappa}{\alpha+d\kappa} f(\mathbf{x})
     \mathrm{d}\mathbf{x} \nonumber \\
      &=-\ln_\frac{\alpha\kappa}{\alpha+d\kappa}
      \left(\int_{\mathbf{x}\in\mathcal{X}} 
    f(\mathbf{x}) f(\mathbf{x})
    ^\frac{\alpha\kappa}{\alpha+d\kappa}
    \mathrm{d}\mathbf{x}\right)^\frac{\alpha+d\kappa}{\alpha\kappa}
\end{align}

\begin{align}
    \exp_\frac{\kappa}{1+d_\alpha\kappa}(- H_\kappa^{T}(\mathbf{X}))
    &= \exp_\frac{\kappa}{1+d_\alpha\kappa}\left( \int_{\mathbf{x\in\mathcal{X}}} f(\mathbf{x})
     \ln_\frac{\kappa}{1+d_\alpha\kappa} f(\mathbf{x})
     dF(\mathbf{x})
     \right) \nonumber \\
     &= \left(1+\int_{\mathbf{x\in\mathcal{X}}} 
     \left(f(\mathbf{x})f(\mathbf{x})^\frac{\kappa}{1+d_\alpha\kappa}
     -f(\mathbf{x})\right)dF(\mathbf{x})\right)^\frac{1+d_\alpha\kappa}{\kappa} \nonumber \\
     &= \left(\sum_{\mathbf{x\in\mathcal{X}}} 
     f(\mathbf{x})f(\mathbf{x})^\frac{\kappa}{1+d_\alpha\kappa}
     dF(\mathbf{x})\right)^\frac{1+d_\alpha\kappa}{\kappa}
\end{align}

\textbf{Normalized Tsallis}
\begin{align}
     H_q^{NT}(\mathbf{X}) &=\frac{1}{1-q}
     \left(1-
     \frac{1}{\int_{\mathbf{x\in\mathcal{X}}}
     f(\mathbf{x})^q \ \mathrm{d}\mathbf{x}}\right) \nonumber \\
     &=\frac{1}{1-q}\int_{\mathbf{x}\in\mathcal{X}}
     f(\mathbf{x})/^q\left(1-f(\mathbf{x})^{-q}\right)
     \mathrm{d}\mathbf{x} 
     \nonumber \\
     H_\kappa^{NT}(\mathbf{X};\alpha,d)
     &={\frac{\alpha+d\kappa}{\alpha\kappa}}
     \int_{\mathbf{x\in\mathcal{X}}}
     f(\mathbf{x})/^{\iota(\kappa,\alpha)}
     \left(f(\mathbf{x})^{-\iota(\kappa,\alpha)}-1\right)
     \mathrm{d}\mathbf{x} 
     \nonumber \\
     &=\int_{\mathbf{x\in\mathcal{X}}}
     f(\mathbf{x})/^{\iota(\kappa,\alpha)}
     \ln_{\iota(\kappa,\alpha)} f(\mathbf{x})^{-1}
     \mathrm{d}\mathbf{x} 
     \nonumber \\
     &=-\ln_\frac{\alpha\kappa}{\alpha+d\kappa}
     \left(\int_{\mathbf{x}\in\mathcal{X}} 
    f(\mathbf{x})/^{\iota(\kappa,\alpha)} f(\mathbf{x})
    ^\frac{\alpha\kappa}{\alpha+d\kappa}
    \mathrm{d}\mathbf{x}\right)^\frac{\alpha+d\kappa}{\alpha\kappa}
\end{align}

\begin{align}
      &\left(\exp_\frac{\alpha\kappa}{\alpha+d\kappa}
      \left(H_\kappa^{NT}(\mathbf{X})\right)\right)^
      {-\frac{\alpha\kappa}{\alpha+(\alpha+d)\kappa}} 
      \nonumber \\
      &= \left(\exp_\frac{\alpha\kappa}{\alpha+d\kappa}
      \left(
      \int_{\mathbf{x\in\mathcal{X}}}
     f(\mathbf{x})/^{\iota(\kappa,\alpha)}
     \ln_\frac{\alpha\kappa}{\alpha+d\kappa} f(\mathbf{x})^
     {-\frac{\alpha+(\alpha+d)\kappa}{\alpha\kappa}} 
      \mathrm{d}\mathbf{x}\right)\right)^
      {-\frac{\alpha\kappa}{\alpha+(\alpha+d)\kappa}} \nonumber \\
      &= \left(1+
      \int_{\mathbf{x\in\mathcal{X}}}
     \left(f(\mathbf{x})/^{\iota(\kappa,\alpha)}
      f(\mathbf{x})^
     {-\iota(\kappa,\alpha)} 
     - f(\mathbf{x})/^{\iota(\kappa,\alpha)} \right)
      \mathrm{d}\mathbf{x}\right)
      ^{-\frac{\alpha+d\kappa}{\alpha+(\alpha+d)\kappa}} 
      \nonumber \\
      &= \left(
      \int_{\mathbf{x\in\mathcal{X}}}
      f(\mathbf{x})^
     {-\iota(\kappa,\alpha)}
      \mathrm{d}\mathbf{x}\right)
      ^{-\frac{\alpha+d\kappa}{\alpha+(\alpha+d)\kappa}} 
\end{align}
\textbf{Hanel-Thurner c,d-Entropy}
\begin{align}
    H_{c,d,r}^{HT} &= \int_{\mathbf{x\in\mathcal{X}}}  
    \left(rA^{-d}\exp\left( A\right) 
    \Gamma \left(1 + d , A - c \ln x\right) - rcx 
    \right) \mathrm{d}\mathbf{x}\\
    A &= \frac{cdr}{ 1-(1-c)r}.
\end{align}
Substituting $c= 1 - \tilde{c} = 1 - \frac{\kappa}{\alpha+d\kappa}$ and $d=\frac{1}{\alpha}$, while noting that in this parameterization $d$ is the dimensions not the scaling, the HT-entropy is 
\begin{align}
    H^{HT}(\kappa,\alpha,d,r) &=
    \int_{\mathbf{x\in\mathcal{X}}}  
    \left( rA^{-d}\exp\left( A\right) 
    \Gamma \left(1 + d , A - \left(1 - \frac{\kappa}{\alpha+d\kappa}\right) \ln x\right) 
    - r\left(1 - \frac{\kappa}{\alpha+d\kappa}\right) x 
    \right) \mathrm{d}\mathbf{x}\\
    A &= \frac{(\alpha - (d-1)\kappa) r}{ \alpha(\alpha - (d-r)\kappa)} 
\end{align}
\textbf{Coupled Entropy}
\begin{align}
    \frac{1}{\alpha}H(\mathbf{X};\kappa,\alpha,d)^\alpha 
    &=\frac{1}{\kappa}
    \int_{\mathbf{x\in\mathcal{X}}} 
    f(\mathbf{x})/^{\iota(\kappa,\alpha)}
    \left(f(\mathbf{x})^{-\frac{\alpha\kappa}{\alpha+d\kappa}}
    -1\right)\mathrm{d}\mathbf{x} 
    \nonumber \\
    &=\int_{\mathbf{x\in\mathcal{X}}} 
    f(\mathbf{x})/^{\iota(\kappa,\alpha)}
    \ln_{\kappa} f(\mathbf{x})^{-\frac{\alpha}{\alpha+d\kappa}}
    \mathrm{d}\mathbf{x} 
    \nonumber \\
    &= \ln_\kappa\left(\left(
    \int_{\mathbf{x\in\mathcal{X}}}
    f(\mathbf{x})/^{\iota(\kappa,\alpha)}
    f(\mathbf{x})^{-\frac{\alpha\kappa}{\alpha+d\kappa}}
    \mathrm{d}\mathbf{x}\right)^{-\frac{1}{\kappa}}\right)
    \nonumber\\
    & = \ln_\kappa\left(\left(
    \int_{\mathbf{x\in\mathcal{X}}}
    f(\mathbf{x})^{\iota(\kappa,\alpha)}
    \mathrm{d}\mathbf{x}
    \right)^{-\frac{1}{\kappa}}\right)
\end{align}
\begin{align}
    &\exp_{\kappa}^{-(1+d\kappa/\alpha)} 
    [\frac{1}{\alpha}H(\mathbf{X};\kappa,\alpha,d)^\alpha]  \nonumber \\
    &=\exp_{\kappa}^
    {-(1+d\kappa/\alpha)}\left(
    \int_{\mathbf{x\in\mathcal{X}}} 
    f(\mathbf{x})/^{\iota(\kappa,\alpha,d)}
    \ln_{\kappa} f(\mathbf{x})^{-\frac{\alpha}{\alpha+d\kappa}}
    \mathrm{d}\mathbf{x}\right) 
    \nonumber \\
    &= \left(1+\frac{\kappa}{\kappa}\int_{\mathbf{x\in\mathcal{X}}}
    \left(f(\mathbf{x})/^{\iota(\kappa,\alpha,d)}
    f(\mathbf{x})^{-\frac{\alpha\kappa}{\alpha+d\kappa}}
    -f(\mathbf{x})/^{\iota(\kappa,\alpha,d)}\right)
    \mathrm{d}\mathbf{x}\right)^{-\frac{\alpha+d\kappa}{\alpha\kappa}} 
    \nonumber \\
    &= \left(\int_{\mathbf{x\in\mathcal{X}}}
    f(\mathbf{x})^{\iota(\kappa,\alpha,d)}
    d(\mathbf{x})\right)^{\frac{1}{\iota(\kappa,\alpha,d)-1}}
\end{align}

 Table \ref{tab_entDOF} shows the computation of degrees of freedom for the entropy functions given the equipartition calibration function, $W(N) = \exp_{\kappa}\left(N\right)^{1+\kappa}$, with $\alpha=d=1$. The parameters of a function (with the exception of BGS and Rényi) are matched to ensure linear extensivity. Both the Tsallis and Normalized Tsallis functions retain dependence on the nonlinear coupling. Although the $(c,d)$-function includes a free parameter $r$, no solution completely removes the dependence on the nonlinear coupling.
 
\begin{table}[ht]
    \centering
     \caption{\textbf{Measurement of the degrees of freedom of the calibrated growth function.} Computation of degrees of freedom by proposed functions for entropy given the calibrated growth function $W(N) = \exp_{\kappa}\left(N\right)^{1+\kappa}$  with $d=\alpha=1$. For the $(c,d)$-function, the asymptotic degrees of freedom are measured to be $(1+\kappa N)(1+\kappa -\kappa r)(\frac{r}{1+\kappa-\kappa r}
    +\ln \left(1+\kappa N\right)^\frac{1}{\kappa}+1)
    - \frac{1+\kappa}{r}$, where $r$ is a free parameter. }
    \label{tab_entDOF}
    \begin{tabular}{ccc}\toprule
        \textbf{Entropy Function}& \textbf{Degrees of Freedom}& \textbf{Description}\\\midrule
        \textbf{\textit{BGS}}& $\frac{1+\kappa}{\kappa}\ln(1+\kappa N)$& The nonlinearity contributes to the degrees of freedom\\
        \textbf{\textit{Rényi}} & $\frac{1+\kappa}{\kappa}\ln(1+\kappa N)$& The equipartition equals BGS\\
        \textbf{\textit{Tsallis }} $2-q = 1-\frac{\kappa}{1+\kappa}$& $\frac{1+\kappa}{\kappa}(1+\kappa N)$& Linear in N but also function of nonlinearity\\
        \textbf{\textit{Normalized Tsallis }} $q = 1 + \frac{\kappa}{1+\kappa}$& $(1+\kappa)N$& Linear in N but also function of nonlinearity\\
 \textbf{\textit{Hanel-Thurner }} $c=1-\frac{\kappa}{1+\kappa},\ d = 1$ & see caption&There is no value of r that reduces the DOF to N\\
        \textbf{Coupled}&  $N$& DOF independent of nonlinearity\\ \bottomrule
    \end{tabular}
   
\end{table}

In Section \nameref{subsec_unique}, the entropies for the centered coupled exponential distribution $(\alpha=1,d=1,\mu=0)$ are mapped onto the distribution, which provides a visual comparison of the generalized entropies. The density value equals $\exp_\kappa^{-(1+\kappa)}(\frac{1}{\alpha}H^\alpha).$ The variable is then determined by solving for $x,$ given a coupled exponential, 
\begin{align}
    x&=\sigma (\frac{1}{\alpha}H^\alpha\ominus_\kappa \ln_\kappa \sigma^\frac{1}{1+\kappa})
    =\sigma\frac{\frac{1}{\alpha}H^\alpha-\ln_\kappa \sigma^\frac{1}{1+\kappa}}{1+\kappa 
    \ln_\kappa \sigma^\frac{1}{1+\kappa}}.
\end{align} 
Each of these entropy functions is maximized by the coupled stretched exponential distribution.  The basic properties of these functions are illustrated by their measure of the coupled exponential distribution $(\alpha=1,d=1)$, shown in Figure \ref{fig_entropies} and Table \ref{tab_entropies}. The limit properties, as the coupling goes to infinity, provide an important distinction in the function. The Shannon and Rényi entropies approach infinity, as shown in Figure \ref{fig_entropies_Shannon}. The Rényi entropy does reduce the dependence on the coupling from a linear to a logarithmic relationship.  There are three versions of the Tsallis entropy that have been explored as potential measures of entropy. In the figure, the $q$-exponential definition is used which sets the scale to $(1+\kappa)q$. The Tsallis entropy converges to a constant, raising questions about its ability to be a metric. The normalized Tsallis entropy goes to infinity faster than the Shannon entropy, reflective of its instability. The Tsallis entropy with the conjugate value of $2-q=1-\frac{\kappa}{1+\kappa}$ with the standard moment as a constraint has been proposed because its equipartition scaling properties can be extensive; however, as shown in Figure \ref{fig_entropies_Tsallis}, its measure of the maximum distribution diverges. The coupled entropy with its calibration approaches the scale $\sigma$.

The relationship between the coupled, the normalized Tsallis, and Tsallis entropy functions is:
\begin{equation}
        \frac{1}{\alpha}H(\mathbf{p};\kappa,\alpha,d)^\alpha = \frac{\alpha g_{\kappa}^{\text{NT}}(\mathbf{p})}{\alpha+d\kappa}
        = \frac{\alpha g_{\kappa}^{\text{T}}
        (\mathbf{p})}{(\alpha+d\kappa)\sum_jp_j^{1+\frac{\alpha\kappa}{\alpha+d\kappa}}}.
    \end{equation}
Given the clarity provided by separating the linear (the informational scale, $\sigma$) and the nonlinear (the coupling, $\kappa$), the relationship with the proposed Tsallis functions is straightforward. Nevertheless, despite decades of research, the Tsallis methods, which include an interdependence between the linear and nonlinear sources of uncertainty, did not result in a fully consistent measure of entropy. The normalized Tsallis function was a significant improvement over the Tsallis function; however, at the time of its investigation, its instability short-circuited further investigation. Use of the incorrect scale was fundamental to the Tsallis entropy having inadequate properties.

  \begin{figure}[ht]
      \centering
    \subfloat[]{
      \includegraphics[width=0.45\linewidth,page=1]{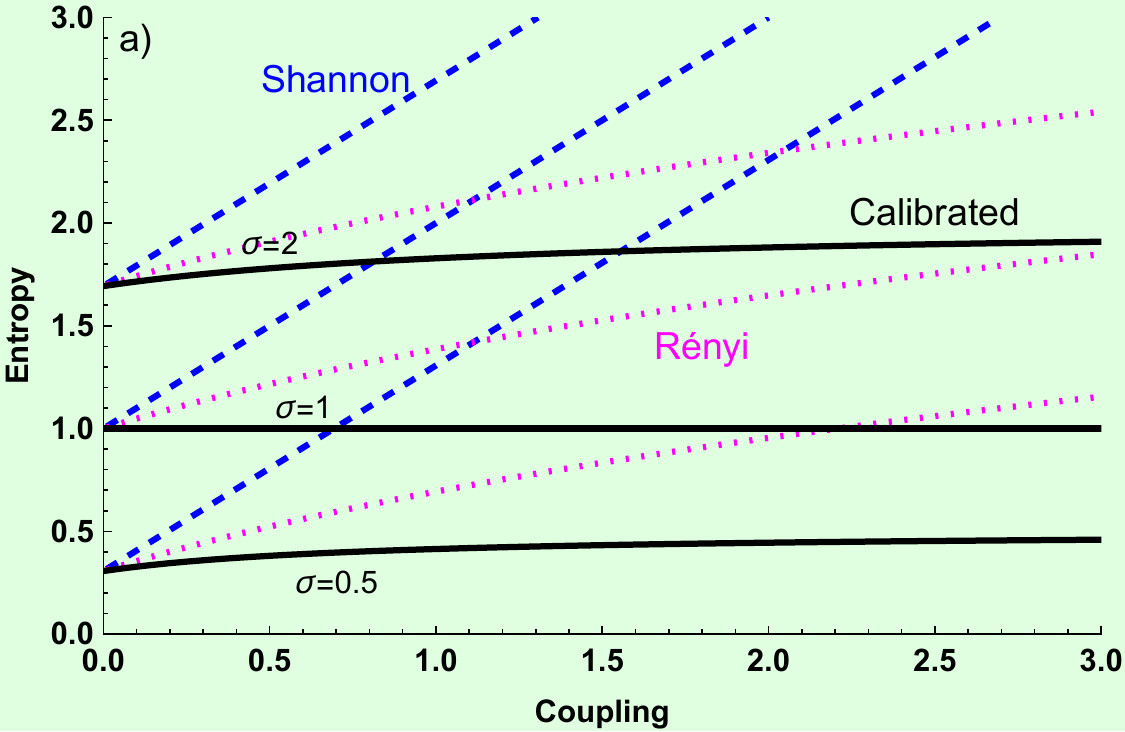} 
      \label{fig_entropies_Shannon}
      }
      \hfill
    \subfloat[]{
    \includegraphics[width=0.45\linewidth,page=1]{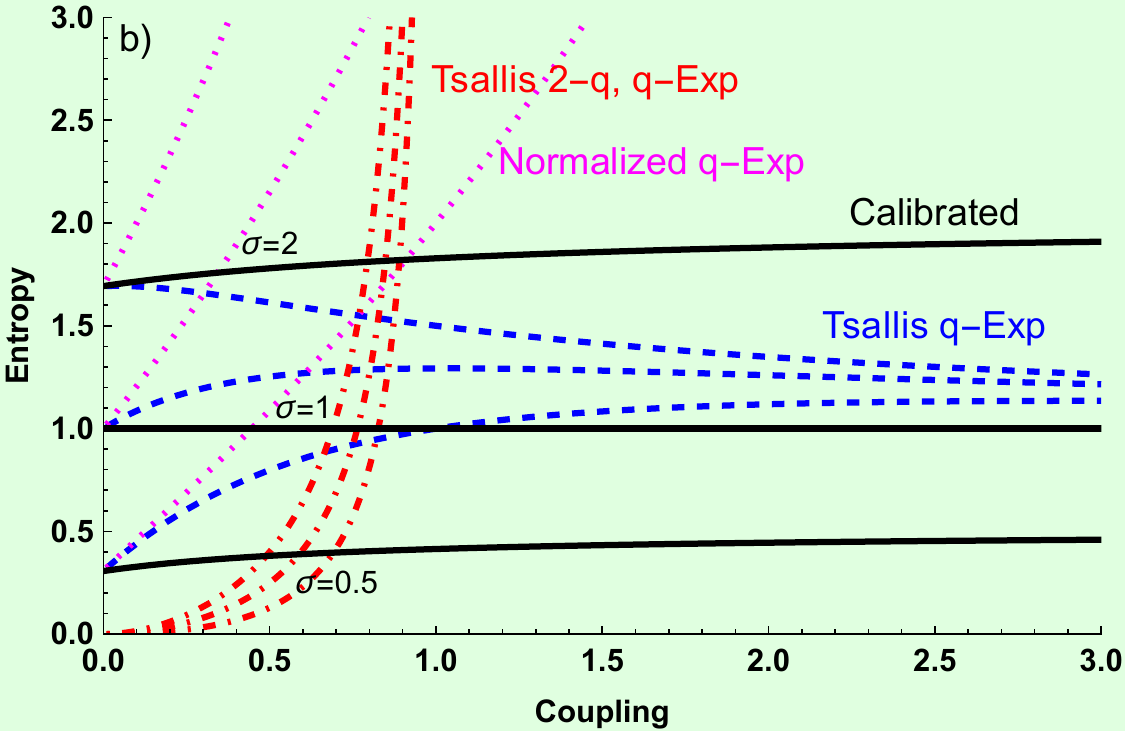} 
    \label{fig_entropies_Tsallis}
      }
      \caption{\textbf{Potential entropy functions of the coupled exponential distribution.} Various entropy functions are compared with the calibrated coupled entropy as a function of the coupling with scale values of $\sigma = 0.5,1.0, 2.0$. The coupled entropy (calibrated) converges to $\sigma$ as the coupling goes to infinity. a) The coupled entropy (black, solid) is compared with Shannon entropy (blue, dash) and Rényi entropy (magenta, dot). Shannon is linear, and Rényi is logarithmic with the coupling. b) The coupled entropy is compared with the Tsallis  (blue, dash),  Normalized Tsallis (magenta, dot), and Tsallis $(2-q)$ entropies. The Tsallis entropies are computed for a $q$-exponential distribution with scale equal to $(1+\kappa)\sigma$.  The Tsallis function has an inverse relationship with the scale, converging to 1, the Normalized Tsallis multiplies the scale term by the coupling, and the Tsallis $(2-q)$ diverges at $\kappa=1$.} \label{fig_entropies}
  \end{figure}
\begin{table}[ht]
    \centering
    \caption{\textbf{The computation of potential entropy functions of the coupled exponential distribution.} The BGS entropy has three components of DOF: energetic (1), configurational ($\ln \sigma$), and nonlinear ($\kappa$). The required solution for a calibrated entropy is to remove the nonlinear DOF. Instead, the nonlinearity contributes to a generalized measure of the configurational DOF. The Rényi, Tsallis, and Normalized Tsallis entropies do not fulfill this requirement.}
    \label{tab_entropies}
    \begin{tabular}{ccc}\toprule
        \textbf{Entropy Function}& \textbf{of $\kappa$-exp or $q$-exp dist.}& \textbf{Description}\\\midrule
        \textbf{\textit{BGS}}& $1+\ln \sigma+\kappa$& Logarithmic in Scale; Linear in Shape \\
        \textbf{\textit{Rényi}} & $\ln\sigma+(1+\frac{1}{\kappa})\ln(1+\kappa)$ & Logarithmic in Scale \& Shape\\
        \textbf{\textit{Tsallis}} & $1-\frac{1}{1+\kappa}\ln_\frac{\kappa}{1+\kappa} {((1+\kappa)\sigma)^{-1}}$& Inverse in Scale \& Shape\\
        \textbf{\textit{Normalized Tsallis}} & $1+(1+\kappa)\ln_\frac{\kappa}{1+\kappa}{(1+\kappa)\sigma}+\kappa$& Multiplicative in Scale \& Shape\\
        \textbf{\textit{Calibrated: Coupled}}&  $1+\ln_\frac{\kappa}{1+\kappa} \sigma$& Generalized Log of Scale\\ \bottomrule
    \end{tabular}
   
\end{table}

\section*{Acknowledgements}

I wish to thank colleagues Amenah Al-Najafi, Igor Oliveira, William Thistleton, and Calden Wloka, whose co-development of the coupled variational inference methods clarified the unique importance of the coupled entropy. Special thanks go to Ugur Tirnakli and Bruce Boghosian for the invitation to the 2025 Nonextensive Statistical Physics Workshop to present the research on the coupled entropy. Conversations with Rudolf Hanel were helpful in clarifying the structure of the entropy scaling classes.

\section*{Funding}
There are no funding grants to acknowledge for this study.

\section*{Conflict of Interest}
There are no conflicts of interest to disclose.

\section*{Author Contributions}
This is a single-author article; all work was performed by the author.

\section*{Code Availability} \label{subsec_avail}

A GitHub repository regarding Nonlinear Statistical Coupling methods is maintained at \cite{nelsonNonlinearStatisticalCoupling2025}. Within the repository is a folder ``nsc-mathematica''. The file ``Coupled Entropy NSP 2025.nb'' and its PDF copy ``Coupled Entropy NSP 2025Nov20.pdf'' contain the computations for the graphics in this paper. The file calls functions from ``Coupled Functions.nb'', which has a PDF copy ``Coupled Functions 2025Nov20.pdf''. Some additional computations are completed in ``Compare Uncertainty Functions v4.nb'', with PDF copy ``Compare Uncertainty Functions v4 2025Nov20.pdf''. 

\section*{Data Availability}
No data were used in this work.

\section*{Declarations}

Mathematica and DeepSeek were used as analytical aids in the research and development process; however, the manuscript was written without AI aids. An editorial review was completed using Gemini, Claude, and Grammarly. Wikipedia was used for background information, though references refer to primary sources.

\newpage
\bibliographystyle{sn-nature}
\bibliography{MICS_BibTex}

\end{document}